\documentclass[11pt, journal, onecolumn]{IEEEtran}
\usepackage{amsmath,stmaryrd,acronym,amssymb,amsthm,dsfont,graphicx, mathtools,paralist,color,cite}
\usepackage{subfiles}
\usepackage[inline]{enumitem}
\acrodef{AEP}{Asymptotic Equipartition Property}
\acrodef{AoA}{Angle of Arrival}
\acrodef{AWGN}{Additive White Gaussian Noise}
\acrodef{BER}{Bit-Error-Rate}
\acrodef{BEC}{Binary Erasure Channel}
\acrodef{BPSK}{Binary Phase-Shift Keying}
\acrodef{BSC}{Binary Symmetric Channel}
\acrodef{CDF}[CDF]{Cumulative Distribution Function}
\acrodef{CLT}[CLT]{Central Limit Theorem}
\acrodef{CSI}[CSI]{Channel State Information}
\acrodef{DMC}[DMC]{Discrete Memoryless Channel}
\acrodef{DMS}[DMS]{Discrete Memoryless Source}
\acrodef{iid}[i.i.d.]{independent and identically distributed}
\acrodef{LPD}[LPD]{Low Probability of Detection}
\acrodef{LDPC}[LDPC]{Low-Density Parity-Check}
\acrodef{MAC}[MAC]{multiple-access channel}
\acrodef{MLC}[MLC]{Multilevel Coding}
\acrodef{MIMO}[MIMO]{Multiple-Input Multiple-Output}
\acrodef{MISO}{Multiple-Input Single-Output}
\acrodef{MSD}[MSD]{Multistage Decoding}
\acrodef{PDF}[PDF]{Probability Distribution Function}
\acrodef{PMF}[PMF]{Probability Mass Function}
\acrodef{PPM}[PPM]{Pulse Position Modulation}
\acrodef{PSD}{Power Spectral Density}
\acrodef{QPSK}{Quadrature Phase-Shift Keying}
\acrodef{SIMO}{Single-Input Multiple-Output}
\acrodef{SNR}{Signal-to-Noise Ratio}
\acrodef{wrt}[w.r.t.]{with respect to}
\acrodef{WSS}{Wide Sense Stationary}

\newcommand{\E}[2][]{{\mathbb{E}_{#1}}{\left(#2\right)}}       
\renewcommand{\P}[2][]{{\mathbb{P}_{#1}}{\left(#2\right)}}
       
\newcommand{\D}[2]{{{\mathbb{D}}\!\left({#1\Vert#2}\right)}}
\newcommand{\avgD}[2]{{{\mathbb{D}}\!\left({#1\Vert#2}\right)}}
\newcommand{\V}[1]{{{\mathbb{V}}\!\left(#1\right)}}

\newcommand{\card}[1]{\ensuremath{\left|{#1}\right|}}           
\newcommand{\abs}[1]{\ensuremath{\left|#1\right|}}              
\newcommand{\eqdef}{\ensuremath{\triangleq}}                    
\newcommand{\intseq}[2]{\ensuremath{\llbracket{#1},{#2}\rrbracket}}  

\renewcommand{\leq}{\leqslant}
\renewcommand{\geq}{\geqslant}

\newcommand{\proddist}{%
  \mathchoice{\raisebox{1pt}{$\displaystyle\otimes$}}
             {\raisebox{1pt}{$\otimes$}}
             {\raisebox{0.5pt}{\scalebox{0.7}{$\scriptstyle\otimes$}}}
             {\raisebox{0.4pt}{\scalebox{0.6}{$\scriptscriptstyle\otimes$}}}}
\newcommand{\pn}{{\proddist n}}

\DeclareMathAlphabet{\eurm}{U}{eur}{m}{n}
\DeclareMathAlphabet{\mathbsf}{OT1}{cmss}{bx}{n}
\DeclareMathAlphabet{\mathssf}{OT1}{cmss}{m}{sl}
\DeclareMathAlphabet{\mathcsf}{OT1}{cmss}{sbc}{n}

\DeclareSymbolFont{bsfletters}{OT1}{cmss}{bx}{n}  
\DeclareSymbolFont{ssfletters}{OT1}{cmss}{m}{n}
\DeclareMathSymbol{\bsfGamma}{0}{bsfletters}{'000}
\DeclareMathSymbol{\ssfGamma}{0}{ssfletters}{'000}
\DeclareMathSymbol{\bsfDelta}{0}{bsfletters}{'001}
\DeclareMathSymbol{\ssfDelta}{0}{ssfletters}{'001}
\DeclareMathSymbol{\bsfTheta}{0}{bsfletters}{'002}
\DeclareMathSymbol{\ssfTheta}{0}{ssfletters}{'002}
\DeclareMathSymbol{\bsfLambda}{0}{bsfletters}{'003}
\DeclareMathSymbol{\ssfLambda}{0}{ssfletters}{'003}
\DeclareMathSymbol{\bsfXi}{0}{bsfletters}{'004}
\DeclareMathSymbol{\ssfXi}{0}{ssfletters}{'004}
\DeclareMathSymbol{\bsfPi}{0}{bsfletters}{'005}
\DeclareMathSymbol{\ssfPi}{0}{ssfletters}{'005}
\DeclareMathSymbol{\bsfSigma}{0}{bsfletters}{'006}
\DeclareMathSymbol{\ssfSigma}{0}{ssfletters}{'006}
\DeclareMathSymbol{\bsfUpsilon}{0}{bsfletters}{'007}
\DeclareMathSymbol{\ssfUpsilon}{0}{ssfletters}{'007}
\DeclareMathSymbol{\bsfPhi}{0}{bsfletters}{'010}
\DeclareMathSymbol{\ssfPhi}{0}{ssfletters}{'010}
\DeclareMathSymbol{\bsfPsi}{0}{bsfletters}{'011}
\DeclareMathSymbol{\ssfPsi}{0}{ssfletters}{'011}
\DeclareMathSymbol{\bsfOmega}{0}{bsfletters}{'012}
\DeclareMathSymbol{\ssfOmega}{0}{ssfletters}{'012}

\newcommand{\svbk}{\mathbf{k}}

\newcommand{\svbu}{{\mathbf{u}}}

\newcommand{\svbv}{{\mathbf{v}}}

\newcommand{\svbw}{{\mathbf{w}}}

\newcommand{\calA}{{\mathcal{A}}}

\newcommand{\calC}{{\mathcal{C}}}

\newcommand{\calE}{{\mathcal{E}}}

\newcommand{\calG}{{\mathcal{G}}}

\newcommand{\calO}{{\mathcal{O}}}

\newcommand{\calT}{{\mathcal{T}}}
\newcommand{\calS}{{\mathcal{S}}}

\newcommand{\calX}{{\mathcal{X}}}
\newcommand{\calY}{{\mathcal{Y}}}

\newcommand{\calZ}{{\mathcal{Z}}}

\renewcommand{\proddist}{%
  \mathchoice{\raisebox{1pt}{$\displaystyle\otimes$}}
  {\raisebox{1pt}{$\otimes$}}
  {\raisebox{0.5pt}{\scalebox{0.7}{$\scriptstyle\otimes$}}}
  {\raisebox{0.4pt}{\scalebox{0.6}{$\scriptscriptstyle\otimes$}}}}
\renewcommand{\pn}{{\proddist n}}

\newtheorem{remark}{Remark}
\newtheorem{proposition}{Proposition}
\newtheorem{corollary}{Corollary}
\newtheorem{lemma}{Lemma}
\newtheorem{definition}{Definition}
\newtheorem{theorem}{Theorem}

\acrodef{RS}[RS]{Reed-Solomon}
\acrodef{MLC}[MLC]{multilevel coding}
\acrodef{MSD}[MSD]{multistage decoding}
\acrodef{LPD}[LPD]{low probability of detection}
\acrodef{MIMO}[MIMO]{multiple-input multiple-output}
\acrodef{PPM}[PPM]{pulse position modulation}
\acrodef{BSC}[BSC]{binary symmetric channel}
\acrodef{BAC}[BAC]{binary asymmetric channel}
\acrodef{DMC}[DMC]{discrete memoryless channel}
\acrodef{BI-DMC}[BI-DMC]{binary-input discrete memoryless channel}

\begin{document}
\title{Multilevel-Coded Pulse-Position Modulation for Covert Communications over Binary-Input Discrete Memoryless Channels}
\author{Ishaque Ashar Kadampot, Mehrdad Tahmasbi, Matthieu R Bloch%
	\thanks{An earlier version of this paper was presented at the 2018 IEEE International Symposium on Information Theory~\cite{Kadampot2018}. This work was supported in part by the National Science Foundation under award 1527074.}}
\date{Draft v1.3}
\maketitle

\begin{abstract}
  We develop a low-complexity coding scheme to achieve covert communications over \acp{BI-DMC}. We circumvent the impossibility of covert communication with linear codes by introducing non-linearity through the use of \ac{PPM} and \ac{MLC}. We show that the MLC-PPM scheme exhibits many appealing properties; in particular, the channel at a given index level remains stationary as the number of level increases, which allows one to use families of channel capacity- and channel resolvability-achieving codes to concretely instantiate the covert communication scheme.
\end{abstract}

\section{Introduction}
\label{sec:introduction}

While signal processing techniques, such as spread-spectrum and \ac{MIMO}, have been widely used to enforce \ac{LPD} over communication channels in the presence of adversarial eavesdroppers, the information-theoretic limits of covert communications have only been recently characterized. A unique feature of covert communications is the existence of a ``square root law,'' similar to that of steganography, which constrains the scaling of the number of reliable and covert communication bits to $\mathcal{O}(\sqrt{n})$ if the blocklength is $n$~\cite{Bash2013,Bash2015}. The exact constants in front of $\sqrt{n}$ have also been characterized for classical and quantum point-to-point channels~\cite{Wang2016,Bloch2016,Sheikholeslami2016,Wang2016c,Tahmasbi2016,Tahmasbi}, and can be interpreted as the ``covert capacity'' of these channels.

One of the key insights obtained from the characterization of the information-theoretic limit is the need to use codebooks comprised of ``low-weight'' codewords. Specifically, if symbol~$0$ denotes the ``innocent symbol'' corresponding to the absence of communication, codewords should contain a fraction of non-zero symbols on the order of $\calO(\tfrac{1}{\sqrt{n}})$ as $n$ goes to infinity. Unfortunately, despite information-theoretic results showing the existence of low-complexity covert communications codes using a concatenated scheme~\cite{Zhang2016}, explicit code constructions have remained elusive. The main challenge in designing such codes is precisely the need for low weights, which is not satisfied with standard linear codes. We actually show a stronger negative result in this paper, namely that no sequence of linear codes with growing dimension can approach any fraction of the covert capacity. Consequently, one must introduce some non-linearity in the coding scheme. A previous attempt considered the use of a binary polarization-based scheme with non-linearity introduced through stochastic encoding~\cite{Freche2017}; however, the analysis of the speed of polarization in~\cite{Freche2017} only leads to ``low-weight'' codewords containing a fraction $\calO(\tfrac{1}{n^\gamma})$ of non-zero symbols with $0<\gamma\ll \frac{1}{2}$ as $n$ goes to infinity. The resulting large codeword weight was mitigated by changing the model and allowing asynchronism and uncertainty in the time of transmission~\cite{Bash2016,Arumugam2016}, at the cost of a significant increase in the effective blocklength. This result could perhaps be improved by using non-binary polar codes.

We follow here a different approach and develop a covert communication scheme using \acf{MLC} together with \acf{PPM} for \acfp{BI-DMC}. As already shown in~\cite{Bloch2017} with random non-binary codes, the use of \ac{PPM} allows one to partly handle covertness through modulation, which potentially simplifies the coding. The use of \ac{MLC} allows us to exploit low-complexity binary codes that are channel capacity- and channel resolvability-achieving, such as polar codes~\cite{Chou2018}, and circumvents the challenges associated with non-binary code design that was left open in~\cite{Bloch2017}. We presented the initial results of this research in \cite{Kadampot2018}, in which we proved the existence of optimal codes for covert communication using MLC-PPM for \acp{BSC}. In this paper, we show the existence of optimal covert communication codes for \acp{BI-DMC}. In~\cite{Kadampot2018}, we also used dithering with public common randomness to simplify the analysis, whereas here we exploit the properties of the channel perceived by each level of \ac{MLC} instead.

The rest of the paper is organized as follows. In Section~\ref{sec:notation}, we briefly introduce the common notations used throughout the paper. In Section~\ref{sec:imposs-covert-comm}, we prove that it is impossible to achieve covert communication over a \ac{BI-DMC} using linear codes without introducing some non-linearity after the encoding operation. In Section~\ref{sec:mult-coding-with}, we introduce \ac{MLC}-\ac{PPM} scheme for covert communication over \ac{BI-DMC} and provide an information-theoretic proof of the achievability of covert capacity using the scheme. In Section~\ref{sec:to_prac_codes}, we discuss the design of practical codes by exploiting various properties of the resulting equivalent channel derived from the use of \ac{MLC}-\ac{PPM} scheme and provide an coding scheme using polar codes. In Section~\ref{sec:conclusion}, we conclude the paper  with a discussion on our main contributions.

\section{Notation}
\label{sec:notation}
We briefly introduce the notation used throughout the paper. We denote random variables by uppercase letters, e.g., $X$, their realizations by lowercase, e.g., $x$, and vectors of random variables and their realizations by their corresponding boldface letters, e.g., $\mathbf{X}$ and $\mathbf{x}$, respectively. We denote the transition probability of a \ac{BI-DMC} channel with input $X\in\mathcal{X} = \{0,1\}$ and output $Y\in\mathcal{Y}$ by $W_{Y|X}$. Let $\intseq{a}{b}$ be the set of integers between $\lfloor a\rfloor$ and $\lceil b\rceil$. When using as a subscript or a superscript, we denote $\intseq{a}{b}$ by $a{:}b$. For $q\in\mathbb{N},~m = 2^q,$ and $i\in\intseq{1}{m}$, we define a \ac{PPM} symbol $\widetilde{x}_i$ of order $m$ as the binary vector of length $m$ such that the $i$-th component is one and all other components are zero. The use of \ac{PPM} modulation over a \ac{BI-DMC} defines a ``super channel'' with transition probability $\widetilde{W}_{\widetilde{Y}\vert \widetilde{X}} \eqdef W_{Y|X}^{\proddist m}$ whose input alphabet is the set  $\widetilde{\mathcal{X}}_q = \{\widetilde{x}_i\}_{i=1}^m$ of all \ac{PPM} symbols of order $m$ and whose output alphabet is $\widetilde{\mathcal{Y}}_q = \mathcal{Y}^{2^q}$. For any set $\mathcal{S}\subseteq\intseq{1}{q}$ and a sequence $X_1,\dots,X_q$, we denote the set of random variables $(X_i)_{i\in\mathcal{S}}$ by $X_\mathcal{S}$. Using this notation, we have $X_{1:q}=(X_1,\dots,X_q)$.
For a binary sequence $x_{1:q}\in\mathbb{F}_2^q$, we let $d(x_{1:q})\in\mathbb{N}$ denote the decimal equivalent by taking the least significant bit as the first bit, and $d^{-1}(\cdot)$ denote the inverse operation of $d(\cdot)$, that is, $d^{-1}(d(x_{1:q})) = x_{1:q}$. 
For any set $\mathcal{S}\subseteq\intseq{1}{q}$, we define
\begin{align}
\calA^q(x_\mathcal{S}) \eqdef \{j\in\intseq{1}{2^q}: (d^{-1}(j-1))_\mathcal{S} = x_\mathcal{S}\},~~\forall x_\calS \in \mathbb{F}_2^{|\calS|},
\end{align}
where $(d^{-1}(j-1))_\mathcal{S}$ denotes the components of $d^{-1}(j-1)$ indexed by the elements of $\mathcal{S}$.
The complement of the set $\calA^q(x_\mathcal{S})$ \ac{wrt} $\intseq{1}{2^q}$ is denoted by $\calA^q(x_\mathcal{S})^c$ and is given by
\begin{align}
	\calA^q(x_\mathcal{S})^c \eqdef \intseq{1}{2^q} \backslash \calA^q(x_\mathcal{S}).
\end{align}
We define a \ac{PPM} mapper as
\begin{align}
	\widetilde{x}:\mathbb{F}_2^q\to \widetilde{\calX}^q:x_{1:q} \mapsto \widetilde{x}(x_{1:q}) = \widetilde{x}_{\calA^q(x_{1:q})} = \widetilde{x}_{d(x_{1:q})+1},\label{eq:ppm_mapper}
\end{align}
which is a one-to-one mapping from the set of all binary vectors of length $q$ to the set of all \ac{PPM} symbols of order $2^q$. Hence, we denote the \ac{PPM} super channel equivalently in terms of the binary vector input as $\widetilde{W}_{\widetilde{Y}\vert \widetilde{X}} = \widetilde{W}_{\widetilde{Y}\vert X_{1:q}}$. We use both notations interchangeably depending on the situation. 
For the channel $\widetilde{W}_{\widetilde{Y}\vert \widetilde{X}}$, the distribution at the output by a uniformly distributed PPM symbols of order $m$ as the input is denoted by $P^m_{\text{PPM}}$ and is given by
\begin{align}
	P^m_{\text{PPM}}(\widetilde{z}) \eqdef \frac{1}{m}\sum_{i=1}^{m}\widetilde{W}_{\widetilde{Y}\vert \widetilde{X}}(\widetilde{z}|\widetilde{x}_i)=\frac{1}{2^q}\sum_{x_{1:q}\in\calX^q}\widetilde{W}_{\widetilde{Y}\vert X_{1:q}}(\widetilde{z}|x_{1:q}).
\end{align}
Through out this paper, $\log$ represents the natural logarithm and $\log_2$ represents the logarithm to the base 2.

\section{Impossibility of covert communication with linear codes over \acfp{BI-DMC}}
\label{sec:imposs-covert-comm}
We now consider a \ac{BI-DMC} $(\calX,W_{Z|X},\calZ)$ with $\calX\eqdef \{0,1\}$ and $0$ the innocent symbol for covert communication corresponding to the absence of transmission. We set $Q_1\eqdef W_{Z|X=1}$ and $Q_0\eqdef W_{Z|X=0}$. Our objective is to show that, under mild assumptions, covert communication with linear codes is not possible because such communication is easily detected; in particular, we show that no linear code can achieve the covert capacity.

\paragraph*{A small misconception} One might think that linear codes might not be covert because the linearity constraint requires the sum of codewords to be a codeword; consequently, linear combinations of low-weight codewords would result in a codeword with high weight. This is, however, insufficient to argue that linear codes cannot be covert. In fact, one could consider a systematic code with \emph{unit-weight} codewords in the generator matrix, i.e, $\mathbf{G}=(  \begin{array}{cc}    \mathbf{I}_k&\mathbf{0}_{k\times n-k}  \end{array})$, where $\mathbf{I}_k$ is the identity matrix of size $k$. All codewords of this code have weight at most $k$. Of course, this code would hardly be covert, mainly because the \emph{structure} of the generator matrix allows an attacker to dismiss the last $n-k$ codeword components. The next proposition formalizes this.

\begin{proposition}
  \label{prop:test_linear_codes}
  Consider an $(n,k)$ binary linear code with $(n,k)\in (\mathbb{N}^*)^2,n\geq k$. Assume that the code is used for communication over a \ac{BI-DMC} ($\mathcal{X},W_{Z|X},\mathcal{Z}$) with uniformly distributed messages. There exists a binary hypothesis test with false alarm probability $\alpha$ and missed-detection probability $\beta$ such that
  \begin{align*}
    0\leq \alpha&\leq \frac{16}{k\chi_2(Q_1\|Q_0)}, \\
    0\leq \beta&\leq \frac{1}{k}\left(\frac{16}{\chi_2(Q_1\|Q_0)} + \frac{8\chi_3(Q_1\|Q_0)}{\chi_2(Q_1\|Q_0)^2} - 4\right),
  \end{align*}
  where $Q_1\eqdef W_{Z|X=1}$, $Q_0\eqdef W_{Z|X=0}$, and for $t\in\mathbb{N}^*$ and $t\geq 2$
  \begin{align}
    \chi_t(Q_1\|Q_0) = \sum_{z}\frac{(Q_1(z)-Q_0(z))^t}{Q_0^{t-1}(z)}.
  \end{align}
\end{proposition}

\begin{proof}
  Let $\mathbf{G}\in\mathbb{F}_2^{k\times n}$ be a generator matrix of the code, with columns $\{\mathbf{g}_i\}_{i=1}^n\in\mathbb{F}_2^k$. Denote the set of $2^k$ codewords by $\mathcal{C} \eqdef \{\mathbf{c}_\ell\}_{\ell=1}^{2^k}$. Let $m$ be the number of non all-zero columns in $\mathbf{G}$, and denote the corresponding column indices $\{i_j\}_{j=1}^m$. Note that $k\leq m$ by definition of the dimension of the code. The observations of the adversary in the $m$ positions $\{i_j\}_{j=1}^m$ constitute a sufficient statistics for the detection of communication, and the optimal test is a log-likelihood ratio test restricted to the $m$ positions. We consider a suboptimal test that only operates on a subset of indices $\calS\subset \{i_j\}_{j=1}^m$ that indexes all the distinct columns of $\mathbf{G}$. Since the matrix $\mathbf{G}$ is binary, this implies that the columns $\{\mathbf{g}_i\}_{i\in\calS}$ are pairwise linearly independent and $\card{\calS}\geq k$. Given an observation $\mathbf{z}=(z_1,\dots,z_n)$ at the output of the channel, construct the statistics
  \begin{align}
    T(\mathbf{z})\eqdef \frac{1}{\card{\calS}}\sum_{j\in\calS}\frac{Q_1(z_{j})-Q_0(z_{j})}{Q_0(z_{j})}.\label{eq:sub_optimal_test}
  \end{align}
  In the absence of communication, we have for any $j\in\calS$ that $Z_{j}\sim Q_0$, and therefore
  \begin{align}
    \E[Q_0^\pn]{T(\mathbf{Z})}  &= \frac{1}{\card{\calS}}\sum_{j\in\calS}\E[Q_0]{\frac{Q_1(Z_j)-Q_0(Z_j)}{Q_0(Z_j)}}=0,\displaybreak[0]\\
    \text{Var}_{Q_0^\pn}\left(T(\mathbf{Z})\right)&=\frac{1}{\card{\calS}^2}\left(\sum_{j\in\calS}\E[Q_0]{\frac{(Q_1(Z_j)-Q_0(Z_j))^2}{Q_0(Z_j)^2}}\right.\nonumber\\
                                &\qquad\left.+2\sum_{j,\ell\in\calS:j<\ell}\E[Q_0]{\frac{(Q_1(Z_j)-Q_0(Z_j))}{Q_0(Z_j)}} \E[Q_0]{\frac{(Q_1(Z_\ell)-Q_0(Z_\ell))}{Q_0(Z_\ell)}}\right)\\
                                &=\frac{1}{\card{\calS}}\chi_2(Q_1\|Q_0).
  \end{align}
  In the presence of communication, let $\widehat{Q}^n$ represent the output distribution induced by the code, where
  \begin{align}
    \widehat{Q}^n(\mathbf{z}) = \frac{1}{2^k} \sum_{\mathbf{v}\in\mathbb{F}_2^k} W_{Z|X}^{\pn} (\mathbf{z}|\mathbf{v}\mathbf{G}).
  \end{align}
  For any $j\in\calS$, we have $Z_{j}\sim \frac{1}{2}Q_0+\frac{1}{2}Q_1$ by~\cite[Problem 3.25]{ErrorCorrectionCoding}. Consider now $j,\ell\in\calS$, $j\neq \ell$ and define $N_{ab}=\card{\{\mathbf{c}=(c_1,\dots,c_n) \in \mathcal{C}:c_{j}=a,c_{\ell}=b \}}$ for $a,b\in\{0,1\}$. By definition of $\calS$, columns $\mathbf{g}_{j}$ and $\mathbf{g}_{\ell}$ are linearly independent so that $N_{00}=N_{01}=N_{10}=N_{11}=2^{k-2}$. Hence,
  \begin{align}
    \P{Z_{j}=z,Z_{\ell}=z'}&= \frac{N_{00}}{2^k}Q_0(z)Q_0(z')+\frac{N_{01}}{2^k}Q_0(z)Q_1(z')+\frac{N_{10}}{2^k}Q_1(z)Q_0(z')+\frac{N_{11}}{2^k}Q_1(z)Q_1(z')\\
                           &=\frac{1}{4}Q_0(z)Q_0(z')+\frac{1}{4}Q_0(z)Q_1(z')+\frac{1}{4}Q_1(z)Q_0(z')+\frac{1}{4}Q_1(z)Q_1(z').\label{eq:mixture}
  \end{align}
  Therefore,
  \begin{align}
    \E[\widehat{Q}^n]{T(\mathbf{Z})} &= \frac{1}{\card{\calS}}\sum_{j\in\calS}\E[\frac{1}{2}Q_0+\frac{1}{2}Q_1]{\frac{Q_1(Z_j)-Q_0(Z_j)}{Q_0(Z_j)}}\\
                                     &=\sum_z\left(Q_0(z)+\frac{1}{2}(Q_1(z)-Q_0(z))\right)\frac{Q_1(z)-Q_0(z)}{Q_0(z)}\\
                                     &=\frac{1}{2}\chi_2(Q_1\|Q_0).
  \end{align}
  Similarly,
  \begin{align}
    &\E[\widehat{Q}^n]{T(\mathbf{Z})^2}= \frac{1}{\card{\calS}^2}\E[\widehat{Q}^n]{\sum_{j\in\calS}\frac{(Q_1(Z_{j})-Q_0(Z_{j})) ^2}{Q_0(Z_{j})^2}+\sum_{j,\ell\in\calS:j\neq \ell}\frac{Q_1(Z_{j})-Q_0(Z_{j})}{Q_0(Z_{j})}\frac{Q_1(Z_{\ell})-Q_0(Z_{\ell})}{Q_0(Z_{\ell})}}\\
                                      &\qquad\quad=\frac{1}{\card{\calS}}\chi_2(Q_1\|Q_0) + \frac{1}{2\card{\calS}}\chi_3(Q_1\|Q_0) +\frac{1}{\card{\calS}^2}\sum_{j,\ell\in\calS:j\neq \ell}\E[\widehat{Q}^n]{\frac{Q_1(Z_{j})-Q_0(Z_{j})}{Q_0(Z_{j})}\frac{Q_1(Z_{\ell})-Q_0(Z_{\ell})}{Q_0(Z_{\ell})}}\label{eq:bound_var_h1}.
  \end{align}
  Note that by~(\ref{eq:mixture}), we have
  \begin{align}
    &\E[\widehat{Q}^n]{\frac{Q_1(Z_{j})-Q_0(Z_{j})}{Q_0(Z_{j})}\frac{Q_1(Z_{\ell})-Q_0(Z_{\ell})}{Q_0(Z_{\ell})}}\nonumber\\
    &\qquad\qquad\qquad\qquad\qquad\qquad=\E[\frac{1}{4}Q_{0}Q_0+\frac{1}{4}Q_{0}Q_1+\frac{1}{4}Q_{1}Q_0+\frac{1}{4}Q_{1}Q_1]{\frac{Q_1(Z_{j})-Q_0(Z_{j})}{Q_0(Z_{j})}\frac{Q_1(Z_{\ell})-Q_0(Z_{\ell})}{Q_0(Z_{\ell})}}\\
    &\qquad\qquad\qquad\qquad\qquad\qquad=\frac{1}{4}\left(\E[Q_1]{\frac{Q_1(Z)-Q_0(Z)}{Q_0(Z)}}\right)^2=\frac{1}{4}\chi_2(Q_1\|Q_0)^2.\label{eq:bound_joint_h1}
  \end{align}
  Therefore, combining~(\ref{eq:bound_var_h1}) and (\ref{eq:bound_joint_h1}), we obtain
  \begin{align}
    \text{Var}_{\widehat{Q}^n}\left(T(\mathbf{Z})\right) &=     \E[\widehat{Q}^n]{T(\mathbf{Z})^2}-\E[\widehat{Q}^n]{T(\mathbf{Z})}^2\\
                                                         &=\frac{1}{\card{\calS}}\chi_2(Q_1\|Q_0) + \frac{1}{2\card{\calS}}\chi_3(Q_1\|Q_0) +\frac{\card{\calS}(\card{\calS}-1)}{4\card{\calS}^2}\chi_2(Q_1\|Q_0)^2-\frac{1}{4}\chi_2(Q_1\|Q_0)^2\\
                                                         &=\frac{1}{\card{\calS}}\chi_2(Q_1\|Q_0) + \frac{1}{2\card{\calS}}\chi_3(Q_1\|Q_0) - \frac{1}{4\card{\calS}}\chi_2(Q_1\|Q_0)^2.
  \end{align}
  
  Finally, to simplify the analysis, we choose a convenient threshold $\gamma=\frac{1}{4}\chi_2(Q_1\|Q_0)$. For the test $T(\mathbf{z})$ defined in~(\ref{eq:sub_optimal_test}) together with the threshold $\gamma$, the probability of false alarm $\alpha$ satisfies
  \begin{align}
    \alpha \eqdef \P[Q_0^\pn]{T(\mathbf{Z})>\gamma} \leq \frac{\text{Var}_{Q_0^\pn}\left(T(\mathbf{Z})\right)}{\gamma^2} = \frac{16}{\card{\calS}\chi_2(Q_1\|Q_0)}.
  \end{align}
  Similarly,
  \begin{align}
    \beta = \P[\widehat{Q}^n]{T(\mathbf{Z})<\gamma}  &=\P[\widehat{Q}^n]{-T(\mathbf{Z})>-\gamma}\displaybreak[0]\\
                                                     &=\P[\widehat{Q}^n]{\E[\widehat{Q}^n]{T(\mathbf{Z})}-T(\mathbf{Z})>\frac{1}{2}\chi_2(Q_1\|Q_0)-\gamma}\displaybreak[0]\\
                                                     &\leq \P[\widehat{Q}^n]{\abs{\E[\widehat{Q}^n]{T(\mathbf{Z})}-T(\mathbf{Z})}>\frac{1}{2}\chi_2(Q_1\|Q_0)-\gamma}\displaybreak[0]\\
                                                     &\leq  \frac{\text{Var}_{\widehat{Q}^n}\left(T(\mathbf{Z})\right)}{\gamma^2}\displaybreak[0]\\
                                                     &= \frac{16}{\card{\calS}\chi_2(Q_1\|Q_0)} + \frac{8\chi_3(Q_1\|Q_0)}{\card{\calS}\chi_2(Q_1\|Q_0)^2} -  \frac{4}{\card{\calS}}.
  \end{align}
  The result follows by recalling that $\card{\calS}\geq k$ and $\beta \geq 0$.
\end{proof}

\begin{remark}
  Proposition~\ref{prop:test_linear_codes} does not completely rule out the existence of covert linear codes for some fixed values of $(n,k)$ and such that $\avgD{\smash{\widehat{Q}^n}}{Q_0^\pn}\leq \delta$ with $0<\delta$. However, for $0<\delta\leq\frac{n}{4}\V{Q_1,Q_0}^2$, the codes would still be subject to a square root law. Specifically, using the same notation as for the proof of Proposition~\ref{prop:test_linear_codes}, we have for $j\in\intseq{1}{m}$~\cite[Problem 3.25]{ErrorCorrectionCoding}
  \begin{align}
    \card{\{\mathbf{c}_\ell\in\mathcal{C}:c_{\ell,i_j}=1\}}=    \card{\{\mathbf{c}_\ell\in\mathcal{C}:c_{\ell,i_j}=0\}}=2^{k-1}.
  \end{align}
  Consequently, the probability of having a non-zero symbol in a randomly chosen position and a randomly chosen codeword, both chosen uniformly at random, is
  \begin{align}
    \mu \eqdef \frac{m2^{k-1}}{n2^k} = \frac{m}{2n}.
  \end{align}
  Define the channel output distribution $\widetilde{Q}$ as the average channel output distribution over all positions and all codewords, i.e.,
  \begin{align}
    \widetilde{Q}(z)\eqdef \frac{1}{n}\sum_{t=1}^n\widetilde{Q}_t(z)=\frac{1}{n2^k}\sum_{\ell=1}^{2^k}\sum_{t=1}^nW_{Z|X}(z|c_{\ell t})=\mu Q_1(z)+(1-\mu)Q_0(z).
  \end{align}
  From the converse argument in~\cite{Wang2016,Bloch2016}, we also know that if $\avgD{\smash{\widehat{Q}^n}}{Q_0^\pn}\leq \delta$, we must have
  \begin{align}
    \frac{\delta}{n}&\geq \avgD{\smash{\widetilde{Q}}}{Q_0} \geq \V{\smash{\widetilde{Q}},Q_0}^2 = \left(\mu\V{Q_1,Q_0}\right)^2,
  \end{align}
  where we have used Pinsker's inequality.
  Therefore, we must have
  \begin{align}
    \mu\leq\frac{1}{\V{Q_1,Q_0}}\sqrt{\frac{\delta}{n}}\quad\text{and}\quad   m\leq\frac{2}{\V{Q_1,Q_0}}\sqrt{n\delta},\label{eq:srl_linear}
  \end{align}
  which is the ``square-root law.'' Note that~(\ref{eq:srl_linear}), which is unlikely to be tight because of Pinsker's inequality, is not completely subsumed by the general  converse~\cite{Wang2016,Bloch2016}, which depends on the probability of decoding error.
\end{remark}

As a corollary of Proposition~\ref{prop:test_linear_codes}, we have the following.
\begin{corollary}
  \label{cor:impossibility}
  A family of $(n,k_n)$ linear codes with $(k_n,n)\in(\mathbb{N}^*)^2$ and $ k_n\leq n$, for which $k_n=\omega(1)$ as $n$ goes to infinity cannot be covert. In particular, linear codes cannot achieve any fraction of the covert capacity.
\end{corollary}
\begin{proof}
  Consider a family of $(n,k_n)$ codes with $k_n=\omega(1)$, i.e., $\lim_{n\rightarrow\infty} k_n=\infty$. By Proposition~\ref{prop:test_linear_codes}, there exists a test for which $\alpha = \calO\left(\tfrac{1}{k_n}\right)$ and $\beta=\calO\left(\tfrac{1}{k_n}\right)$ as $n$ goes to infinity and the communication is detected with non zero probability for $n$ large enough. The result follows since achieving covert capacity would require $k_n=\theta \sqrt{n\delta}$ for $\delta \geq \avgD{\widehat{Q}^n}{Q_0^\pn}$ and some $\theta>0$~\cite{Wang2016}.
\end{proof}

\section{Covert communication using \ac{MLC} with PPM symbols}
\label{sec:mult-coding-with}
We now present our  solution for covert communications over a \ac{BI-DMC} using \ac{MLC} and \ac{PPM}. We consider the scenario in which a transmitter (Alice) attempts to communicate over a \ac{BI-DMC} with transition probability $W_{Y|X}$ with a legitimate receiver (Bob), while being eavesdropped by a warden (Willie) who obtains observations of Alice's transmission through another \ac{BI-DMC} with transition probability $W_{Z|X}$. By convention, channel input ``0'' is the innocent symbol corresponding to the absence of communication. For $j\in\{0,1\}$ we set $P_j\eqdef W_{Y|X=j}$ and $Q_j\eqdef W_{Z|X=j}$. Alice encodes her messages $W\in\intseq{1}{M}$ into codewords ${\mathbf{X}}$ of $n$ binary symbols, which are observed by Bob and Willie as ${\mathbf{Y}}$ and ${\mathbf{Z}}$, respectively.  Alice's encoding may be assisted by private randomness $S\in\intseq{1}{K}$ shared only with Bob. Bob forms an estimate $\widehat{W}$ of the transmitted message using $\mathbf{Y}$ and $S$. The objective for Alice is two-fold:
\begin{inparaenum}[(i)]  
\item communicate reliably with Bob, measured by the average probability of error $\P{W\neq\widehat{W}}$;
\item escape detection from the adversary, measured through the relative entropy $\avgD{P_{\mathbf{Z}}}{Q_0^\pn}$ between the distribution $P_{\mathbf{Z}}$ induced by the coding scheme at Willie's output and the ``innocent'' product distribution $Q_0^\pn$.
\end{inparaenum}
A code is said to achieve a covert throughput $R$ with a covert key throughput $R_K$, if as the block length increases we have
\begin{align*}
  &\lim_{n\rightarrow\infty}\frac{\log M}{\sqrt{n\delta}}\geq R,\qquad
  \lim_{n\rightarrow\infty}\frac{\log K}{\sqrt{n\delta}}\leq R_K,\\
  & \lim_{n\rightarrow\infty}\P{W\neq\widehat{W}} = 0\text{, and } \lim_{n\rightarrow\infty}\avgD{P_{\mathbf{Z}}}{Q_0^\pn} \leq \delta,
\end{align*}
for some chosen $\delta>0$. The supremum of all achievable covert throughputs is called the covert capacity and has been characterized in~\cite{Wang2016,Bloch2016} as
\begin{align}
  C_{\text{covert}} = \sqrt{\frac{2}{\chi_2(Q_1\Vert Q_0)}} \avgD{P_1}{P_0},
\end{align}
and the optimal key throughput is given by
\begin{align}
	\sqrt{\frac{2}{\chi_2(Q_1\Vert Q_0)}} \left(\avgD{Q_1}{Q_0}-\avgD{P_1}{P_0}\right).
\end{align}

\subsection{Setup for \ac{MLC} with PPM}

\begin{figure}[b]
  \centering
  \includegraphics[width=0.7\linewidth]{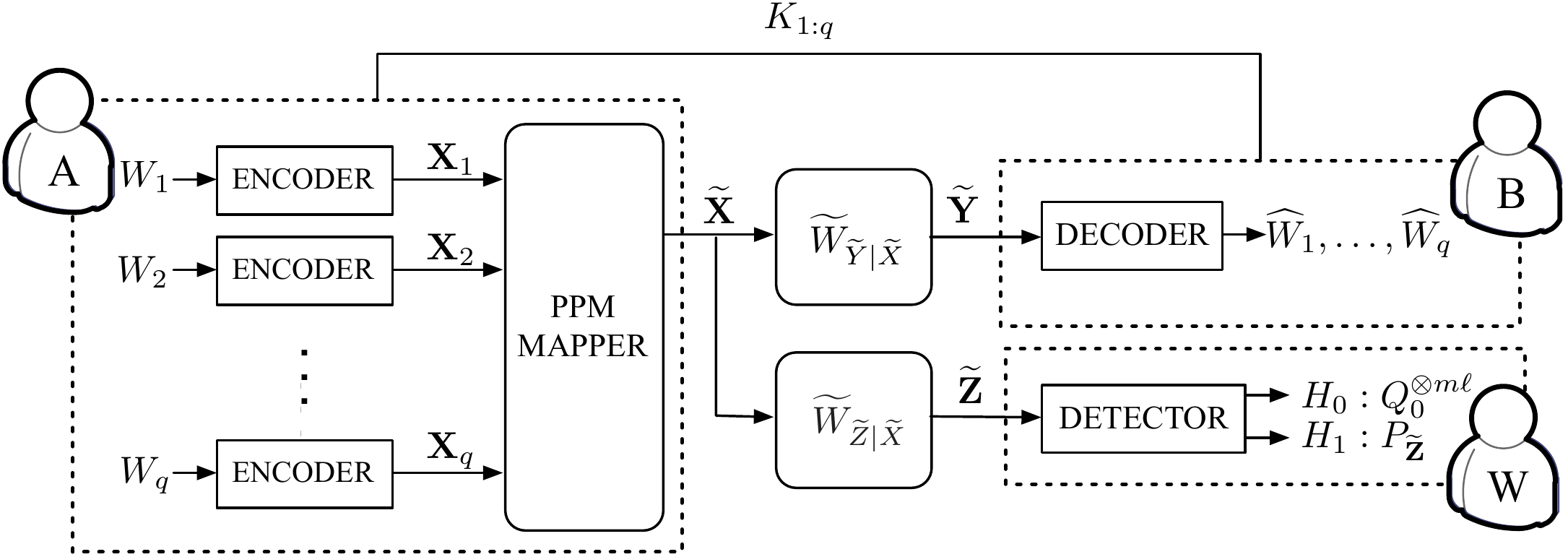}
  \caption{Setup for \ac{MLC} over \ac{PPM} super channel.}
  \label{fig:mlc_ppm_setup}
\end{figure}

Motivated by the negative result of Corollary~\ref{cor:impossibility}, we consider here a non-linear scheme in which the non-linearity is introduced through \ac{PPM} as illustrated in Fig.\ref{fig:mlc_ppm_setup}.  The results in \cite{Bloch2017} show that \ac{PPM} achieves the covert capacity when combined with suitable random non-binary codes. One subtlety behind the results in~\cite{Bloch2017} is that the \ac{PPM} order (and therefore the non-binary field size) grows \emph{linearly} with the blocklength $\ell$ of the code used over the super channel as 
\begin{align}
  \ell = \left\lceil \frac{2\delta}{\chi_2(Q_1\Vert Q_0)} m\right\rceil,
  \label{eq:blocklength_vs_order}
\end{align} 
making it hard to realize such codes in practice with low-complexity.
\ac{RS} codes offer the desired scaling (they are $q$-ary code of length $q-1$), but not the flexibility needed to meet the requirement in~\eqref{eq:blocklength_vs_order}.

To circumvent the design of non-binary codes, we use \ac{MLC} to decompose the \ac{PPM} super-channel into $q\in\mathbb{N}^*$ binary-input channels. Alice divides the message $W$ into $q$ messages $(W_1,\dots,W_q)$ where $W_i$ represents the input message to the $i$-th level, which we will precisely define later. The setup is illustrated in Fig.~\ref{fig:mlc_ppm_setup}, in which $q$ binary encoders feed their outputs to a \ac{PPM} symbol mapper of order $m\eqdef 2^q$. In the $i$-th level, the encoder encodes message $W_i$ into a binary vector $\textbf{X}_i\in \mathbb{F}_2^\ell$. The \ac{PPM} mapper takes these coded binary vectors from each level in parallel and maps the binary vector of length $q$ entering the mapper to a \ac{PPM} symbol; we thus obtain a vector of \ac{PPM} symbols, $\widetilde{\bf X}\in \widetilde{\mathcal{X}}_q^\ell$ at the output of the \ac{PPM} mapper. Bob and Willie observe the outputs of channels $\widetilde{W}_{\widetilde{Y}\vert \widetilde{X}}$ and $\widetilde{W}_{\widetilde{Z}\vert \widetilde{X}}$ denoted by $\widetilde{\bf Y}\in \widetilde{\mathcal{Y}}_q^\ell$ and $\widetilde{\bf Z}\in \widetilde{\mathcal{Z}}_q^\ell$, respectively. Let $Q_0^{\proddist m\ell}$ denote the distribution induced when only innocent symbol ``0" is sent through the channel $W_{Z|X}$ over $m\ell$ channel uses, and $P_{\widetilde{\mathbf{Z}}}$ denote the distribution induced at Willie's receiver by coding over $\ell$ super-channel uses. Before analyzing the coding scheme further, it is worth noting that the benefits of \ac{MLC} are not a priori obvious. In fact, since the number of levels changes with the order of the \ac{PPM} symbol $m$ and the blocklength $\ell$, which are related by (\ref{eq:blocklength_vs_order}), one could expect that the channels perceived at each level $i\in\intseq{1}{q}$ would vary as the blocklength grows, making code design particularly challenging. Perhaps surprisingly, we show that this is \emph{not} the case and that codes may be designed for fixed channels; this is particularly convenient as it allows us to exploit families of channel capacity- and channel resolvability-achieving codes, such as polar codes \cite{Arikan2009,Chou2018}.

\subsection{Information-theoretic analysis of \ac{MLC}}
\label{sec:it_analysis_mlc}

We now present the information-theoretic analysis of \ac{MLC} and show using random coding argument that there exist binary coding schemes with \ac{MLC} that achieve the covert capacity of \acp{BI-DMC}. We also later show in Section~\ref{sec:to_prac_codes} how to design an explicit low-complexity coding scheme by using \ac{MLC} with \ac{MSD}. From \cite{Bloch2016}, we know that we only need a key to achieve covertness when $\mathbb{D}(P_1||P_0) < \mathbb{D}(Q_1||Q_0)$. However, in \ac{MLC}, we are coding over multiple levels and the channel corresponding to individual levels may not hold the same relation. Consequently, irrespective of the relationship between the original channels, we might need keys for some levels, whereas we might be able to send some covert and secret messages over some other levels depending on the mutual information relations between the two channels corresponding to that particular level. We can overcome this by using chaining, which consists of breaking the transmission into several blocks and using the secret messages from one block as the keys for the next block.
By using chaining, we only require extra keys for the first block, and on average the throughputs of messages and keys tend to the optimal throughputs as the number of blocks tends to infinity. We prove this by first characterizing the size of keys and messages for a single block and then developing a chaining strategy that achieves optimal throughput and covertness over $B$ blocks.

We first describe the coding scheme for a single block. Let $W_i = (U_i,V_i)$ with $U_i\in\intseq{1}{M_{U,i}}$ and $V_i\in\intseq{1}{M_{V,i}}$  be a random variable that represents the message for the $i$-th level and $K_i\in\intseq{1}{M_{K,i}}$ be the key shared between Alice and Bob. $V_i$ represents the secret part of the message, which is reused as a key for the next block. $U_i$, $V_i$, and $K_i$ are independent and uniformly distributed random variables. Let $R_{U,i}\eqdef\log_2{M_{U,i}}$, $R_{V,i}\eqdef\log_2{M_{V,i}}$, and $R_{K,i}\eqdef\log_2{M_{K,i}}$ represent the rates of $U_i$, $V_i$, and $K_i$, respectively. Encoders for each level are independent; therefore, we can view the \ac{MLC} scheme in Fig.\ref{fig:mlc_ppm_setup} as a $q$-user \ac{MAC} in which the \ac{PPM} mapper is integrated to the channel. Because of the symmetry of the \ac{PPM} super channel, a uniform distribution on \ac{PPM} symbols achieves the capacity for Bob's channel. Note that, to ensure~\eqref{eq:blocklength_vs_order} holds, the number of levels $q$ depends on the blocklength of the code $\ell$. Hence, we cannot directly use the results for \ac{MAC} with a fixed number of users. After careful analysis, we show that the result remains what would have expected from \cite{ElGamal2011}. The codebook, encoding, and decoding for one block are as follows.
\begin{itemize}
	\item {\it Codebook}: For $i$-th level, generate a codebook of  $M_{U,i}\times M_{V,i}\times M_{K,i}$ codewords, each of length $\ell$, denoted by $\mathbf{X}_i(u_i,v_{i},k_i)$, where $u_i\in\intseq{1}{M_{U,i}}$, $v_i\in\intseq{1}{M_{V,i}}$, and $k_i\in\intseq{1}{M_{K,i}}$. Each codeword is generated independently according to a uniform distribution on $\calX^\ell$, denoted by $P_{X_i}^{\proddist\ell}$. 
	\item {\it Encoding}: Given messages $U_i$ and $V_{i}$ and key $K_i$ for $i$-th level, the $i$-th encoder selects the codeword $\mathbf{X}_i(U_i,V_{i},K_i)$.
	\item {\it Decoding}: The receiver decodes $\hat{\mathbf{w}} = (\hat{u}_1,\dots,\hat{u}_q,\hat{v}_{1},\dots,\hat{v}_{q})$ from $\widetilde{\mathbf{y}}$ knowing the key $\mathbf{k} = (k_1,\dots,k_q)$ if $\hat{\mathbf{w}}$ is the unique message such that $(\mathbf{X}_1(\hat{u}_1,\hat{v}_{1},k_1),\dots,\mathbf{X}_q(\hat{u}_q,\hat{v}_{q},k_q),\widetilde{\mathbf{y}})\in \mathcal{T}_\Gamma^\ell$, where $\mathcal{T}_\Gamma^\ell$ is defined as follows. For $\Gamma = \{\gamma_{\mathcal{S}}\}_{\mathcal{S}\subseteq\intseq{1}{q}}$, 
	\begin{align}
	\mathcal{T}_\Gamma^\ell \eqdef \left\{ \left(\mathbf{x}_1,\dots,\mathbf{x}_q,\mathbf{\widetilde{y}} \right)\in \mathcal{X}^\ell \times \cdots \times \mathcal{X}^\ell \times \widetilde{\mathcal{Y}}^\ell : \log\frac{W_{\widetilde{Y}|X_{1:q}}^{\proddist \ell} \left( \mathbf{\widetilde{y}} | \mathbf{x}_{1:q} \right) }{P_{\widetilde{\mathbf{Y}}|\mathbf{X}_{\calS^c}}\left(\mathbf{\widetilde{y}} | \mathbf{x}_{\mathcal{S}^c}\right)} > \gamma_{\mathcal{S}}, ~~~\forall \mathcal{S} \subseteq \intseq{1}{q} \right\},
	\end{align}
	where $P_{\widetilde{\mathbf{Y}}|\mathbf{X}_{\calS^c}}$ is the conditional distribution given by
	\begin{align}
	P_{\widetilde{\mathbf{Y}}|\mathbf{X}_{\calS^c}}\left(\mathbf{\widetilde{y}} | \mathbf{x}_{\mathcal{S}^c}\right) &= \sum_{\mathbf{x}_\calS}\prod_{j\in\calS}P_{X_j}^{\proddist\ell}(\mathbf{x}_j)W_{\widetilde{Y}|X_{1:q}}^{\proddist \ell} \left( \mathbf{\widetilde{y}} | \mathbf{x}_{1:q} \right)\displaybreak[0]\\
	&=\prod_{i=1}^{\ell}\sum_{x_{\calS,i}}~\prod_{j\in\calS}P_{X_j}(x_{j,i})W_{\widetilde{Y}|X_{1:q}}\left( \widetilde{y}_i | x_{1:q,i}\right),\displaybreak[0]\\
	&=\prod_{i=1}^{\ell}P_{\widetilde{{Y}}|{X}_{\calS^c}}\left(\widetilde{y}_i|x_{\mathcal{S}^c,i}\right).
	\end{align}
\end{itemize}

To simplify the notation, we use the following definitions:
\begin{align}
\svbu &\eqdef (u_1,\dots,u_q),~\svbv \eqdef (v_{1},\dots,v_{q}),~\svbw \eqdef (\svbu,\svbv),\displaybreak[0]\\
\mathbf{X}_{1:q}(\mathbf{w},\mathbf{k}) &\eqdef \mathbf{X}_{1:q}(\svbu,\svbv,\mathbf{k}) \eqdef (\mathbf{X}_1(u_1,v_1,k_1),\dots,\mathbf{X}_q(u_q,v_q,k_q)),\displaybreak[0]\\
X_{1:q,i}(\mathbf{w},\mathbf{k})&\eqdef X_{1:q,i}(\mathbf{u},\mathbf{v},\mathbf{k})\eqdef(X_{1,i}(u_1,v_1,k_1),\dots,X_{q,i}(u_q,v_q,k_q),\displaybreak[0]\\
\mathbf{X}_{\calG}(\svbw,\svbk) &\eqdef (\mathbf{X}_i(u_i,v_{i},k_i))_{i\in\calG},~~~\forall \calG\subseteq\intseq{1}{q},\displaybreak[0]\\
M_U &\eqdef M_{U,1}\times\dots\times M_{U,q},\displaybreak[0]\\
M_V &\eqdef M_{V,1}\times\dots\times M_{V,q},\displaybreak[0]\\
M_K &\eqdef M_{K,1}\times\dots\times M_{K,q},\displaybreak[0]\\
M &\eqdef M_U\times M_V.
\end{align}

We first analyze channel reliability and channel resolvability for one block and prove the existence of a random code that satisfies the requirements. Then, we use chaining over blocks using the same code for each block to obtain a coding scheme that achieves optimal rate. For chaining over $B$ blocks, the variables in the $j$-th block is identified by a superscript; for example, we denote the message vector by $\svbw^{(j)}$, the key vector by $\svbk^{(j)}$, the corresponding codeword for the $i$-th level by $\mathbf{x}_i(\svbw^{(j)},\svbk^{(j)})\eqdef\mathbf{x}_i(w_i^{(j)},k_i^{(j)})$, and the outputs of Bob's channel and Willie's channel by $\widetilde{\mathbf{y}}^{(j)}$ and $\widetilde{\mathbf{z}}^{(j)}$ respectively.

\subsubsection{Reliability Analysis}
We now analyze the rate requirements for achieving vanishing probability of error using random codes for a single block. We denote the probability of error for a given codebook by $\P{\widehat{\mathbf{W}}\neq\mathbf{W}}$ and the expectation of it over the random codebook distribution by $\mathbb{E}\left[\P{\widehat{\mathbf{W}}\neq\mathbf{W}}\right]$. The following lemma summarizes the requirement.
\begin{lemma}\label{lemma:reliability_one_block}
	For $\calS\subseteq\intseq{1}{q}$ and $0<\epsilon<\varepsilon/2$, if the rates of random codes satisfy 
	\begin{align}
	\sum_{i\in\calS}(R_{U,i}+R_{V,i}) \leq  I(X_{\calS};\widetilde{Y}|X_{\calS^c})-\frac{\varepsilon}{q},\label{eq:reliability_rate_conditions_1}
	\end{align}
	then
	\begin{align}
	\mathbb{E}\left[\P{\widehat{\mathbf{W}}\neq\mathbf{W}}\right]  &\leq \delta_1 \eqdef m\left( e^{-\frac{2\ell\epsilon^2}{q^4}} +  2^{-\frac{\ell\epsilon}{q}} \right).
	\end{align}
\end{lemma}
\begin{proof}
	The probability of error averaged over the codebook distribution is given by
	\begin{align}
	\mathbb{E}\left[\P{\widehat{\mathbf{W}}\neq\mathbf{W}|\mathbf{K}=\svbk}\right] &= \mathbb{E}\left[\sum_{\widetilde{\mathbf{y}}} \sum_{\svbw} \frac{1}{M} \widetilde{W}_{\widetilde{Y}|X_{1:q}}^{\proddist \ell}\left(\widetilde{\mathbf{y}}\vert \mathbf{X}_{1:q}(\svbw,\svbk)\right) \right.\nonumber\displaybreak[0]\\
	&\qquad\qquad \left. \vphantom{\sum_{\widetilde{\mathbf{y}}} \sum_{\svbw}  \frac{1}{M} \widetilde{W}^{\proddist \ell}\left(\widetilde{\mathbf{y}}\vert \mathbf{X}_{1:q}(\svbw,\svbk)\right)} \mathds{1}\left\{\left(\mathbf{X}_{1:q}(\svbw,\svbk),\widetilde{\mathbf{y}}\right)\notin\mathcal{T}_\Gamma^\ell \text{ or } \exists \svbw'\neq \svbw \text{ s.t. } \left(\mathbf{X}_{1:q}(\svbw',\svbk),\widetilde{\mathbf{y}}\right)\in\mathcal{T}_\Gamma^\ell \right\}\right]\displaybreak[0]\\
	&=\mathbb{E}\left[\sum_{\widetilde{\mathbf{y}}}  \widetilde{W}_{\widetilde{Y}|X_{1:q}}^{\proddist \ell}\left(\widetilde{\mathbf{y}}\vert \mathbf{X}_{1:q}(\mathbf{1},\svbk) \right)  \right.\nonumber\displaybreak[0]\\
	&\qquad\qquad \left. \vphantom{\sum_{\widetilde{\mathbf{y}}}  \widetilde{W}^{\proddist \ell}\left(\widetilde{\mathbf{y}}\vert \mathbf{X}_{1:q}(\mathbf{1},\svbk) \right)} \mathds{1}\left\{\left(\mathbf{X}_{1:q}(\mathbf{1},\svbk),\widetilde{\mathbf{y}}\right)\notin\mathcal{T}_\Gamma^\ell \text{ or } \exists \svbw\neq \mathbf{1} \text{ s.t. } \left(\mathbf{X}_{1:q}(\svbw,\svbk),\widetilde{\mathbf{y}}\right)\in\mathcal{T}_\Gamma^\ell \right\}\right]\displaybreak[0]\\
	&\leq \mathbb{E}\left[\sum_{\widetilde{\mathbf{y}}}  \widetilde{W}_{\widetilde{Y}|X_{1:q}}^{\proddist \ell}\left(\widetilde{\mathbf{y}}\vert \mathbf{X}_{1:q}(\mathbf{1},\svbk) \right) \mathds{1}\left\{\left(\mathbf{X}_{1:q}(\mathbf{1},\svbk) ,\widetilde{\mathbf{y}}\right)\notin\mathcal{T}_\Gamma^\ell \right\} \right] +  \nonumber \displaybreak[0]\\
	& \qquad\qquad \sum_{\svbw \neq \mathbf{1}} \mathbb{E}\left[\sum_{\widetilde{\mathbf{y}}} \widetilde{W}_{\widetilde{Y}|X_{1:q}}^{\proddist \ell}\left(\widetilde{\mathbf{y}}\vert \mathbf{X}_{1:q}(\mathbf{1},\svbk)  \right) \mathds{1}\left\{ \left(\mathbf{X}_{1:q}(\svbw,\svbk) ,\widetilde{\mathbf{y}}\right)\in\mathcal{T}_\Gamma^\ell \right\} \right]\displaybreak[0]\\
	&\leq \sum_{\mathcal{S}\subseteq\intseq{1}{q}}\P[\widetilde{W}^{\proddist\ell}_{\widetilde{Y}|{X}_{1:q}}\prod_{i=1}^{q}P^{\proddist\ell}_{{X}_{i}}]{\log \frac{\widetilde{W}_{\widetilde{Y}|X_{1:q}}^{\proddist \ell} \left( \widetilde{Y}^\ell | \mathbf{X}_{1:q}(\mathbf{1},\svbk) \right) }{P_{\widetilde{\mathbf{Y}}|\mathbf{X}_{\calS^c}}\left(\widetilde{Y}^\ell | \mathbf{X}_{\mathcal{S}^c}({\bf 1},\svbk)\right)} \leq \gamma_{\mathcal{S}}}  + \nonumber \displaybreak[0]\\
	& \qquad\qquad \sum_{\svbw\neq \mathbf{1}} \mathbb{E}\left[\sum_{\widetilde{\mathbf{y}}} \widetilde{W}_{\widetilde{Y}|X_{1:q}}^{\proddist \ell}\left(\widetilde{\mathbf{y}}\vert \mathbf{X}_{1:q}(\mathbf{1},\svbk) \right) \mathds{1}\left\{ \left(\mathbf{X}_{1:q}(\svbw,\svbk),\widetilde{\mathbf{y}}\right)\in\mathcal{T}_\Gamma^\ell \right\} \right].\label{eq:mac_sum_prob_error_gen}
	\end{align}
	We first analyze the first term on the right hand side of \eqref{eq:mac_sum_prob_error_gen}. We have
	\begin{align}
	\log\frac{\widetilde{W}_{\widetilde{Y}|X_{1:q}}^{\proddist \ell}\left(\widetilde{\mathbf{Y}}|\mathbf{X}_{1:q}(\mathbf{1},\svbk)\right)}{P_{\widetilde{\mathbf{Y}}|\mathbf{X}_{\calS^c}}\left(\widetilde{\mathbf{Y}}|\mathbf{X}_{\mathcal{S}^c}({\bf 1},\svbk)\right)} = \sum_{i=1}^{\ell} \log\frac{\widetilde{W}_{\widetilde{Y}|X_{1:q}} \left(\widetilde{Y}_i|X_{1:q,i}(\mathbf{1},\svbk)\right)} {P_{\widetilde{{Y}}|{X}_{\calS^c}}\left(\widetilde{Y}_i|X_{\mathcal{S}^c,i}(\mathbf{1},\svbk)\right)}.
	\end{align}
	Let 
	\begin{align}
	A_i = \log\frac{\widetilde{W}_{\widetilde{Y}|X_{1:q}} \left(\widetilde{Y}_i|X_{1:q,i}(\mathbf{1},\svbk)\right)} {P_{\widetilde{{Y}}|{X}_{\calS^c}}\left(\widetilde{Y}_i|X_{\mathcal{S}^c,i}(\mathbf{1},\svbk)\right)}.
	\end{align}
	We have 
	\begin{align}
	\mathbb{E}\left[A_i\right] = I(X_\mathcal{S};\widetilde{Y}|X_{\mathcal{S}^c}).
	\end{align}
	We now show that $A_i$ is bounded. For any $\mathcal{S}\subseteq \intseq{1}{q}$,
	\begin{align}
	P_{\widetilde{{Y}}|{X}_{\calS^c}}\left(\widetilde{Y}_i|X_{\mathcal{S}^c,i}(\mathbf{1},\svbk)\right) &= \frac{1}{2^{|\mathcal{S}|}} \sum_{x_\mathcal{S}} \widetilde{W}_{\widetilde{Y}|X_{1:q}}\left( \widetilde{Y}_i|x_{\mathcal{S}},X_{\mathcal{S}^c,i}(\mathbf{1},\svbk)\right) \\
	&\geq \frac{1}{2^{|\mathcal{S}|}} \widetilde{W}_{\widetilde{Y}|X_{1:q}}\left( \widetilde{Y}_i|X_{\mathcal{S},i}(\mathbf{1},\svbk),X_{\mathcal{S}^c,i}(\mathbf{1},\svbk)\right) .
	\end{align}
	Therefore,
	\begin{align}
	A_i \leq |\mathcal{S}| \leq q.\label{eq:Ai_upper_bound}
	\end{align}
	Also,
	\begin{align}
	\frac{\widetilde{W}_{\widetilde{Y}|X_{1:q}}\left(\widetilde{Y}_i|X_{1:q,i}(\mathbf{1},\svbk)\right)}{P_{\widetilde{{Y}}|{X}_{\calS^c}}\left(\widetilde{Y}_i|X_{\mathcal{S}^c,i}(\mathbf{1},\svbk)\right)} &= \frac{P_1(Y_{i,\mathcal{A}^q(X_{1:q,i}(\mathbf{1},\svbk))}) \prod_{h'\neq \mathcal{A}^q(X_{1:q,i}(\mathbf{1},\svbk))} P_0(Y_{i,h'})}{\frac{1}{2^{|\mathcal{S}|}}\sum_{j\in\mathcal{A}^q(X_{\mathcal{S}^c,i}(\mathbf{1},\svbk))} P_1(Y_{i,j}) \prod_{h\neq j} P_0(Y_{i,h})}\displaybreak[0]\\
	&= \frac{\frac{P_1(Y_{i,j'})}{P_0(Y_{i,j'})}} {\frac{1}{2^{|\mathcal{S}|}}\sum_{j\in\mathcal{A}^q(X_{\mathcal{S}^c,i}(\mathbf{1},\svbk))}  \frac{P_1(Y_{i,j})}{P_0(Y_{i,j})}}\displaybreak[0]\\
	&\geq \frac{P_1(Y_{i,j'})}{\frac{1}{2^{|\mathcal{S}|}} \sum_{j\in\mathcal{A}^q(X_{\mathcal{S}^c,i}(\mathbf{1},\svbk))}  \frac{1}{P_0(Y_{i,j})}}\displaybreak[0]\\
	&\geq \frac{\mu_1}{\frac{1}{2^{|\mathcal{S}|}} \sum_{j\in\mathcal{A}^q(X_{\mathcal{S}^c,i}(\mathbf{1},\svbk))}  \frac{1}{\mu_0}}\displaybreak[0]\\
	&\geq \mu_0 \mu_1,\label{eq:Ai_lower_bound}
	\end{align}
	where $$\mu_0 = \min_{y\in\text{supp}(P_0)} P_0(y),~~~~~~~\mu_1 = \min_{y\in\text{supp}(P_1)} P_1(y).$$
	From \eqref{eq:Ai_upper_bound} and \eqref{eq:Ai_lower_bound}, we have
	\begin{align}
	\log(\mu_0\mu_1) \leq A_i \leq q \implies |A_i| \leq \max\left(\log\frac{1}{\mu_0\mu_1},q\right).
	\end{align}
	By choosing $\gamma_{\mathcal{S}}^\ell = \ell \left(I(X_\mathcal{S};\widetilde{Y}|X_{\mathcal{S}^c})-\epsilon/q\right)$ and using Hoeffding's inequality, for large enough $q$, we have
	\begin{align}
	\P{\log \frac{\widetilde{W}_{\widetilde{Y}|X_{1:q}}^{\proddist \ell} \left( \widetilde{\mathbf{y}} | \mathbf{x}_1,\dots,\mathbf{x}_q \right) }{P_{\widetilde{\mathbf{Y}}|\mathbf{X}_{\calS^c}}\left(\widetilde{\mathbf{y}} | \mathbf{x}_{\mathcal{S}^c}\right)} \leq \gamma_{\mathcal{S}}}  \leq e^{-\frac{2\ell\epsilon^2}{q^4}}.\label{eq:reliab_first_term}
	\end{align}
	
	We now analyze the second term in \eqref{eq:mac_sum_prob_error_gen}. For any $\svbw \neq \mathbf{1}$, we let $\calG\subset \intseq{1}{q}$ denote the set of levels $j$ for which $\svbw_j =(u_j,v_j) =  \mathbf{1}$. Then, we have 
	\begin{align}
	&\mathbb{E}\left[\sum_{\widetilde{\mathbf{y}}} \widetilde{W}_{\widetilde{Y}|X_{1:q}}^{\proddist \ell}\left(\widetilde{\mathbf{y}}\vert \mathbf{X}_{1:q}(\mathbf{1},\svbk) \right) \mathds{1}\left\{ \left(\mathbf{X}_{1:q}(\svbw,\svbk),\widetilde{\mathbf{y}}\right)\in\mathcal{T}_\Gamma^\ell \right\} \right]\nonumber\\
	&\qquad=\sum_{\widetilde{\mathbf{y}}}\sum_{\{\mathbf{x}_j(\mathbf{1},k_j)\}_{j\in \calG}} ~~\prod_{j\in \calG} P_{X_j}(\mathbf{x}_j(\mathbf{1},k_j)) \sum_{\{\mathbf{x}_{j'}(\mathbf{1},k_{j'}),\mathbf{x}_{j'}(\svbw_{j'},k_{j'})\}_{j'\in \calG^c}} ~~\prod_{j' \in \calG^c} P_{X_{j'}}(\mathbf{x}_{j'}(\mathbf{1},k_{j'})) P_{X_{j'}}(\mathbf{x}_{j'}(\svbw_{j'},k_{j'})) \nonumber \\
	&\vphantom{\sum_{\widetilde{\mathbf{y}}}}\qquad\qquad\qquad\qquad\qquad\qquad\qquad\qquad\qquad\qquad \widetilde{W}_{\widetilde{Y}|X_{1:q}}^{\proddist \ell}\left( \widetilde{\mathbf{y}} | \mathbf{x}_{1:q}(\mathbf{1},\svbk) \right) \mathds{1}\left\{(\mathbf{x}_{\calG}(\mathbf{1},\svbk),\mathbf{x}_{\calG^c}(\svbw,\svbk)) \in \mathcal{T}_\Gamma^\ell \right\}\displaybreak[0]\\
	&\qquad=\sum_{\widetilde{\mathbf{y}}}\sum_{\{\mathbf{x}_j(\mathbf{1},k_j)\}_{j\in \calG}} ~~\prod_{j\in \calG} P_{X_j}(\mathbf{x}_j(\mathbf{1},k_j)) \sum_{\{\mathbf{x}_{j'}(\svbw_{j'},k_{j'})\}_{j'\in \calG^c}} ~~\prod_{j' \in \calG^c} P_{X_{j'}}(\mathbf{x}_{j'}(\svbw_{j'},k_{j'})) \nonumber \\
	&\vphantom{\sum_{\widetilde{\mathbf{y}}}}\qquad\qquad\qquad\qquad\qquad\qquad\qquad\qquad\qquad\qquad P_{\widetilde{\mathbf{Y}}|\mathbf{X}_\calG}\left( \widetilde{\mathbf{y}} | \mathbf{x}_{\calG}(\mathbf{1},\svbk) \right) \mathds{1}\left\{(\mathbf{x}_{\calG}(\mathbf{1},\svbk),\mathbf{x}_{\calG^c}(\svbw,\svbk)) \in \mathcal{T}_\Gamma^\ell \right\}\displaybreak[0]\\
	&\qquad\overset{(a)}{\leq} 2^{-\gamma_{\calG^c}} \sum_{\widetilde{\mathbf{y}}}~~\sum_{\{\mathbf{x}_j(\mathbf{1},k_j)\}_{j\in \calG}} ~~\prod_{j\in \calG} P_{X_j}(\mathbf{x}_j(\mathbf{1},k_j)) \nonumber\\
	&\qquad\qquad\qquad\qquad\qquad\sum_{\{\mathbf{x}_{j'}(\svbw_{j'},k_{j'})\}_{j'\in \calG^c}} ~~~\prod_{j' \in \calG^c} P_{X_{j'}}(\mathbf{x}_{j'}(\svbw_{j'},k_{j'})) \widetilde{W}_{\widetilde{Y}|X_{1:q}}^{\proddist \ell}\left( \widetilde{\mathbf{y}} | \mathbf{x}_{\calG}(\mathbf{1},\svbk),\mathbf{x}_{\calG^c}(\svbw_{\calG_c},\svbk) \right)\displaybreak[0]\\
	&\qquad = 2^{-\gamma_{\calG^c}},
	\end{align}
	where (a) follows since $P_{\widetilde{\mathbf{Y}}|\mathbf{X}_\calG}\left( \widetilde{\mathbf{y}} | \mathbf{x}_{\calG}(\mathbf{1},\svbk) \right) \leq 2^{-\gamma_{\calG^c}} \widetilde{W}_{\widetilde{Y}|X_{1:q}}^{\proddist \ell}\left( \widetilde{\mathbf{y}} | \mathbf{x}_{\calG}(\mathbf{1},\svbk),\mathbf{x}_{\calG^c}({\svbw},\svbk) \right)$ for $(\mathbf{x}_{\calG}(\mathbf{1},\svbk),\mathbf{x}_{\calG^c}({\svbw},\svbk)) \in \mathcal{T}_\Gamma^\ell$ and upper bounding the indicator function by 1.
	
	Hence,
	\begin{align}
	&\sum_{\svbw\neq \mathbf{1}} \mathbb{E}\left[\sum_{\widetilde{\mathbf{y}}} \widetilde{W}_{\widetilde{Y}|X_{1:q}}^{\proddist \ell}\left(\widetilde{\mathbf{y}}\vert \mathbf{X}_{1:q}(\mathbf{1},\svbk) \right) \mathds{1}\left\{ \left(\mathbf{X}_{1:q}(\svbw,\svbk),\widetilde{\mathbf{y}}\right)\in\mathcal{T}_\Gamma^\ell \right\} \right]\\
	&\qquad\qquad\qquad\qquad\qquad\qquad\qquad\qquad\qquad \leq \sum_{\calG\subset \intseq{1}{q}} \prod_{i\in \calG^c} (M_{U,i} M_{V,i}-1) 2^{-\gamma_{\calG^c}}\displaybreak[0]\\
	&\qquad\qquad\qquad\qquad\qquad\qquad\qquad\qquad\qquad = \sum_{\calG\subset \intseq{1}{q}} \prod_{i\in \calG^c} (M_{U,i} M_{V,i}-1) 2^{-\gamma_{\calG^c}}\displaybreak[0]\\
	&\qquad\qquad\qquad\qquad\qquad\qquad\qquad\qquad\qquad \leq \sum_{\calG\subset \intseq{1}{q}} \ 2^{-\gamma_{\calG^c}} \prod_{i\in \calG^c} M_{U,i} M_{V,i}.\label{eq:reliab_sec_term}
	\end{align}
	
	Therefore, from \eqref{eq:mac_sum_prob_error_gen}, \eqref{eq:reliab_first_term}, and \eqref{eq:reliab_sec_term}, we get
	\begin{align}
	\mathbb{E}\left[\P{\widehat{\mathbf{W}}\neq\mathbf{W}}\right]  &\leq \sum_{\mathcal{S}\subseteq\intseq{1}{q}}  e^{-\frac{2\ell\epsilon^2}{q^4}} + \sum_{\calG\subset \intseq{1}{q}} 2^{-\gamma_{\calG^c}} \prod_{i\in \calG^c} M_{U,i} M_{V,i}\displaybreak[0]\\
	&= \sum_{\mathcal{S}\subseteq\intseq{1}{q}}  e^{-\frac{2\ell\epsilon^2}{q^4}}+  \sum_{\calG\subset \intseq{1}{q}} 2^{-\ell\left(I(X_{\mathcal{G}^c};\widetilde{Y}|X_{\mathcal{G}}) - \frac{\epsilon}{q} - \sum_{i\in{\mathcal{G}^c}} (R_{U,i} + R_{V,i}) \right)} \displaybreak[0]\\
	&\overset{(a)}{\leq} \sum_{\mathcal{S}\subseteq\intseq{1}{q}}   e^{-\frac{2\ell\epsilon^2}{q^4}} +  \sum_{\calG\subset \intseq{1}{q}} 2^{-\frac{\ell\left(\varepsilon  -\epsilon\right)}{q}} \displaybreak[0]\\
	&\overset{(b)}{<}m\left( e^{-\frac{2\ell\epsilon^2}{q^4}} +  2^{-\frac{\ell\epsilon}{q}} \right) ,\label{eq:prob_error_single_block}
	\end{align}
	where (a) follows from choosing $\sum_{i\in\mathcal{G}^c} (R_{U,i}+R_{V,i}) \leq I(X_{\mathcal{G}^c};\widetilde{Y}|X_{\mathcal{G}})  - \varepsilon/q$ and (b) follows by upper bounding the number of subsets of $\intseq{1}{q}$ and by choosing $\varepsilon>2\epsilon$.
	
	For $\mathcal{G}\subset\intseq{1}{q}$, let $i\in\mathcal{G}$. Then, we have 
	\begin{align}
	I(X_{\mathcal{G}^c};\widetilde{Y}|X_{\mathcal{G}}) =I(X_{\mathcal{G}^c\backslash\{i\}};\widetilde{Y}|X_{\mathcal{G}})+I(X_i;\widetilde{Y}|X_{\intseq{1}{q}\backslash\{i\}})\geq I(X_i;\widetilde{Y}|X_{\intseq{1}{q}\backslash\{i\}}).
	\end{align}
	Because $X_i$'s are independent and uniformly distributed, the quantity $I(X_i;\widetilde{Y}|X_{\intseq{1}{q}\backslash\{i\}})$ is same for all $i$. Note that given the realizations of $X_{\intseq{1}{q}\backslash\{i\}}=x_{\intseq{1}{q}\backslash\{i\}}$, the position of symbol ``1" in the PPM symbol $\widetilde{x}(X_i,x_{\intseq{1}{q}\backslash\{i\}})$ is given by one of the two positions indexed by $\calA^q(x_{\intseq{1}{q}\backslash\{i\}})$. Therefore, $I(X_i;\widetilde{Y}|X_{\intseq{1}{q}\backslash\{i\}})$ is equivalent to the mutual information between the input and outputs when using \ac{PPM} symbols of order $2$ in the \ac{MLC}-\ac{PPM} scheme, and using~\cite[Eq.(13)]{Bloch2017}, we obtain
	\begin{align}
	I(X_i;\widetilde{Y}|X_{\intseq{1}{q}\backslash\{i\}})= \mathbb{D}\left( P_1 \Vert P_0 \right) - \mathbb{D}\left( {P}_{\text{PPM}}^{2} \Vert P_0^{\proddist 2} \right),
	\end{align}
	where ${P}_{\text{PPM}}^{2}$ represents the output distribution of Bob's channel when the input is uniform over PPM symbols of order $2$.
	Hence, we can find rates satisfying $\sum_{i\in\mathcal{G}^c} (R_{U,i}+R_{V,i}) \leq I(X_{\mathcal{G}^c};\widetilde{Y}|X_{\mathcal{G}})  - \varepsilon/q$ for $\varepsilon/q < \mathbb{D}\left( P_1 \Vert P_0 \right) - \mathbb{D}\left( {P}_{\text{PPM}}^{2} \Vert P_0^{\proddist 2} \right)$ for all $\mathcal{G}\subset\intseq{1}{q}$.
\end{proof}

\subsubsection{Covertness Analysis}
We now analyze the channel resolvability of Willie's channel and prove the existence of channel resolvability codes if the rates satisfy some condition. Let $P_{\widetilde{\mathbf{Z}}}$ be the distribution induced by the code for a given block, ${Q}_{\text{PPM}}^{m}$ be the output distribution of Willie's channel when the input is uniform over all possible PPM symbols of order $m$ ,and $Q_0^{\proddist m\ell}$ be the innocent distribution.
By calculations similar to~\cite[(77)-(79)]{Bloch2016}, we have
\begin{align}
\avgD{P_{\mathbf{\widetilde{Z}}}}{Q_0^{\proddist m\ell}} \leq \avgD{P_{\mathbf{\widetilde{Z}}}}{(Q_{\text{PPM}}^{m})^{\proddist \ell}} + 2\V{P_{\mathbf{\widetilde{Z}}},(Q_{\text{PPM}}^{m})^{\proddist \ell}}\max_{\widetilde{z}}\ell \abs{\log\frac{(Q_{\text{PPM}}^{m})(\widetilde{z})}{Q_0^{\proddist m}(\widetilde{z})}} + \D{(Q_{\text{PPM}}^{m})^{\proddist \ell} }{Q_0^{\proddist m\ell}}.\label{eq:resolve_one_block}
\end{align}
We show using the following lemma that the first two terms go to zero exponentially in $\ell$.

\begin{lemma}\label{lemma:resolve_one_block}
	For $\calS\subseteq\intseq{1}{q}$ and $0<\epsilon<\varepsilon/2$, if the rates of random codes satisfy 
	\begin{align}
	\sum_{i\in\calS}(R_{U,i}+R_{K,i}) \geq  I(X_{\calS};\widetilde{Z})+\frac{\varepsilon}{q},\label{eq:resolve_rate_conditions_1}
	\end{align}
	then
	\begin{align}
	\mathbb{E}\left[\D{P_{\widetilde{\mathbf{Z}}|\mathbf{V} = \mathbf{v}}}{\left(Q_{\textnormal{PPM}}^m\right)^{\proddist\ell}}\right] &\leq \delta_2 \eqdef m2^{-\ell\epsilon/q}  + \frac{m^2\ell 2^{q/m\ell}}{\mu_z} e^{-\frac{2\ell\epsilon^2}{q^4}},\label{eq:resolve_average}
	\end{align}
	where $\mu_z = \min_{z\in\textnormal{supp}(Q_Z)}Q_Z(z)$ with $Q_Z(z) = \frac{1}{2}\sum_{x\in\calX}W_{Z|X}(z|x)$.
\end{lemma}
Lemma~\ref{lemma:resolve_one_block} actually proves a stronger result that what we need to bound~\eqref{eq:resolve_one_block}, as it guarantees that, on average over all possible values of the message $\mathbf{v}$, the distribution induced by coding remains identical regardless of the values of $\mathbf{v}$. This stronger result will prove useful in Section~\ref{sec:chaining-over-b} when we integrate the code in a chained construction.
\begin{proof}
	We denote $Q_{\text{PPM}}^m$ by $Q_{\widetilde{Z}}$ and $(Q_{\text{PPM}}^{m})^{\proddist \ell}$ by $Q_{\mathbf{\widetilde{Z}}}$. We have
	\begin{align}
	Q_{\widetilde{Z}}(\widetilde{z}) &= \frac{1}{2^q}Q^{2^q}_0(\widetilde{z})\sum_{j=1}^{2^q}\frac{Q_1(z_j)}{Q_0(z_j)},\displaybreak[0]\\
	Q_{\mathbf{\widetilde{Z}}}(\widetilde{\mathbf{z}}) &= \frac{1}{2^{q\ell}}\prod_{i=1}^{\ell}Q^{2^q}_0(\widetilde{z}_i)\sum_{j=1}^{2^q}\frac{Q_1(z_{i,j})}{Q_0(z_{i,j})},\displaybreak[0]\\
	Q_{\widetilde{\mathbf{Z}}|\mathbf{X}_\calS}(\widetilde{\mathbf{z}}|\mathbf{x}_\calS) &= \frac{1}{2^{|\calS^c|\ell}}\prod_{i=1}^{\ell}Q^{2^q}_0(\widetilde{z}_i)\sum_{j\in\calA(x_{\calS,i})}\frac{Q_1(z_{i,j})}{Q_0(z_{i,j})}.
	\end{align}
	We define a typical set as follows:
	\begin{align}
	\mathcal{T}_{\epsilon}^\ell(X_{1:q},\widetilde{Z}) \eqdef \left\{ \left(\mathbf{x}_1,\dots,\mathbf{x}_q,\mathbf{\widetilde{z}} \right)\in \mathcal{X}^\ell \times \cdots \times \mathcal{X}^\ell \times \widetilde{\mathcal{Z}}^\ell : \frac{1}{\ell}\log\frac{Q\left( \mathbf{\widetilde{z}} | \mathbf{x}_\calS \right) }{Q\left(\mathbf{\widetilde{z}}\right)} < I(X_\calS;\widetilde{Z})+\frac{\epsilon}{q}, ~\forall \mathcal{S} \subseteq \intseq{1}{q} \right\}.
	\end{align}
	We denote the expectation over all the random variables except $\mathbf{X}_{1:q}(\svbu,\svbv,\svbk)$ by $\mathbb{E}_{\sim{\mathbf{X}_{1:q}(\svbu,\svbv,\svbk)}}$. Then, the expectation of relative entropy over the codebook distribution is
	\begin{align}
	&\mathbb{E}\left[\D{P_{\widetilde{\mathbf{Z}}|\mathbf{V} = \svbv}}{Q_{\mathbf{\widetilde{Z}}}}\right] = \mathbb{E} \left[\sum_{\mathbf{\widetilde{z}}}\frac{1}{M_UM_K} \sum_{\svbu,\svbk}\widetilde{W}_{\widetilde{Z}|X_{1:q}}(\mathbf{\widetilde{z}}|\mathbf{X}_{1:q}(\svbu,\svbv,\svbk)) \vphantom{\log\left(\frac{\sum_{\svbu',\svbk'}\widetilde{W}_{\widetilde{Z}|X_{1:q}}\left(\mathbf{\widetilde{z}}|\mathbf{X}_{1:q}(\svbu',\svbv,\svbk')\right)}{MQ_{\mathbf{\widetilde{Z}}}(\mathbf{\widetilde{z}})}\right)}\right.\nonumber\\ &\qquad\qquad\qquad\qquad\qquad\qquad\qquad\qquad\qquad\qquad\qquad\left.\log\left(\frac{\sum_{\svbu',\svbk'}\widetilde{W}_{\widetilde{Z}|X_{1:q}}\left(\mathbf{\widetilde{z}}|\mathbf{X}_{1:q}(\svbu',\svbv,\svbk')\right)}{M_UM_KQ_{\mathbf{\widetilde{Z}}}(\mathbf{\widetilde{z}})}\right)\right]\displaybreak[0]\\
	&\qquad\qquad\qquad= \sum_{\mathbf{\widetilde{z}}}\frac{1}{M_UM_K} \sum_{\svbu,\svbk} \mathbb{E}_{\mathbf{X}_{1:q}(\svbu,\svbv,\svbk)}\left[\widetilde{W}_{\widetilde{Z}|X_{1:q}}(\mathbf{\widetilde{z}}|\mathbf{X}_{1:q}(\svbu,\svbv,\svbk))\vphantom{\mathbb{E}_{\sim\mathbf{X}_{1:q}(\svbu,\svbv,\svbk)}\left[\log\left(\frac{\sum_{\svbu',\svbk'}\widetilde{W}_{\widetilde{Z}|X_{1:q}}(\mathbf{\widetilde{z}}|\mathbf{X}_{1:q}(\svbu',\svbv,\svbk'))}{MQ_{\mathbf{\widetilde{Z}}}(\mathbf{\widetilde{z}})}\right)\right]}\right.\nonumber\\ &\qquad\qquad\qquad\qquad\qquad\qquad\qquad\qquad\left.\mathbb{E}_{\sim\mathbf{X}_{1:q}(\svbu,\svbv,\svbk)}\left[\log\left(\frac{\sum_{\svbu',\svbk'}\widetilde{W}_{\widetilde{Z}|X_{1:q}}(\mathbf{\widetilde{z}}|\mathbf{X}_{1:q}(\svbu',\svbv,\svbk'))}{M_UM_KQ_{\mathbf{\widetilde{Z}}}(\mathbf{\widetilde{z}})}\right)\right]\right]\displaybreak[0]\\
	&\qquad\qquad\qquad\overset{(a)}{\leq} \sum_{\mathbf{\widetilde{z}}}\frac{1}{M_UM_K} \sum_{\svbu,\svbk} \mathbb{E}_{\mathbf{X}_{1:q}(\svbu,\svbv,\svbk)}\left[\widetilde{W}_{\widetilde{Z}|X_{1:q}}(\mathbf{\widetilde{z}}|\mathbf{X}_{1:q}(\svbu,\svbv,\svbk)) \vphantom{\log\left(\frac{\sum_{\svbu',\svbk'}\mathbb{E}_{\sim\mathbf{X}_{1:q}(\svbu,\svbv,\svbk)}\left[\widetilde{W}_{\widetilde{Z}|X_{1:q}}(\mathbf{\widetilde{z}}|\mathbf{X}_{1:q}(\svbu',\svbv,\svbk'))\right]}{MQ_{\mathbf{\widetilde{Z}}}(\mathbf{\widetilde{z}})}\right)}\right.\nonumber\\ &\qquad\qquad\qquad\qquad\qquad\qquad\qquad\qquad\left.\log\left(\frac{\sum_{\svbu',\svbk'}\mathbb{E}_{\sim\mathbf{X}_{1:q}(\svbu,\svbv,\svbk)}\left[\widetilde{W}_{\widetilde{Z}|X_{1:q}}(\mathbf{\widetilde{z}}|\mathbf{X}_{1:q}(\svbu',\svbv,\svbk'))\right]}{M_UM_KQ_{\mathbf{\widetilde{Z}}}(\mathbf{\widetilde{z}})}\right)\right]\displaybreak[0]\\
	&\qquad\qquad\qquad\overset{(b)}{=} \sum_{\mathbf{\widetilde{z}}}\frac{1}{M_UM_K} \sum_{\svbu,\svbk} \mathbb{E}_{\mathbf{X}_{1:q}(\svbu,\svbv,\svbk)}\left[\widetilde{W}_{\widetilde{Z}|X_{1:q}}(\mathbf{\widetilde{z}}|\mathbf{X}_{1:q}(\svbu,\svbv,\svbk)) \vphantom{\log\left(\frac{\sum_{\calS\subseteq\intseq{1}{q}}\prod_{i\in\calS}(M_{U,i}M_{K,i}-1)Q_{\mathbf{\widetilde{Z}}|\mathbf{X}_{\calS^c}}(\mathbf{\widetilde{z}}|\mathbf{X}_{\calS^c}(\svbu,\svbv,\svbk))}{MQ_{\mathbf{\widetilde{Z}}}(\mathbf{\widetilde{z}})}\right)}\right.\nonumber\\ &\qquad\qquad\qquad\qquad\qquad\qquad\qquad\qquad\left.\log\left(\frac{\sum_{\calS\subseteq\intseq{1}{q}}\prod_{i\in\calS}(M_{U,i}M_{K,i}-1)Q_{\mathbf{\widetilde{Z}}|\mathbf{X}_{\calS^c}}(\mathbf{\widetilde{z}}|\mathbf{X}_{\calS^c}(\svbu,\svbv,\svbk))}{M_UM_KQ_{\mathbf{\widetilde{Z}}}(\mathbf{\widetilde{z}})}\right)\right]\displaybreak[0]\\
	&\qquad\qquad\qquad\leq \sum_{\mathbf{\widetilde{z}}}\frac{1}{M_UM_K} \sum_{\svbu,\svbk} \mathbb{E}_{\mathbf{X}_{1:q}(\svbu,\svbv,\svbk)}\left[\widetilde{W}_{\widetilde{Z}|X_{1:q}}(\mathbf{\widetilde{z}}|\mathbf{X}_{1:q}(\svbu,\svbv,\svbk)) \vphantom{\log\left(1+\frac{\sum_{\calS\subset\intseq{1}{q}}\prod_{i\in\calS}(M_{U,i}M_{K,i}-1)Q_{\mathbf{\widetilde{Z}}|\mathbf{X}_{\calS^c}}(\mathbf{\widetilde{z}}|\mathbf{X}_{\calS^c}(\svbu,\svbv,\svbk))}{MQ_{\mathbf{\widetilde{Z}}}(\mathbf{\widetilde{z}})}\right)}\right.\nonumber\\ &\qquad\qquad\qquad\qquad\qquad\qquad\qquad\left.\log\left(1+\frac{\sum_{\calS\subset\intseq{1}{q}}\prod_{i\in\calS}(M_{U,i}M_{K,i}-1)Q_{\mathbf{\widetilde{Z}}|\mathbf{X}_{\calS^c}}(\mathbf{\widetilde{z}}|\mathbf{X}_{\calS^c}(\svbu,\svbv,\svbk))}{M_UM_KQ_{\mathbf{\widetilde{Z}}}(\mathbf{\widetilde{z}})}\right)\right]\displaybreak[0]\\
	&\qquad\qquad\qquad= \frac{1}{M_UM_K} \sum_{\svbu,\svbk}\sum_{\mathbf{\widetilde{z}}} \sum_{\mathbf{x}_{1:q}} \frac{1}{2^{q\ell}} \widetilde{W}_{\widetilde{Z}|X_{1:q}}(\mathbf{\widetilde{z}}|\mathbf{x}_{1:q})\nonumber\\ &\qquad\qquad\qquad\qquad\qquad\qquad\qquad\qquad\log\left(1+\frac{\sum_{\calS\subset\intseq{1}{q}}\prod_{i\in\calS}(M_{U,i}M_{K,i}-1)Q_{\mathbf{\widetilde{Z}}|\mathbf{X}_{\calS^c}}(\mathbf{\widetilde{z}}|\mathbf{x}_{\calS^c})}{M_UM_KQ_{\mathbf{\widetilde{Z}}}(\mathbf{\widetilde{z}})}\right)\displaybreak[0]\\
	&\qquad\qquad\qquad\overset{(c)}{\leq}  \sum_{(\mathbf{x}_{1:q},\mathbf{\widetilde{z}})\in\calT_{\epsilon}^{\ell}(X_{1:q},\widetilde{Z})}\frac{1}{2^{q\ell}}\widetilde{W}_{\widetilde{Z}|X_{1:q}}(\mathbf{\widetilde{z}}|\mathbf{x}_{1:q})\log\left(1+\sum_{\calS\subset\intseq{1}{q}}2^{-\ell(\sum_{i\in\calS^c}(R_{U,i}+R_{K,i}) - I(X_{\calS^c};\widetilde{Z})-\epsilon/q)}\right)\nonumber\\ &\qquad\qquad\qquad\quad\qquad\qquad\qquad\qquad\qquad +\sum_{(\mathbf{x}_{1:q},\mathbf{\widetilde{z}})\notin\calT_{\epsilon}^{\ell}(X_{1:q},\widetilde{Z})}\frac{1}{2^{q\ell}}\widetilde{W}_{\widetilde{Z}|X_{1:q}}(\mathbf{\widetilde{z}}|\mathbf{x}_{1:q})\log\left(1+2^q\mu_z^{-m\ell}\right)\displaybreak[0]\\
	&\qquad\qquad\qquad\overset{(d)}{\leq}  \sum_{(\mathbf{x}_{1:q},\mathbf{\widetilde{z}})\in\calT_{\epsilon}^{\ell}(X_{1:q},\widetilde{Z})}\frac{1}{2^{q\ell}}\widetilde{W}_{\widetilde{Z}|X_{1:q}}(\mathbf{\widetilde{z}}|\mathbf{x}_{1:q})\sum_{\calS\subset\intseq{1}{q}}2^{-\ell(\sum_{i\in\calS^c}(R_{U,i}+R_{K,i}) - I(X_{\calS^c};\widetilde{Z})-\epsilon/q)}\vphantom{}\nonumber\\ &\qquad\qquad\qquad\quad\qquad\qquad\qquad\qquad\qquad\qquad\qquad+\sum_{(\mathbf{x}_{1:q},\mathbf{\widetilde{z}})\notin\calT_{\epsilon}^{\ell}(X_{1:q},\widetilde{Z})}\frac{1}{2^{q\ell}}\widetilde{W}_{\widetilde{Z}|X_{1:q}}(\mathbf{\widetilde{z}}|\mathbf{x}_{1:q})\frac{m\ell 2^{q/m\ell}}{\mu_z}\displaybreak[0]\\
	&\qquad\qquad\qquad\leq \sum_{\calS\subset\intseq{1}{q}}2^{-\ell(\sum_{i\in\calS}(R_{U,i}+R_{K,i}) - I(X_{\calS};\widetilde{Z})-\epsilon/q)} + \frac{m\ell 2^{q/m\ell}}{\mu_z} \P{(\mathbf{X}_{1:q},\mathbf{\widetilde{Z}})\notin\calT_{\epsilon}^{\ell}(X_{1:q},\widetilde{Z})}\\
	&\qquad\qquad\qquad\overset{(e)}{\leq} m2^{-\ell(\varepsilon-\epsilon)/q}  + \frac{m\ell 2^{q/m\ell}}{\mu_z} \P{(\mathbf{X}_{1:q},\mathbf{\widetilde{Z}})\notin\calT_{\epsilon}^{\ell}(X_{1:q},\widetilde{Z})},
	\end{align}
	where (a) follows from Jensen's inequality, (b) follows from expressing the summation over all $(\svbu',\svbk')$ as a summation over $\calS\in\intseq{1}{q}$ and  $(\svbu',\svbk')\in\{(\svbu'',\svbk''):(u_j'',k_j'')=(u_j,k_j) \forall j\in\calS^c\text{ and }(u_j'',k_j'')=(u_j,k_j)\forall j\in\calS\}$ and noting that $\mathbb{E}_{\sim\mathbf{X}_{1:q}(\svbu,\svbv,\svbk)}\left[\widetilde{W}_{\widetilde{Z}|X_{1:q}}(\mathbf{\widetilde{z}}|\mathbf{X}_{1:q}(\svbu',\svbv,\svbk'))\right]=Q_{\mathbf{\widetilde{Z}}|\mathbf{X}_{\calS^c}}(\mathbf{\widetilde{z}}|\mathbf{X}_{\calS^c}(\svbu,\svbv,\svbk))$ for all $(\svbu',\svbk')$ in that set, (c) follows because for $(\mathbf{x}_{1:q},\mathbf{\widetilde{z}})\in\calT_{\epsilon}^{\ell}(X_{1:q},\widetilde{Z})$ and $\calS \subseteq \intseq{1}{q}$, $\frac{Q_{\mathbf{\widetilde{Z}}|\mathbf{X}_{\calS^c}}(\mathbf{\widetilde{z}}|\mathbf{x}_{\calS^c})}{Q_{\mathbf{\widetilde{Z}}}(\mathbf{\widetilde{z}})} \leq 2^{\ell(I(X_{\calS^c};\widetilde{Z})+\epsilon/q)}$, (d) follows because $\log(1+x^m)\leq mx$ for $x\geq0$, and (e) follows because 
	$\sum_{i\in\calS}(R_{U,i}+R_{K,i}) \geq  I(X_{\calS};\widetilde{Z})+\varepsilon/q.$
	
	Using Hoeffding's inequality as in the previous section, we can show that
	\begin{align}
	\P{(\mathbf{X}_{1:q},\mathbf{\widetilde{Z}})\notin\calT_{\epsilon}^{\ell}(X_{1:q},\widetilde{Z})} \leq me^{-\frac{2\ell\epsilon^2}{q^4}}.
	\end{align}
	Therefore, for $0<\epsilon<\varepsilon/2$, we have
	\begin{align}
	\mathbb{E}\left[\D{P_{\widetilde{\mathbf{Z}}|\mathbf{V} = \svbv}}{Q_{\mathbf{\widetilde{Z}}}}\right] &\leq m2^{-\ell\epsilon/q}  + \frac{m^2\ell 2^{q/m\ell}}{\mu_z} e^{-\frac{2\ell\epsilon^2}{q^4}}.
	\end{align}
\end{proof}

\subsubsection{Identifying a specific code}
We now show that there exists a specific code that has small probability of error and divergence between the induced distribution and the innocent distribution close to the desired value. We denote the probability measure \ac{wrt} codebook distribution by $\P[{\calC}]{\cdot}$. By using Markov's inequality, for $\alpha>0$ and $\beta>0$, we have
\begin{align}
&\P[{\calC}]{\P{\widehat{\mathbf{W}}\neq\mathbf{W}}<\alpha\delta_1,\sum_{\mathbf{v}} \frac{1}{M_{V}} \D{P_{\widetilde{\mathbf{Z}}^{(j)}|\mathbf{V}^{(j)}=\mathbf{v}}}{\left(Q_{\text{PPM}}^m\right)^{\proddist\ell}} < \beta\delta_2}\nonumber\\ &\qquad\qquad\qquad\qquad= 1 - \P[{\calC}]{\P{\widehat{\mathbf{W}}\neq\mathbf{W}}\geq\alpha\delta_1 \text{  or  }\sum_{\mathbf{v}} \frac{1}{M_{V}} \D{P_{\widetilde{\mathbf{Z}}^{(j)}|\mathbf{V}^{(j)}=\mathbf{v}}}{\left(Q_{\text{PPM}}^m\right)^{\proddist\ell}} \geq \beta\delta_2}\displaybreak[0]\\
&\qquad\qquad\qquad\qquad \geq 1-\P[{\calC}]{\P{\widehat{\mathbf{W}}\neq\mathbf{W}}\geq\alpha\delta_1}-\P[{\calC}]{\sum_{\mathbf{v}} \frac{1}{M_{V}} \D{P_{\widetilde{\mathbf{Z}}^{(j)}|\mathbf{V}^{(j)}=\mathbf{v}}}{\left(Q_{\text{PPM}}^m\right)^{\proddist\ell}} \geq \beta\delta_2}\\
&\qquad\qquad\qquad\qquad \geq 1-\frac{1}{\alpha}-\frac{1}{\beta}.
\end{align}
By choosing $\alpha$ and $\beta$ such that $\frac{1}{\alpha}+\frac{1}{\beta}<1$, we have the above probability positive, which means there exist codes that satisfy 
\begin{align}
&\P{\widehat{\mathbf{W}}\neq\mathbf{W}}<\alpha\delta_1,\\
\sum_{\mathbf{v}} \frac{1}{M_{V}} &\D{P_{\widetilde{\mathbf{Z}}^{(j)}|\mathbf{V}^{(j)}=\mathbf{v}}}{\left(Q_{\text{PPM}}^m\right)^{\proddist\ell}} < \beta\delta_2.\label{eq:secret_message}
\end{align}
Moreover, because of the convexity of divergence,
\begin{align}
\D{P_{\mathbf{\widetilde{Z}}}}{\left(Q_{\text{PPM}}^m\right)^{\proddist\ell}} \leq \sum_{\mathbf{v}} \frac{1}{M_{V}} \D{P_{\widetilde{\mathbf{Z}}^{(j)}|\mathbf{V}^{(j)} = \mathbf{v}}}{\left(Q_{\text{PPM}}^m\right)^{\proddist\ell}} <\beta\delta_2.
\end{align}
This implies that the first two terms of \eqref{eq:resolve_one_block} goes to zero exponentially in $\ell$.
Using \cite[Lemma 1]{Bloch2017}, we have
\begin{align}
\D{(Q_{\text{PPM}}^{m})^{\proddist \ell} }{Q_0^{\proddist m\ell}} &\leq \frac{\ell}{m}\left(\frac{\chi_2(Q_1\Vert Q_0)}{2}+\calO(\frac{1}{m})\right).\label{eq:divergence_PPM}
\end{align}
By choosing $\ell = \lfloor\frac{2m\delta}{B\chi_2(Q_1\Vert Q_0)}\rfloor$,
\begin{align}
\D{(Q_{\text{PPM}}^{m})^{\proddist \ell} }{Q_0^{\proddist m\ell}} &\leq \frac{\delta}{B}+\calO(\frac{1}{Bm}).\label{eq:ppm_divergence}
\end{align} 
Hence, we have
\begin{align}
\avgD{P_{\mathbf{\widetilde{Z}}}}{Q_0^{\proddist m\ell}} \leq \frac{\delta}{B}+\calO(\frac{1}{Bm}).
\end{align}
From \eqref{eq:secret_message}, we show the secrecy of $\mathbf{V}$ as follows.
\begin{align}
I\left({\widetilde{\mathbf{Z}};\mathbf{V}}\right) &=\D{P_{\widetilde{\mathbf{Z}}\mathbf{V}}}{P_{\widetilde{\mathbf{Z}}}P_{\mathbf{V}}}\displaybreak[0]\label{eq:secrecy_v1}\\
&= \sum_{\mathbf{v}} \frac{1}{M_{V}} \D{P_{\widetilde{\mathbf{Z}}|\mathbf{V} = \mathbf{v}}}{\left(Q_{\text{PPM}}^m\right)^{\proddist\ell}} - \D{P_{\widetilde{\mathbf{Z}}}}{\left(Q_{\text{PPM}}^m\right)^{\proddist\ell}}\\
&\leq\sum_{\mathbf{v}} \frac{1}{M_{V}} \D{P_{\widetilde{\mathbf{Z}}|\mathbf{V} = \mathbf{v}}}{\left(Q_{\text{PPM}}^m\right)^{\proddist\ell}} \\
&{\leq} \beta\delta_2\label{eq:secrecy_v2}.
\end{align}

\subsubsection{\ac{MSD} operating point for \ac{MLC}}
We know from \cite{Bloch2017} that the capacity of the \ac{PPM} super channel converges to  $\avgD{P_1}{P_0}$ as the \ac{PPM} order $m$ tends to infinity. Since the \ac{PPM} mapper is a one-to-one map between $X_{1:q}$ and the PPM symbol $\widetilde{X} = \widetilde{x}(X_{1:q})$, we have the following decomposition of mutual information between the input and the output of the PPM super channel $\widetilde{W}_{\widetilde{Y}|\widetilde{X}}$.
\begin{align}
I(\widetilde{X};\widetilde{Y}) &= I(X_{1:q};\widetilde{Y})
= \sum_{i=1}^{q} I(X_i;\widetilde{Y}|X_{i+1:q}). \label{eq:ppm_capacity_sum_levels}
\end{align} 
The mutual information term $I(X_i;\widetilde{Y}|X_{i+1:q})$ represents the average mutual information between the input of the $i$-th level and $\widetilde{Y}$ for a given input to the levels $i+1$ to $q$. Since, $X_i$ is independent of $X_{i+1:q}$, we have
\begin{align}
I(X_i;\widetilde{Y}\vert X_{i+1:q}) &= I(X_i;\widetilde{Y},X_{i+1:q}).
\end{align}
This suggests that we can consider the $i$-th level as a channel with input $X_i$ and output $(\widetilde{Y},X_{i+1:q})$ and decode all levels successively in descending order from $q$ to $1$ using \ac{MSD}; specifically, we decode the $q$-th level first, then the $(q-1)$-th level by treating the decoded bits from $q$-th level as a side information for the channel defining $(q-1)$-th level, and so on. From \cite{Wachsmann1999}, we know that the capacity of the PPM super channel could be achieved using \ac{MSD} by choosing the corresponding terms in the summation of \eqref{eq:ppm_capacity_sum_levels} as the rates for each level. In our problem, we show that by choosing the rates for $i$-th level as follows, we can achieve the optimal rate.
\begin{align}
R_{U,i} &= \min \left(I(X_i;\widetilde{Y}|X_{i+1:q})-\varepsilon/q,I(X_i;\widetilde{Z}|X_{i+1:q})+\varepsilon/q \right), \label{eq:mlc_rates1}\\
R_{V,i} &= \max \left( 0, I(X_i;\widetilde{Y}|X_{i+1:q}) - I(X_i;\widetilde{Z}|X_{i+1:q}) -2\varepsilon/q \right),\label{eq:mlc_rates2}\\
R_{K,i} &= \max \left( 0, I(X_i;\widetilde{Z}|X_{i+1:q}) - I(X_i;\widetilde{Y}|X_{i+1:q})+2\varepsilon/q \right).
\label{eq:mlc_rates3}
\end{align}
The sum in \eqref{eq:ppm_capacity_sum_levels} converges to $\avgD{P_1}{P_0}$ as $q$ tends to infinity. This suggests that we can achieve rates arbitrarily close to $\avgD{P_1}{P_0}$ by choosing the number of levels in \ac{MLC} large enough and rates given in (\ref{eq:mlc_rates1}-\ref{eq:mlc_rates3}) for each level. 

We now show that the rates in (\ref{eq:mlc_rates1}-\ref{eq:mlc_rates3}) satisfy the rate requirements in \eqref{eq:reliability_rate_conditions_1} and \eqref{eq:resolve_rate_conditions_1}. For $\calS\subseteq\intseq{1}{q}$, we have
\begin{align}
\sum_{i\in\calS} (R_{U,i} + R_{V,i}) &= \sum_{i\in\calS} \left(I(X_i;\widetilde{Y}|X_{i+1:q})-\varepsilon/q\right) \displaybreak[0]\\
&= \sum_{i\in {\mathcal{S}}} I(X_i;\widetilde{Y}|X_{\intseq{i+1}{q}\cap\mathcal{S}},X_{\intseq{i+1}{q}\cap\mathcal{S}^c}) - |\calS|\varepsilon/q\displaybreak[0]\\
&= \sum_{i\in {\mathcal{S}}} \left( H(X_i|X_{\intseq{i+1}{q}\cap\mathcal{S}},X_{\intseq{i+1}{q}\cap\mathcal{S}^c}) - H(X_i|\widetilde{Y},X_{\intseq{i+1}{q}\cap\mathcal{S}},X_{\intseq{i+1}{q}\cap\mathcal{S}^c})\right)- |\calS|\varepsilon/q\displaybreak[0]\\
&\overset{(a)}{\leq} \sum_{i\in {\mathcal{S}}} \left( H(X_i|X_{\intseq{i+1}{q}\cap\mathcal{S}},X_{\mathcal{S}^c}) - H(X_i|\widetilde{Y},X_{\intseq{i+1}{q}\cap\mathcal{S}},X_{\mathcal{S}^c})\right)- |\calS|\varepsilon/q\displaybreak[0]\\
&= I(X_{\mathcal{S}}, \widetilde{Y}| X_{\mathcal{S}^c})- |\calS|\varepsilon/q,
\end{align}
where (a) follows from the fact that conditioning reduces entropy and the $X_i$'s are independent of each other.
Hence, the rates satisfy the constraints for reliability.
Similarly,
\begin{align}
\sum_{i\in {\mathcal{S}}} (R_{U,i} + R_{K,i}) &=  \sum_{i\in {\mathcal{S}}} \left(I(X_i;\widetilde{Z}|X_{i+1:q}) +\varepsilon/q\right) \displaybreak[0]\\
&= \sum_{i\in {\mathcal{S}}} I(X_i; \widetilde{Z}| \{X_j\}_{j\in\intseq{i+1}{q}\cap {\mathcal{S}}}, \{X_j\}_{j\in\intseq{i+1}{q}\backslash {\mathcal{S}}} ) + |\calS|\varepsilon/q\displaybreak[0]\\
&= \sum_{i\in {\mathcal{S}}} I(X_i; \widetilde{Z},\{X_j\}_{j\in\intseq{i+1}{q}\backslash {\mathcal{S}}}| \{X_j\}_{j\in\intseq{i+1}{q}\cap {\mathcal{S}}} ) \nonumber\\
&\qquad\qquad\qquad\qquad- I(X_i;\{X_j\}_{j\in\intseq{i+1}{q}\backslash {\mathcal{S}}} | \{X_j\}_{j\in\intseq{i+1}{q}\cap {\mathcal{S}}})+ |\calS|\varepsilon/q\displaybreak[0]\\
&\overset{(a)}{=} \sum_{i\in {\mathcal{S}}} I(X_i; \widetilde{Z},\{X_j\}_{j\in\intseq{i+1}{q}\backslash {\mathcal{S}}}| \{X_j\}_{j\in\intseq{i+1}{q}\cap {\mathcal{S}}} )+ |\calS|\varepsilon/q \displaybreak[0]\\
&\geq \sum_{i\in {\mathcal{S}}} I(X_i; \widetilde{Z}| \{X_j\}_{j\in\intseq{i+1}{q}\cap {\mathcal{S}}} )+ |\calS|\varepsilon/q\displaybreak[0]\\
&= I(X_{\mathcal{S}};\widetilde{Z})+ |\calS|\varepsilon/q,
\end{align}
where (a) follows from the independence of $\{X_i\}$. This shows that the rates satisfy the resolvability rate constraints.

We now compute the covert message and key throughputs. The number of bits transmitted is given by
\begin{align}
\log{M_U+\log{M_{V}}} &= \ell\sum_{i=1}^q (R_{U,i}+R_{V,i})\\
&=\ell\left(I(X_{1:q};\widetilde{Y})-\varepsilon\right)\\
&=\ell\left(I(\widetilde{X};\widetilde{Y}) - \varepsilon\right)\\
&\geq \ell \left(\D{P_1}{P_0} - \frac{1}{m}\chi_2(P_1\Vert P_0)-\varepsilon\right),
\end{align} 
where the last inequality follows from~\cite[Lemma 2]{Bloch2017}.
The covert message throughput is given by
\begin{align}
\frac{\log{M_U+\log{M_{V}}}}{\sqrt{\ell m\delta/B}} &\geq \sqrt{\frac{B\ell}{m\delta}}\left(\D{P_1}{P_0} - \frac{1}{m}\chi_2(P_1\Vert P_0)-\varepsilon\right)\\
&=\sqrt{\frac{2}{\chi_2(Q_1\Vert Q_0)}}\left(\D{P_1}{P_0} - \frac{1}{m}\chi_2(P_1\Vert P_0)-\varepsilon\right).
\end{align}
The number of key bits used is given by
\begin{align}
\log{M_K} &= \ell\sum_{i=1}^{q}R_{K,i} \\
&\leq\ell\sum_{i=1}^{q} \max \left( 0, I(X_i;\widetilde{Z}|X_{i+1:q}) - I(X_i;\widetilde{Y}|X_{i+1:q})+2\varepsilon/q \right),
\end{align}
and the covert key throughput is given by
\begin{align}
\frac{\log{M_K}}{\sqrt{\ell m\delta/B}} &\leq \sqrt{\frac{2}{\chi_2(Q_1\Vert Q_0)}}\left(\sum_{i=1}^{q} \max \left( 0, I(X_i;\widetilde{Z}|X_{i+1:q}) - I(X_i;\widetilde{Y}|X_{i+1:q})+2\varepsilon/q \right)\right).
\end{align}

We now show using an example that the above key rate is not optimal. Let Bob's channel be a \ac{BSC} with probability of flipping $P_0(1)=P_1(0)=0.2$ and Willie's channel be a \ac{BAC} with probability of flipping $Q_0(1)=0.1$ and $Q_1(0)=0.4$. The relative entropies for these channels are
$\D{P_1}{P_0} = 1.2$ and $\D{Q_1}{Q_0}= 1.083$. For this case, ideally we do not need any key to achieve covert communication, but from the computation of $I(X_i;\widetilde{Z}|X_{i+1:q}) - I(X_i;\widetilde{Y}|X_{i+1:q})$ for $q=10$ levels shown in Table \ref{table:mlc_rate_diff_per_levels}, we can conclude that we need some key bits. 

\begin{table}[h]
	\centering
	\caption{$I(X_i;\widetilde{Z}|X_{i+1:q}) - I(X_i;\widetilde{Y}|X_{i+1:q})$ computed for 10 levels for the given channels}
	\label{table:mlc_rate_diff_per_levels}
	\begin{tabular}{|c|c|c|c|c|c|c|c|c|c|}
		\hline
		1&2&3&4&5&6&7&8&9&10\\
		\hline
		-0.0905 & -0.0452 & -0.0084 &  0.0077 &  0.0084 & 0.0052 & 0.0028 & 0.0014 & 0.0007 & 0.0004 \\
		\hline
	\end{tabular}
\end{table}

Note that for the first 3 levels, $I(X_i;\widetilde{Z}|X_{i+1:q}) - I(X_i;\widetilde{Y}|X_{i+1:q})$ is negative. This suggests that we can send some secret bits over the first 3 levels as $R_{V,i}$ is positive for those levels. Moreover, the absolute value of sum of negative terms is greater than the sum of positive terms. This means that we can use chaining so that the secret messages from one block can be used as the keys for the next block.

\subsubsection{Chaining over B blocks}
\label{sec:chaining-over-b}
We now show that modifying the coding scheme by chaining over $B$ blocks such that the secret messages from $i$-th block $\{V_j^{(i)}\}_{j\in\intseq{1}{q}}$ are used as the keys for $(i+1)$-th block $\{K_j^{(i+1)}\}_{j\in\intseq{1}{q}}$, we can achieve the optimal throughputs. We can aggregate the secret messages from the levels that support secret messages and distribute them across the levels for which we need a key. We first show that the probability of error goes to zero asymptotically.
To bound the probability of error for $B$ blocks with chaining, define $\calE^{(i)} = \{\widehat{\mathbf{W}}^{(i)}\neq\mathbf{W}^{(i)}\}$ and $\calE^{(1:B)} = \bigcup_{i=1}^B\calE^{(i)}$, where $\mathbf{W}^{(i)}$ is the transmitted message and $\widehat{\mathbf{W}}^{(i)}$ is the decoded message at the receiver for the $i$-th block. Following similar steps as in the proof of \cite[Lemma 8]{Freche2017}, we obtain
\begin{align}
\P{\calE^{(1:B)}} &\leq\sum_{i=1}^{B}\P{\calE^{(i)}\left\vert{\calE^{(i-1)}}^c\right.}\displaybreak[0]\\
&\leq B m\left( e^{-\frac{2\ell\epsilon^2}{(\log m)^2}} +  2^{-\ell\left(\varepsilon  -\epsilon\right)} \right).
\end{align}
This shows that the probability of error goes to zero asymptotically.

We now show that the relative entropy between the induced and innocent distributions is upper bounded by $\delta$ asymptotically. Let $P_{\widetilde{\mathbf{Z}}^{(1:B)}}$ be the distribution induced by the coding scheme for $B$ blocks. We want  $P_{\widetilde{\mathbf{Z}}^{(1:B)}}$ to be close to the innocent distribution $Q_0^{\proddist Bm\ell}$. We have
\begin{align}
\D{P_{\widetilde{\mathbf{Z}}^{(1:B)}}}{Q_0^{\proddist Bm\ell}} &= \sum_{j=1}^{B} \mathbb{E}_{\widetilde{\mathbf{Z}}^{(j+1:B)}}\left[\D{P_{\widetilde{\mathbf{Z}}^{(j)}|\widetilde{\mathbf{Z}}^{(j+1:B)}}}{Q_0^{\proddist m\ell}}\right]\\
&= \sum_{j=1}^{B} \left[ \D{P_{\widetilde{\mathbf{Z}}^{(j)}}}{Q_0^{\proddist m\ell}} + \mathbb{E}_{\widetilde{\mathbf{Z}}^{(j+1:B)}}\left[\D{P_{\widetilde{\mathbf{Z}}^{(j)}|\widetilde{\mathbf{Z}}^{(j+1:B)}}}{P_{\widetilde{\mathbf{Z}}^{(j)}}} \right]\right]\\
&= \sum_{j=1}^{B} \left[ \D{P_{\widetilde{\mathbf{Z}}^{(j)}}}{Q_0^{\proddist m\ell}} + I(\widetilde{\mathbf{Z}}^{(j)};\widetilde{\mathbf{Z}}^{(j+1:B)}) \right],
\end{align}
and
\begin{align}
I(\widetilde{\mathbf{Z}}^{(j)};\widetilde{\mathbf{Z}}^{(j+1:B)}) &\leq I(\widetilde{\mathbf{Z}}^{(j)};\widetilde{\mathbf{Z}}^{(j+1:B)},\mathbf{V}^{(j)})\\
&=I(\widetilde{\mathbf{Z}}^{(j)};\mathbf{V}^{(j)}) + I(\widetilde{\mathbf{Z}}^{(j)};\widetilde{\mathbf{Z}}^{(j+1:B)}|\mathbf{V}^{(j)})\\
&\overset{(a)}{=}I(\widetilde{\mathbf{Z}}^{(j)};\mathbf{V}^{(j)}),
\end{align}
where (a) is because of the Markov chain $\widetilde{\mathbf{Z}}^{(j)}\to\mathbf{V}^{(j)}\to\widetilde{\mathbf{Z}}^{(j+1:B)}$. By similar steps as in \eqref{eq:secrecy_v1}-\eqref{eq:secrecy_v2}, we can show that

\begin{align}
I\left({\widetilde{\mathbf{Z}}^{(j)};\mathbf{V}^{(j)}}\right) \leq \beta\delta_2.
\end{align}
Therefore,
\begin{align}
\D{P_{\widetilde{\mathbf{Z}}^{(1:B)}}}{Q_0^{\proddist Bm\ell}} &\leq \delta + \calO(\frac{1}{m}).
\end{align}

We now show that the covert message and key throughputs are close to the covert capacity.
The number of transmitted bits is given by
\begin{align}
B\left(\log{M_U+\log{M_{V}}}\right) 
&\geq B\ell \left(\D{P_1}{P_0} - \frac{1}{m}\chi_2(P_1\Vert P_0)-\varepsilon\right)
\end{align} 
Hence, the covert rate is given by
\begin{align}
\frac{B\left(\log{M_U+\log{M_{V}}}\right)}{\sqrt{B\ell m\delta}} &\geq \sqrt{\frac{B\ell}{m\delta}}\left(\D{P_1}{P_0} - \frac{1}{m}\chi_2(P_1\Vert P_0)-\varepsilon\right)\\
&=\sqrt{\frac{2}{\chi_2(Q_1\Vert Q_0)}}\left(\D{P_1}{P_0} - \frac{1}{m}\chi_2(P_1\Vert P_0)-\varepsilon\right).
\end{align}
The number of key bits used is given by
\begin{align}
\log{M_K}+(B&-1)(\log{M_K}-\log{M_V})^+ \nonumber\\&= \ell\sum_{i=1}^{q}R_{K,i} + (B-1)\ell \left(\sum_{i=1}^{q}(R_{K,i}-R_{V,i})\right)^+\\
&\leq\ell\sum_{i=1}^{q} I(X_i;\widetilde{Z}|X_{i+1:q}) + (B-1)\ell\left(\sum_{i=1}^{q}\left(I(X_i;\widetilde{Z}|X_{i+1:q})-I(X_i;\widetilde{Y}|X_{i+1:q})+2\varepsilon/q\right)\right)^+\\
&=\ell I(\widetilde{X};\widetilde{Z}) + (B-1)\ell\left(I(\widetilde{X};\widetilde{Z})-I(\widetilde{X};\widetilde{Y})+2\varepsilon\right)^+\\
&\leq \ell\D{Q_1}{Q_0} + (B-1)\ell\left(\D{Q_1}{Q_0}-\D{P_1}{P_0}+\frac{1}{m}\chi_2(P_1\Vert P_0)+2\varepsilon\right)^+
\end{align}
Therefore, the key throughput is
\begin{align}
\frac{\log{M_K}+(B-1)(\log{M_K}-\log{M_V})^+}{\sqrt{B\ell m\delta}} &\leq \sqrt{\frac{2}{\chi_2(Q_1\Vert Q_0)}}\left(\D{Q_1}{Q_0}-\D{P_1}{P_0}+\frac{1}{m}\chi_2(P_1\Vert P_0)+2\varepsilon\right)^+\nonumber\\
&\qquad\qquad\qquad\qquad\qquad\qquad + \sqrt{\frac{2}{\chi_2(Q_1\Vert Q_0)}}\frac{\D{Q_1}{Q_0}}{B},
\end{align}
and the last term vanishes for large $B$.

\subsubsection{Degraded case}\label{sec:degraded_case}
We now show that when Willie's channel is degraded \ac{wrt} Bob's channel, we do not require any chaining and any key. We show that using following proposition.

\begin{proposition}
	The \ac{MLC} rates 
	\begin{align}
	R_{U,i} = I(X_i;\widetilde{Y}|X_{i+1:q}) - \varepsilon/q,\quad R_{V,i} = 0,\quad R_{K,i} = 0\label{eq:mlc_rates_degraded}
	\end{align}
	satisfy both reliability and resolvability constraints when Willie's channel is degraded \ac{wrt} Bob's channel.
\end{proposition}
\begin{proof}
	Using the same arguments as for the general case, we can show that for any $\mathcal{S} \subseteq \intseq{1}{q}$,
	\begin{align}
	\sum_{i\in {\mathcal{S}}} (R_{U,i}+R_{V,i}) &\leq I(X_{\mathcal{S}}, \widetilde{Y}| X_{\mathcal{S}^c})- |\calS|\varepsilon/q.
	\end{align}
	Hence, the rates satisfy the constraints for reliability given in \eqref{eq:reliability_rate_conditions_1}.

	When Willie's channel is degraded \ac{wrt} Bob's channel, we have
	\begin{align}
	W_{Z|X}(z|x) &= \sum_y W_{Y|X}(y|x) W_{Z|Y}(z|y).
	\end{align}
	The transition probability of the super channel is given by
	\begin{align}
	\widetilde{W}_{\widetilde{Z}|\widetilde{X}}(\widetilde{z}|\widetilde{x}) &= \prod_{j=1}^m W(z_j|x_j)\\
	&= \prod_{j=1}^m \sum_{y_j} W(y_j|x_j) W(z_j|y_j)\\
	&= \sum_{\widetilde{y}\in \widetilde{\mathcal{Y}}} \prod_{j=1}^m W(y_j|x_j) W(z_j|y_j)\\
	&= \sum_{\widetilde{y}\in \widetilde{\mathcal{Y}}} \widetilde{W}(\widetilde{y}|\widetilde{x}) \widetilde{W}(\widetilde{z}|\widetilde{y}).
	\end{align}
	Hence, the super channel corresponding to Willie's channel is degraded \ac{wrt} the super channel corresponding to Bob's channel.
	
	The channel corresponding to the $i$-th level of \ac{MLC} for Willie's channel is given by:
	\begin{align}
	W_{\widetilde{Z},X_{i+1:q}|X_i} (\widetilde{z},x_{i+1:q}|x_i) &= \frac{\sum_{x_{1:i-1}} P_X^{\proddist q}(x_{1:q}) \widetilde{W}_{\widetilde{Z}|\widetilde{X}}(\widetilde{z}|\widetilde{x}(x_{1:q}))}{P_X(x_i)}\displaybreak[0]\label{eq:degraded_equi_ch1}\\
	&= \frac{\sum_{x_{1:i-1}} P_X^{\proddist q}(x_{1:q}) \sum_{\widetilde{y}\in \widetilde{\mathcal{Y}}} \widetilde{W}(\widetilde{y}|\widetilde{x}(x_{1:q})) \widetilde{W}(\widetilde{z}|\widetilde{y})} {P_X(x_i)}\displaybreak[0]\\
	&= \sum_{\widetilde{y}\in \widetilde{\mathcal{Y}}} \widetilde{W}(\widetilde{z}|\widetilde{y}) \frac{\sum_{x_{1:i-1}} P_X^{\proddist q}(x_{1:q}) \widetilde{W}(\widetilde{y}|\widetilde{x}(x_{1:q})) } {P_X(x_i)}\displaybreak[0]\\
	&= \sum_{\widetilde{y}\in \widetilde{\mathcal{Y}}} \widetilde{W}(\widetilde{z}|\widetilde{y}) W_{\widetilde{Y},X_{i+1:q}|X_i} (\widetilde{y},x_{i+1:q}|x_i) .\label{eq:degraded_equi_ch2}
	\end{align}
	This shows that the channels corresponding to each levels of \ac{MLC} for Willie's channel are degraded \ac{wrt} those of Bob's channel. Hence, we have
	\begin{align}
	I(X_i;\widetilde{Y}|X_{i+1:q})  = I(X_i;\widetilde{Y},X_{i+1:q}) > I(X_i;\widetilde{Z},X_{i+1:q})  = I(X_i;\widetilde{Z}|X_{i+1:q}).
	\end{align}
	We choose $\varepsilon$ such that $I(X_i;\widetilde{Y},X_{i+1:q})-\varepsilon/q > I(X_i;\widetilde{Z},X_{i+1:q})+\varepsilon/q$. Then, for every $\mathcal{S} \subseteq \intseq{1}{q}$, we can bound the sum-rates as follows:
	\begin{align}
	\sum_{i\in {\mathcal{S}}} (R_{U,i}+R_{K,i}) &= \sum_{i\in {\mathcal{S}}} \left(I(X_i;\widetilde{Y}|X_{i+1:q})-\frac{\varepsilon}{q}\right)\\
	&\geq \sum_{i\in {\mathcal{S}}} \left(I(X_i;\widetilde{Z}|X_{i+1:q})+\frac{\varepsilon}{q}\right)\displaybreak[0]\\
	&= \sum_{i\in {\mathcal{S}}} I(X_i; \widetilde{Z}| \{X_j\}_{j\in\intseq{i+1}{q}\cap {\mathcal{S}}}, \{X_j\}_{j\in\intseq{i+1}{q}\backslash {\mathcal{S}}} )+\frac{|\calS|\varepsilon}{q}\displaybreak[0]\\
	&= \sum_{i\in {\mathcal{S}}} \bigg(I(X_i; \widetilde{Z},\{X_j\}_{j\in\intseq{i+1}{q}\backslash {\mathcal{S}}}| \{X_j\}_{j\in\intseq{i+1}{q}\cap {\mathcal{S}}} ) \nonumber\\&\qquad\qquad\qquad\qquad- I(X_i;\{X_j\}_{j\in\intseq{i+1}{q}\backslash {\mathcal{S}}} | \{X_j\}_{j\in\intseq{i+1}{q}\cap {\mathcal{S}}})\bigg)+\frac{|\calS|\varepsilon}{q}\displaybreak[0]\\
	&\overset{(a)}{=} \sum_{i\in {\mathcal{S}}} I(X_i; \widetilde{Z},\{X_j\}_{j\in\intseq{i+1}{q}\backslash {\mathcal{S}}}| \{X_j\}_{j\in\intseq{i+1}{q}\cap {\mathcal{S}}} )+\frac{|\calS|\varepsilon}{q} \displaybreak[0]\\
	&\geq \sum_{i\in {\mathcal{S}}} I(X_i; \widetilde{Z}| \{X_j\}_{j\in\intseq{i+1}{q}\cap {\mathcal{S}}} )+\frac{|\calS|\varepsilon}{q}\\
	&= I(X_{\mathcal{S}};\widetilde{Z})+\frac{|\calS|\varepsilon}{q},
	\end{align}
	where (a) follows from the independence of $\{X_i\}$. This shows that the MLC rates in \eqref{eq:mlc_rates_degraded} satisfy the resolvability rate constraints given in~\eqref{eq:resolve_rate_conditions_1}.
\end{proof}

Since we are not using any key, we do not need chaining to achieve optimal rates. Fig.\ref{fig:rate_region_example} shows an illustration of the rate region of \ac{MLC} with two levels for a degraded case.
\begin{figure}[h]
	\centering
	\includegraphics[width=0.7\linewidth]{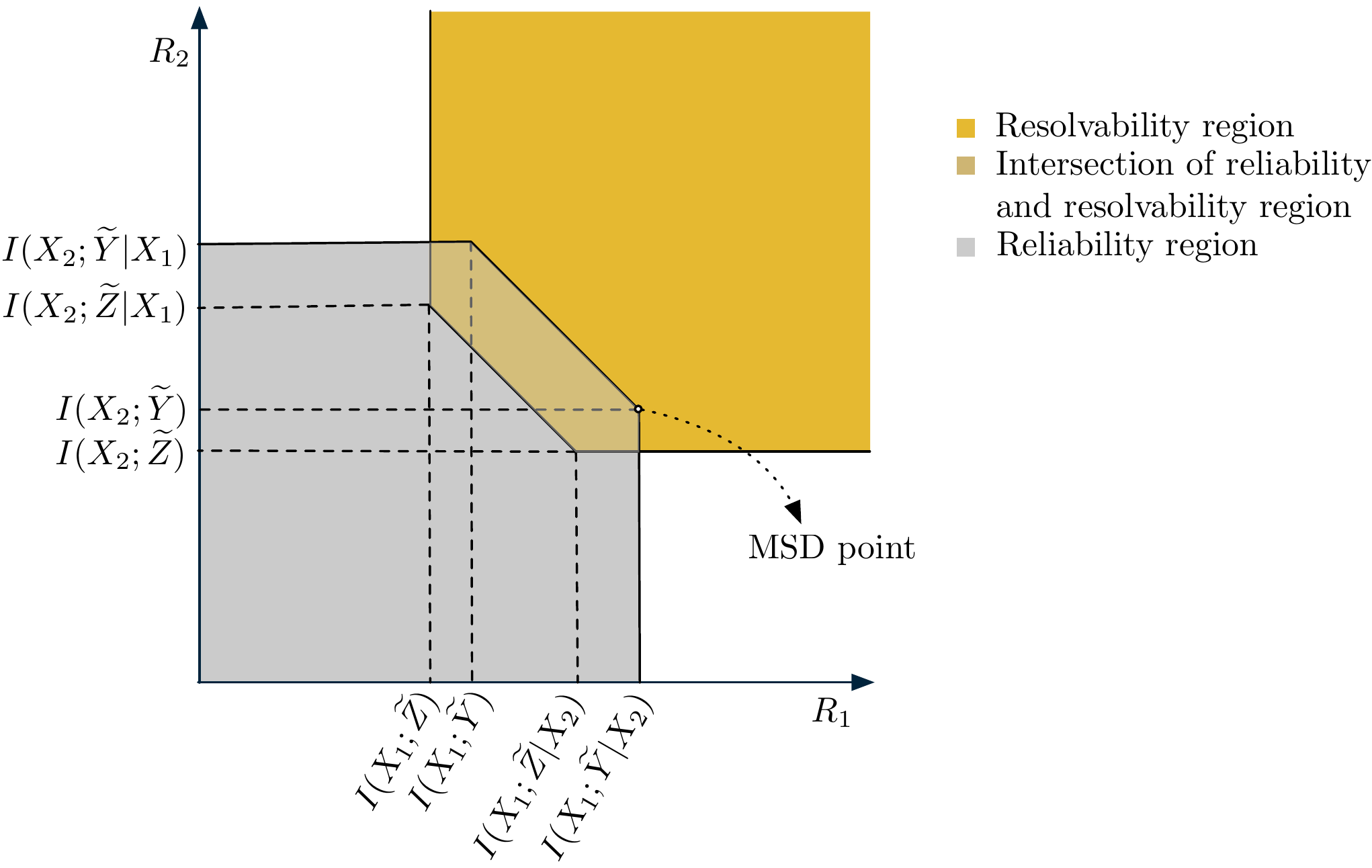}
	\caption{Rate region for MLC with two levels with MSD operating point}
	\label{fig:rate_region_example}
\end{figure}

\section{Towards practical codes}
\label{sec:to_prac_codes}
In this section, we discuss the design of practical codes for covert communication over a \ac{BI-DMC} using \ac{MLC}. The code design, in general, involves using keys and chaining over multiple blocks. To simplify the analysis, we discuss code design for the degraded case. We can generalize the code design for non-degraded cases using a chaining scheme as explained in the previous section. First, we show that the equivalent channel corresponding to a particular level when using \ac{MSD} remains unchanged when we alter the number of levels, which tremendously simplifies the code design. After establishing additional properties of the equivalent channels, we analyze the probability of error at Bob's decoder and covertness at Willie's receiver. Towards the end of this section, we will discuss an explicit low-complexity code construction using polar codes.
\subsection{Equivalent channel for each level}\label{sec:equivalent_channel}
We now prove that the channel corresponding to each level of \ac{MLC} with \ac{MSD} for uniformly distributed inputs can be represented by an equivalent channel that does not depend on the number of levels used in the scheme. In Section~\ref{sec:it_analysis_mlc}, we established that we can consider the $i$-th channel as a channel with input $X_i$ and output $(\widetilde{Y},X_{i+1:q})$. Because of the structure of PPM, for a known input to the levels $i+1$ to $q$, $x_{i+1:q}$, the position of the symbol ``1" in the PPM symbol can be narrowed down to $2^i$ positions. Using the notation defined earlier, $\calA^q(x_{i+1:q})$ denotes the set of indices of those positions. 
According to the PPM mapper defined in Section~\ref{sec:notation}, the position of the symbol ``1'' is in the first half of the \ac{PPM} symbol if $x_q = 0$ and in the second half otherwise. Similarly, if $x_i = 0$, the position of the symbol ``1'' is in the first half of the subset of indices determined by the values of $(x_{i+1},\dots,x_q)$ and second half otherwise. Table~\ref{table:ppm_mapper_illustration} shows an illustration of this mapping for $m=16$.
\begin{table}[h]
	\centering
	\caption{Illustration of PPM mapper for $m=16$}
	\label{table:ppm_mapper_illustration}
	\begin{tabular}{|c|c |c| c| c|c |c| c| c|c |c| c| c|c |c| c| c|}
		\hline
		PPM symbol index&1&2&3&4&5&6&7&8&9&10&11&12&13&14&15&16\\\hline
		$x_1$&0&1&0&1&0&1&0&1&0&1&0&1&0&1&0&1\\\hline
		$x_2$&\multicolumn{2}{|c|}{0}&\multicolumn{2}{c|}{1}&\multicolumn{2}{c|}{0}&\multicolumn{2}{c|}{1}&\multicolumn{2}{c|}{0}&\multicolumn{2}{c|}{1}&\multicolumn{2}{c|}{0}&\multicolumn{2}{c|}{1}\\\hline
		$x_3$&\multicolumn{4}{|c|}{0}&\multicolumn{4}{c|}{1}&\multicolumn{4}{c|}{0}&\multicolumn{4}{c|}{1}\\\hline
		$x_4$&\multicolumn{8}{|c|}{0}&\multicolumn{8}{c|}{1}\\\hline
	\end{tabular}
\end{table}

The set $\calA^q(x_{i+1:q})$ also defines the indices for the aforementioned PPM symbols because of the indexing defined earlier. To estimate the $i$-th bit, we have to consider only the outputs at the positions with indices in $\calA^q(x_{i+1:q})$. Assuming that the inputs to levels $1$ to $i-1$ are uniformly distributed, we can express the effective channel as
\begin{align}
  W^i(\widetilde{y}_{ \calA^q(x_{i+1:q})} \vert x_i) &= \frac{1}{2^{i-1}}\sum_{x_{1}}\cdots \sum_{x_{i-1}} ~~\sum_{\widetilde{y}_{\intseq{1}{2^q}\backslash\calA^q(x_{i+1:q})}} W^{\proddist 2^{q}} \left( \widetilde{y} \vert \widetilde{x}(x_{1:q}) \right)\displaybreak[0]\\
                                                     &= \frac{1}{2^{i-1}}\sum_{x_{1}}\cdots \sum_{x_{i-1}} P_0^{\proddist 2^i-1}\left( \widetilde{y}_{\calA^q(x_{i+1:q})\backslash \calA^q(x_{1:q})} \right) P_1(\widetilde{y}_{\calA^q(x_{1:q})}).
\end{align}
Since the channel is memoryless, once the decoder selects the $2^i$ positions indexed by $\calA^q(x_{i+1:q})$, the distribution of the selected output symbols is independent of the input to the levels $i+1$ to $q$. Hence, we can represent the above channel equivalently by
\begin{align}
  W^i(y_1,\dots,y_{2^i}\vert x_i) &= \frac{1}{2^{i-1}}\sum_{x_{1}}\cdots \sum_{x_{i-1}} P_0^{\proddist 2^i} (y_1,\dots,y_{2^i}) \frac{P_1(y_{\calA^i(x_{1:i})})}{P_0(y_{\calA^i(x_{1:i})})}\\
                                  &=\frac{1}{2^{i-1}} \sum_{k\in\calA^i(x_{i})}P_0^{\proddist 2^i} (y_1,\dots,y_{2^i}) \frac{P_1(y_{k})}{P_0(y_{k})}\\
                                  &= \frac{1}{2^{i-1}}\sum_{j=2^{i-1}x_i+1}^{2^{i-1}(1+x_i)} P_0^{\proddist 2^i} (y_1,\dots,y_{2^i}) \frac{P_1(y_j)}{P_0(y_j)}.\label{eq:equivalent_channel}
\end{align}
The crucial aspect of this characterization is showing that this channel remains unchanged irrespective of the number of levels $q$ used. Note that the above channel is equivalent to the $i$-th channel in an \ac{MLC} with $i$ levels. 
Since the $i$-th level channel is constant, we know that its capacity is fixed, and we can characterize it more precisely as follows.
\begin{lemma}
  \label{lm:capacity_level_bound}
  For $i\in\intseq{1}{q}$, the capacity $C_i$ of the $i$-th level satisfies
  \begin{multline}
    C_i\leq \frac{1}{2^i}\chi_2(P_1\Vert P_0)
    - \frac{1}{2^{2i-1}}\left[\theta(P_1\Vert P_0) - \chi_2(P_1 \Vert P_0) - 2 \chi_2(P_1 \Vert P_0)^3\right]\\
    + \frac{1}{3\times 2^{3(i-1)}} \left[\rho(P_1\Vert P_0) - 3 \chi_2(P_1 \Vert P_0)\right].
  \end{multline}
  where
  \begin{align*}
    \theta(P_1\|P_0)&\eqdef   \sum_{y}P_1(y)\left(\frac{P_1(y)-P_0(y)}{P_0(y)}\right)^2, \\
    \rho(P_1\|P_0) &\eqdef   \sum_{y}P_1(y)\left( \frac{P_1(y)-P_0(y)}{P_0(y)}\right)^3.
  \end{align*}
\end{lemma}
\begin{proof}
  The capacity of the $i$-th level is  $I(X_i;\widetilde{Y}\vert X_{i+1:q})$ with $X_j$ for $j\in\intseq{1}{q}$ distributed uniformly in $\{0,1\}$. Since $X_i$ is independent of $X_{i+1:q}$, we have
  \begin{align}
    I(X_i;\widetilde{Y}\vert X_{i+1:q})
    &= I(X_i;\widetilde{Y},X_{i+1:q})\\
    &=I(X_{1:i};\widetilde{Y},X_{i+1:q})-I(X_{1:i-1};\widetilde{Y},X_{i:q})\\
    &=I(X_{1:i};\widetilde{Y}|X_{i+1:q})-I(X_{1:i-1};\widetilde{Y}|X_{i:q}).
  \end{align}
  Similar to the argument used for the equivalent channel of each level, the quantity $I(X_{1:i};\widetilde{Y},X_{i+1:q})$ is equivalent to $I(X_{1:i};\widetilde{Y}_{\calA^q(X_{i+1:q})})$, which represents the capacity of a \ac{PPM} channel of order $2^i$. Using~\cite[Eq.(13)]{Bloch2017}, we obtain
  \begin{align}
    I(X_1,\dots,X_i;Y_1,\dots,Y_{2^i}) = \avgD{P_1}{P_0} - \mathbb{D}\left( P_{\text{PPM}}^{2^{i}} \Vert P_0^{\proddist 2^{i}} \right),
  \end{align}
  where ${P}_{\text{PPM}}^{2^{i}}$ represents the output distribution when the input is uniform over all possible PPM symbols of order $2^i$. Hence, we have
  \begin{align}
    C_i = I(X_i;\widetilde{Y}\vert X_{i+1:q}) &=\mathbb{D}\left( P_{\text{PPM}}^{2^{i-1}} \Vert P_0^{\proddist 2^{i-1}} \right) - \mathbb{D}\left( P_{\text{PPM}}^{2^{i}} \Vert P_0^{\proddist 2^{i}} \right)
  \end{align}Upper bounding the right-hand side using~\cite[Lemma 1]{Bloch2017} yields the desired result.
\end{proof}

Note that, for higher levels, the capacity $C_i$ goes to zero exponentially in the index of the level. Since the sum of the capacities  converges to $\avgD{P_1}{P_0}$ by~\cite[Lemma 2]{Bloch2017} and by the capacity-achieving property of \ac{MLC} with \ac{MSD}~\cite{Wachsmann1999}, very few levels concentrate most of the capacity. As an example, Table \ref{tab:mlc_capacity} shows the capacity per level of first 16 levels for a binary symmetric channel with cross-over probability 0.1. Notice that the first 5 levels concentrate 93.4\% of the total capacity.

\begin{table}[h!]
  \centering
  \caption{Capacity of first 16 levels for a binary symmetric channel with cross-over probability 0.1}
  \label{tab:mlc_capacity}
  \begin{tabular}{|c|c|c|c|c|c|c|c|c|}
    \hline
    Level, $i$ & 1 & 2 & 3 & 4 & 5 & 6 & 7 & 8 \\
    \hline
    Capacity, $C_i$ & 0.7421 &   0.6387  &  0.4918  &  0.3214  &  0.1749  &  0.0853  &  0.0413 &   0.0203  \\
    \hline
    Level, $i$ & 9 & 10 & 11 & 12 & 13 & 14 & 15 & 16 \\
    \hline
    Capacity, $C_i$ & 0.0101  &  0.0050  &  0.0025  &  0.0013  &  0.0006  &  0.0003 &   0.0002  &  0.0001\\
    \hline
  \end{tabular}
\end{table}

We state and prove some of the properties of equivalent channels in the following lemmas.
\begin{lemma}
  The equivalent channel defined in~\eqref{eq:equivalent_channel} is symmetric.
\end{lemma}
\begin{proof}
  The equivalent channel $W^i$ has the following property:
  \begin{align}
    W^i(y_1,\dots,y_{2^i}\vert x_i) &= W^i(y_{2^{i-1}+1},\dots,y_{2^i},y_1,\dots,y_{2^{i-1}}|x_i\oplus 1).
  \end{align}
  Therefore, the equivalent channel is symmetric.
\end{proof}
\begin{definition}
  Let $\mathcal{P}_i$ represent the set of all permutations of $\intseq{1}{2^i}$. We define the set $\Pi_i$ as $$\Pi_i \eqdef \left\{ \sigma \in \mathcal{P}_i: \forall i \in \intseq{1}{2^{i-1}}, \sigma(i)\in\intseq{1}{2^{i-1}} \right\}.$$
  Let $\Pi_i^\ell$ represent the set of permutations of a vector of length $\ell$, whose components are vectors of length $2^i$, such that each component is permuted by one element of $\Pi_i$.
\end{definition}
Note that for $\pi\in\Pi_i^\ell$ and $\widetilde{\mathbf{y}} \in \widetilde{\mathcal{Y}}_i^\ell$, $\pi(\widetilde{\mathbf{y}})$ is a permutation  of $\widetilde{\mathbf{y}} $ such that the components of each $\widetilde{y}_j\in \mathcal{Y}^{2^i}$ are permuted in such a way that the components in the first half remain in the first half.
\begin{lemma}\label{lemma:perm_invariance}
  The equivalent channel $W^i$ is invariant under any permutation $\pi \in \Pi_i$.
\end{lemma}
\begin{proof}
  We have
  \begin{align}
    W^i\left(y_{\pi(1)},\dots,y_{\pi(2^i)} | x_i \right) &= \frac{1}{2^{i-1}} \sum_{k\in\calA^i(x_{i})}P_0^{\proddist 2^i} (y_{\pi(1)},\dots,y_{\pi(2^i)}) \frac{P_1(y_{\pi(k)})}{P_0(y_{\pi(k)})}\displaybreak[0]\\
                                                         &\overset{(a)}{=} \frac{1}{2^{i-1}} \sum_{k\in\calA^i(x_{i})}P_0^{\proddist 2^i} (y_1,\dots,y_{2^i}) \frac{P_1(y_{k})}{P_0(y_{k})}\\
                                                         &= W^i(y_1,\dots,y_{2^i}|x_i),
  \end{align}
  where (a) follows because $\calA^i(x_i)$ is either $\intseq{1}{2^{i-1}}$ or $\intseq{2^{i-1}+1}{2^i}$ depending on the value of $x_i$, and from the the definition of $\Pi_i$, the summation has the same terms with and without the permutation $\pi$. Therefore, $W^i$ is invariant under any permutation $\pi \in \Pi_i$.
\end{proof}

\subsection{Analysis of \ac{MSD}}
We now prove that we can achieve reliability for the \ac{MLC} scheme with \ac{MSD} if we use  \emph{independent} reliability codes for the equivalent channel corresponding to each level. Note that the actual channel for each level is different from the equivalent channel defined in \eqref{eq:equivalent_channel} because of the use of codes for the lower levels instead of uniformly distributed inputs. Let $W^i$ be the equivalent channel for level~$i$ as defined in \eqref{eq:equivalent_channel}. Let $\phi_i$ and $\psi_i$ represent the encoder and decoder for an $(\ell,m_i)$-code for the $i$-th channel with rate $R_i = \frac{m_i}{\ell}$. We define the rate of the \ac{MLC} scheme as
\begin{align}
  R \eqdef \sum_{i=1}^q R_i.
\end{align}
\emph{Encoding}: The encoder partitions the message $W$ with $\ell R$ bits into $q$ messages such that the $i$-th message $W_i$ contains $\ell R_i$ bits. It then uses $\phi_i$ to encode $W_i$ into an $\ell$-bit sequence $\mathbf{X}_i = (X_{i,1}, \cdots, X_{i,\ell})$. Let $X_{i:q,j} \eqdef \left( X_{i,j},X_{i+1,j},\dots,X_{q,j}\right)$. The \ac{PPM} mapper maps $X_{1:q,j}$ to $\widetilde{x}(X_{1:q,j})$ for $j\in \intseq{1}{\ell}$  as defined in \eqref{eq:ppm_mapper} to form a sequence of $\ell$ PPM symbols, which is transmitted over the channel.\\
\emph{Decoding}: Let $\widetilde{\mathbf{Y}} = (\widetilde{Y}_1, \cdots, \widetilde{Y}_{\ell})$ be the received sequence, where $\widetilde{Y}_i = (Y_{i,1},\dots,Y_{i,m})$. The decoder starts from the level $q$ and decodes the messages successively from level $q$ to $1$.  To decode the $i$-th level, we assume that estimates of $W_{i+1}, \dots, W_q $ are available at the decoder as $\widehat{W}_{i+1}, \cdots, \widehat{W}_{q}$. The decoder also has the estimates of the inputs to these decoded levels as $\widehat{\mathbf{X}}_j \eqdef \phi(\widehat{W}_j )$. It uses these estimates to form a sequence 
\begin{align}
	\widetilde{\mathbf{Y}}_{\calA^q(\hat{\mathbf{X}}_{i+1:q})} \eqdef \left( \widetilde{Y}_{1,\calA^q(\hat{X}_{i+1:q,1})}, \dots, \widetilde{Y}_{\ell,\calA^q(\hat{X}_{i+1:q,\ell})} \right),\label{eq:def_of_y^{(i)}}
\end{align}
where  $\widetilde{Y}_{k,\calA^q(\hat{X}_{i+1:q,k})} \eqdef \left( Y_{k,t}\right)_{t\in\calA^q(\hat{X}_{i+1:q,k})}$. The decoder estimates the message $W_i$ as $$\widehat{W}_i \eqdef \psi_i(\widetilde{\mathbf{Y}}_{\calA^q(\hat{\mathbf{X}}_{i+1:q})}).$$
We define the decoding region of message $w_i$ for the decoder $\psi_i$ as $\mathcal{D}_{w_i} \eqdef \{ \widetilde{\mathbf{y}}\in\widetilde{\calY}_i^\ell: \psi_i(\widetilde{\mathbf{y}}) = w_i\}$. \\
{\bf Assumption}: For $\widetilde{\mathbf{y}} = (\widetilde{y}_1,\dots,\widetilde{y}_\ell)$, where $\widetilde{y}_j = (y_{j,1},\dots,y_{j,2^i})$, and for any permutation $\pi\in\Pi_i^\ell$, we assume that if $\widetilde{\mathbf{y}} \in \mathcal{D}_{w_i}$, then $\pi(\widetilde{\mathbf{y}})\in\mathcal{D}_{w_i}$. \\
This assumption is reasonable because there exist efficient decoders such as the successive cancellation decoder for polar codes that have this property.

\begin{lemma}
  For the successive cancellation decoder of polar codes, $\widetilde{\mathbf{y}} \in \mathcal{D}_{w_i}\implies\pi(\widetilde{\mathbf{y}})\in\mathcal{D}_{w_i}$ for every $\pi\in\Pi_i^\ell$.
\end{lemma}
\begin{proof}
  Assume $\ell$ is a power of two. Let $\svbu_{i,1:\ell} = \mathbf{x}_{i,1:\ell}G_\ell$, where $G_\ell$ is the polar code transform matrix defined in \cite{Arikan2009}. In successive decoding of polar codes, the decoder produces the decision by computing  
  \begin{align}
    \P{U_j=u_j \vert \widetilde{\mathbf{Y}}=\widetilde{{\bf y}}, \mathbf{U}_{1:j-1}=\hat{\mathbf{u}}_{1:j-1}} = \frac{\sum_{\mathbf{u}_{j+1:\ell}}p_{U_{1:\ell}}(\hat{\mathbf{u}}_{1:j},\mathbf{u}_{j:\ell}) (W^i)^{\proddist\ell}(\widetilde{{\bf y}} \vert \mathbf{x}_i)}{p_{U_{1:j-1}}(\mathbf{u}_{1:j-1}) W^i(\widetilde{{\bf y}} | \mathbf{u}_{1:j-1})},
  \end{align}
  where
  \begin{align}
    W^i(\widetilde{{\bf y}} \vert \mathbf{u}_{1:j}) &= \sum_{u_{j+1}=0}^{1} \cdots \sum_{u_\ell = 0}^{1} (W^i)^{\proddist \ell}(\widetilde{{\bf y}} | \mathbf{x}_i) p(u_{j+1},\dots,u_q).
  \end{align}
  From Lemma~\ref{lemma:perm_invariance}, we have
  \begin{align}
  (W^i)^{\proddist \ell}\left( \pi(\widetilde{{\bf y}}) | \mathbf{x}_i \right) &= (W^i)^{\proddist \ell}\left( \widetilde{{\bf y}} | \mathbf{x}_i \right).
  \end{align}
  Hence, $\P{U_j=u_j \vert \widetilde{\mathbf{Y}}=\widetilde{{\bf y}}, \mathbf{U}_{1:j-1}=\hat{\mathbf{u}}_{1:j-1}}$ is also invariant under the permutation $\pi$ and the result follows.
\end{proof}

\begin{lemma}\label{lm:channel_reliability} 
  Suppose we have a code $(\phi_i, \psi_i)$ for the channel 
  \begin{align}
  	W^i(y_1,\dots,y_{2^i}\vert x_i) &=\frac{1}{2^{i-1}} \sum_{k\in\calA^i(x_{i})}P_0^{\proddist 2^i} (y_1,\dots,y_{2^i}) \frac{P_1(y_{k})}{P_0(y_{k})},
  \end{align}
  with rate $R_i$ and probability of error $\epsilon_i$, such that $\psi_i$ satisfies the above assumption on the decoding region. We can then design a code for the MLC-PPM scheme with rate $R=\sum_i^q R_i$ and probability of error $\epsilon = \sum_i^q \epsilon_i$.
\end{lemma}
\begin{proof}
  Let $\mathcal{E}_i \eqdef \{\widehat{W}_i \neq W_i\}$. We can express the probability of error as follows:
  \begin{align}
    \P{(\widehat{W}_1,\cdots,\widehat{W}_q)\neq(W_1,\cdots,W_q)} &=\P{\bigcup_{i=1}^q \mathcal{E}_i}\\
                                                                      &= \P{ \bigcup_{i=1}^q \left(\mathcal{E}_i\bigcap\bigcap_{j=i+1}^{q} \mathcal{E}_j^c\right)}\\
                                                                      &= \sum_{i=1}^{q}\P{\mathcal{E}_i\bigcap\bigcap_{j=i+1}^{q} \mathcal{E}_j^c}.
  \end{align}
  Expanding one term of the summation, we obtain
  \begin{align}
    \P{\mathcal{E}_i\bigcap\bigcap_{j=i+1}^{q} \mathcal{E}_j^c} &= \P{\widehat{W}_i\neq W_i,\widehat{W}_{i+1}=W_{i+1}, \cdots, \widehat{W}_q=W_q}\displaybreak[0]\\
                                                                     &=\sum_{w_i}\cdots\sum_{w_q} \P{\widehat{W}_i \neq w_i, W_i = w_i, \widehat{W}_{i+1} = W_{i+1} = w_{i+1}, \dots, \widehat{W}_q = W_q = w_q}\displaybreak[0]\displaybreak[0]\\
                                                                     &= \sum_{w_i}\cdots\sum_{w_q} \P{\widehat{W}_{i+1} = W_{i+1} = w_{i+1}, \dots, \widehat{W}_q = W_q = w_q} \P{W_i=w_i}\nonumber\displaybreak[0]\\  &\qquad\qquad\P{\widehat{W}_i \neq w_i\left\vert W_i = w_i, \widehat{W}_{i+1} = W_{i+1} = w_{i+1}, \dots, \widehat{W}_q = W_q = w_q\right.},
  \end{align}
  and
  \begin{align}
    &\P{\widehat{W}_i \neq w_i | \widehat{W}_{i+1}=w_{i+1},\dots,\widehat{W}_q = w_q,W_{i+1}=w_{i+1},\dots,W_q=w_q}\nonumber\displaybreak[0]\\
    &\qquad\qquad=\sum_{\widetilde{\mathbf{y}}:\psi_i(\widetilde{\mathbf{y}})\neq w_i} \P{\widetilde{\mathbf{Y}}_{\calA^q({\mathbf{X}}_{i+1:q})} = \widetilde{\mathbf{y}} | \mathbf{X}_i = \mathbf{x}_i,\dots,\mathbf{X}_q = \mathbf{x}_q}\displaybreak[0]\\
    &\qquad\qquad=\sum_{\widetilde{\mathbf{y}}:\psi_i(\widetilde{\mathbf{y}})\neq w_i} \sum_{\mathbf{x}_1} \cdots \sum_{\mathbf{x}_{i-1}} \P{\mathbf{X}_1=\mathbf{x}_1,\dots,\mathbf{X}_{i-1}=\mathbf{x}_{i-1}} \prod_{j=1}^\ell P_0^{\proddist 2^i}(\widetilde{y}_j) \frac{P_1\left(y_{j,\calA^i(x_{1:i,j})}\right)}{P_0\left(y_{j,\calA^i(x_{1:i,j})}\right)}  \displaybreak[0]\\
    &\qquad\qquad=\sum_{\mathbf{x}_1} \cdots \sum_{\mathbf{x}_{i-1}} \P{\mathbf{X}_1=\mathbf{x}_1,\dots,\mathbf{X}_{i-1}=\mathbf{x}_{i-1}} \sum_{\widetilde{\mathbf{y}}:\psi_i(\widetilde{\mathbf{y}})\neq w_i} \prod_{j=1}^\ell P_0^{\proddist 2^i}(\widetilde{y}_j) \frac{P_1\left(y_{j,\calA^i(x_{1:i,j})}\right)}{P_0\left(y_{j,\calA^i(x_{1:i,j})}\right)}  \displaybreak[0]\\
    &\qquad\qquad=\sum_{\mathbf{x}_1} \cdots \sum_{\mathbf{x}_{i-1}} \P{\mathbf{X}_1=\mathbf{x}_1,\dots,\mathbf{X}_{i-1}=\mathbf{x}_{i-1}}\nonumber\\ &\qquad\qquad\qquad\qquad\qquad\qquad\qquad\qquad\qquad\sum_{\widetilde{\mathbf{y}}:\psi_i(\widetilde{\mathbf{y}})\neq w_i} \left(\frac{1}{2^{i-1}}\right)^\ell \prod_{j=1}^\ell \sum_{k\in \calA^i(x_{i,j})}P_0^{\proddist 2^i}(\widetilde{y}_j) \frac{P_1\left(y_{j,k}\right)}{P_0\left(y_{j,k}\right)} \label{eq:prob_error_1} \displaybreak[0]\\
    &\qquad\qquad=\sum_{\widetilde{\mathbf{y}}:\psi_i(\widetilde{\mathbf{y}})\neq w_i} \left(\frac{1}{2^{i-1}}\right)^\ell \prod_{j=1}^\ell \sum_{k\in \calA^i(x_{i,j})}P_0^{\proddist 2^i}(\widetilde{y}_j) \frac{P_1\left(y_{j,k}\right)}{P_0\left(y_{j,k}\right)} .
  \end{align}
  where \eqref{eq:prob_error_1} results from the assumption on the decoding region. In fact, each term inside $\sum_{k\in \calA^i(x_{i,j})}$ can be obtained by permuting $\widetilde{y}_j$ by one of the permutation defined earlier. Since, all those permutations are included in the set $\{\widetilde{\mathbf{y}}:\psi_i(\widetilde{\mathbf{y}})\neq w_i\}$, the same product term is repeated $\left(2^{i-1}\right)^\ell$ times.
  
  Hence,
  \begin{align}
    \P{\mathcal{E}_i\bigcap\bigcap_{j=i+1}^{q} \mathcal{E}_j^c} &= \sum_{w_i}\cdots\sum_{w_q} \P{\widehat{W}_{i+1} = W_{i+1} = w_{i+1}, \dots, \widehat{W}_q = W_q = w_q} \P{W_i=w_i}\nonumber\\  &\phantom{=================}\sum_{\widetilde{\mathbf{y}}:\psi_i(\widetilde{\mathbf{y}})\neq w_i} \left(\frac{1}{2^{i-1}}\right)^\ell \prod_{j=1}^\ell \sum_{k\in \calA^i(x_{i,j})}P_0^{\proddist 2^i}(\widetilde{y}_j) \frac{P_1\left(y_{j,k}\right)}{P_0\left(y_{j,k}\right)} \displaybreak[0]\\
                                                                     &= \sum_{w_i} \P{W_i = w_i}\sum_{\widetilde{\mathbf{y}}:\psi_i(\widetilde{\mathbf{y}})\neq w_i} \left(\frac{1}{2^{i-1}}\right)^\ell \prod_{j=1}^\ell \sum_{k\in \calA^i(x_{i,j})}P_0^{\proddist 2^i}(\widetilde{y}_j) \frac{P_1\left(y_{j,k}\right)}{P_0\left(y_{j,k}\right)} \nonumber\\ &\phantom{==============} \sum_{w_{i+1}}\cdots\sum_{w_q} \P{\widehat{W}_{i+1} = W_{i+1} = w_{i+1}, \dots, \widehat{W}_q = W_q = w_q} \displaybreak[0]\\
                                                                     &\leq \sum_{w_i} \P{W_i = w_i}\sum_{\widetilde{\mathbf{y}}:\psi_i(\widetilde{\mathbf{y}})\neq w_i} \left(\frac{1}{2^{i-1}}\right)^\ell \prod_{j=1}^\ell \sum_{k\in \calA^i(x_{i,j})}P_0^{\proddist 2^i}(\widetilde{y}_j) \frac{P_1\left(y_{j,k}\right)}{P_0\left(y_{j,k}\right)}  \nonumber\\ &\phantom{==================} \sum_{w_{i+1}}\cdots\sum_{w_q} \P{W_{i+1} = w_{i+1}, \dots,  W_q = w_q}\displaybreak[0]\\
                                                                     &= \sum_{w_i} \P{W_i = w_i}\sum_{\widetilde{\mathbf{y}}:\psi_i(\widetilde{\mathbf{y}})\neq w_i} \left(\frac{1}{2^{i-1}}\right)^\ell \prod_{j=1}^\ell \sum_{k\in \calA^i(x_{i,j})}P_0^{\proddist 2^i}(\widetilde{y}_j) \frac{P_1\left(y_{j,k}\right)}{P_0\left(y_{j,k}\right)} .
  \end{align}
  
  For the equivalent channel, the probability of error for a given input message is 
  \begin{align}
    \P{\widehat{W} \neq w_i | W_i = w_i} &= \sum_{\widetilde{\mathbf{y}}:\psi_i(\widetilde{\mathbf{y}})\neq w_i} \prod_{j=1}^\ell \frac{1}{2^{i-1}} \sum_{k\in\calA^i(x_{i,j})}P_0^{\proddist 2^i}(\widetilde{y}_j) \frac{P_1\left(y_{j,k}\right)}{P_0\left(y_{j,k}\right)} .
  \end{align}
  Hence, the probability of error for the equivalent channel is given by
  \begin{align}
    \P{\widehat{W}_i \neq W_i} &= \sum_{w_i} \P{W_i = w_i} \P{\widehat{W} \neq w_i | W_i = w_i} \\
                               &= \sum_{w_i} \P{W_i = w_i} \sum_{\widetilde{\mathbf{y}}:\psi_i(\widetilde{\mathbf{y}})\neq w_i} \left(\frac{1}{2^{i-1}}\right)^\ell \prod_{j=1}^\ell \sum_{k\in \calA^i(x_{i,j})}P_0^{\proddist 2^i}(\widetilde{y}_j) \frac{P_1\left(y_{j,k}\right)}{P_0\left(y_{j,k}\right)} .
  \end{align}
  Since we know that $\P{\widehat{W}_i \neq W_i} \leq \epsilon_i$, we have
  \begin{align}
    \P{\mathcal{E}_i\bigcap\bigcap_{j=i+1}^{q} \mathcal{E}_j^c} &\leq \P{\widehat{W}_i \neq W_i}\\
                                                                     &\leq \epsilon_i.
  \end{align}
  Therefore,
  \begin{align}
    \P{(\widehat{W}_1,\cdots,\widehat{W}_q)\neq(W_1,\cdots,W_q)} &\leq \sum_{i=1}^{q} \epsilon_i.
  \end{align}
\end{proof}

\subsection{Analysis of covertness}
\label{sec:analysis-covertness}
We now turn our attention to the covertness properties of the \ac{MLC} scheme. Instead of dealing with relative entropy, we work  with the variational distance $\V{P_{\mathbf{\widetilde{Z}}},(Q_{\text{PPM}}^{m})^{\proddist \ell}}$ between the distribution $P_{\mathbf{\widetilde{Z}}}$ induced at the output of the \ac{PPM} super-channel when coding over $\ell$ \ac{PPM} symbols of order $m$, and the distribution $(Q_{\text{PPM}}^{m})^{\proddist \ell}$, which is a product distribution over the $\ell$ uses of the super channel, and $Q_{\text{PPM}}^{m}$ is the output induced by a uniform input distribution on $\ac{PPM}$ symbols of order $m$. For $j\in\intseq{1}{q}$, assume that the codebook at level $j$ consists of $M_j$ codewords $\mathcal{C}_i = \{\mathbf{c}(i_j)\}_{i_j=1}^{M_j}$. Upon denoting the super channel transition probability $\widetilde{W}_{\widetilde{Z}|\widetilde{X}}$ by $\widetilde{W}$, we introduce the distribution
\begin{align}
  P_{\bf\widetilde{Z}}^{(j)}({\bf\widetilde{z}}) = \frac{1}{2^\ell}\sum_{\bf x_1\in\{0,1\}^\ell} \cdots \frac{1}{2^\ell} \sum_{\bf x_{j}\in\{0,1\}^\ell}\frac{1}{M_{j+1}}  \sum_{\mathbf{c}_{j+1}\in\mathcal{C}_{j+1}}\cdots \frac{1}{M_q} \sum_{\mathbf{c}_q\in\mathcal{C}_q} 
  \widetilde{W}^{\proddist\ell}\left({\bf \widetilde{z}} \vert {\bf x}_1,\dots,{\bf x}_{j},{\bf c}_{j+1},\dots,{\bf c}_q\right).
\end{align}
Intuitively, $  P_{\bf\widetilde{Z}}^{(j)}$ represents the distribution induced at Willie's output when coding all levels from $i$ down to $q$ and transmitting uniformly distributed bits on all lower levels.
Note that $P_{\bf\widetilde{Z}}^{(0)} = P_{\bf\widetilde{Z}}$ and $P_{\bf \widetilde{Z}}^{(q)} = (Q_{\text{PPM}}^{m})^{\proddist \ell}$. Using a triangle inequality repeatedly, we obtain that
\begin{align}
  \mathbb{V}\left(P_{\bf\widetilde{Z}},(Q_{\text{PPM}}^{m})^{\proddist \ell}\right) &\leq \sum_{j=1}^{q} \mathbb{V}\left( P_{\bf \widetilde{Z}}^{(j-1)}, P_{\bf \widetilde{Z}}^{(j)} \right).\label{eq:variational_dist_mlc}
\end{align}
Now, we can further upper bound $\mathbb{V}\left( P_{\bf \widetilde{Z}}^{(j-1)}, P_{\bf \widetilde{Z}}^{(j)} \right)$ as
\begin{align}
  &\mathbb{V}\left( P_{\bf \widetilde{Z}}^{(j-1)}, P_{\bf \widetilde{Z}}^{j} \right) =\sum_{\mathbf{\widetilde{z}}} \left|\frac{1}{2^\ell}\sum_{\bf x_1} \cdots \frac{1}{2^\ell} \sum_{\bf x_{j-1}}\frac{1}{M_{j+1}}\sum_{\mathbf{c}_{j+1}\in\mathcal{C}_{j+1}}\cdots \frac{1}{M_q} \sum_{\mathbf{c}_q\in\mathcal{C}_q} \right.\nonumber\\ &\qquad\left.\left(\frac{1}{M_{j}}\sum_{\mathbf{c}_j\in\mathcal{C}_j}\widetilde{W}^{\proddist\ell}\left({\bf \widetilde{z}} \vert {\bf x}_1,\dots,{\bf x}_{j-1},{\bf c}_{j},{\bf c}_{j+1},\dots,{\bf c}_q\right)- \frac{1}{2^\ell} \sum_{\mathbf{x}_{j}}\widetilde{W}^{\proddist\ell}\left({\bf \widetilde{z}} \vert {\bf x}_1,{\bf c}_{j+1},\dots,{\bf x}_{j-1},{\bf x}_{j},\dots,{\bf c}_q\right) \right)\right|\displaybreak[0]\\
  &\quad\leq   \frac{1}{M_{j+1}}\sum_{\mathbf{c}_{j+1}\in\mathcal{C}_{j+1}}\cdots \frac{1}{M_q} \sum_{\mathbf{c}_q\in\mathcal{C}_q}  \sum_{\mathbf{\widetilde{z}}} \left| \frac{1}{2^\ell} \sum_{\bf x_1} \cdots \frac{1}{2^\ell} \sum_{\bf x_{j-1}}\left(\frac{1}{M_{j}} \sum_{\mathbf{c}_j\in\mathcal{C}_j}\widetilde{W}^{\proddist\ell}\left({\bf \widetilde{z}} \vert {\bf x}_1,\dots,{\bf x}_{j-1},{\bf c}_{j},{\bf c}_{j+1},\dots,{\bf c}_q\right) \right. \right. \nonumber\\ &\quad\qquad\qquad\qquad\qquad\qquad\qquad\qquad\qquad\qquad\qquad\quad \left. \left.- \frac{1}{2^\ell} \sum_{\mathbf{x}_{j}}\widetilde{W}^{\proddist\ell}\left({\bf \widetilde{z}} \vert {\bf x}_1,\dots,{\bf x}_{j-1},{\bf x}_{j},{\bf c}_{j+1},\dots,{\bf c}_q\right) \right)\right| \label{eq:bound_v_differential}
\end{align}
Notice that the two terms inside the absolute are distributions that only differ in that one has a coded $j$-th level, while the other has an uncoded $j$-th level with uniform random bits.
\begin{lemma}
  \label{lm:resolvability}
  For every $j\in\intseq{1}{q}$, consider the channel
  \begin{align}
    W^j(z_1,\dots,z_{2^j}\vert x_j) = \frac{1}{2^{j-1}}\sum_{k\in\calA^j(x_j)} Q_0^{\proddist 2^j} (z_1,\dots,z_{2^j}) \frac{Q_1(z_k)}{Q_0(z_k)}.\label{eq:equivalent_channel_resolvability}
  \end{align}
  Let $P_{\mathbf{\widetilde{Z}}^{(j)}}$ denote the output distribution induced by a code over this channel, and let $Q_{\widetilde{Z}^{(j)}}^{\proddist\ell}$ denote the product output distribution when the input is uniform. If $\V{P_{\mathbf{\widetilde{Z}}^{(j)}},Q_{\widetilde{Z}^{(j)}}^{\proddist\ell}}\leq \delta_j$, then the same code ensures that $\mathbb{V}\left( P_{\bf \widetilde{Z}}^{(j-1)}, P_{\bf \widetilde{Z}}^{(j)} \right)\leq\delta_j$ irrespective of the code used for the higher levels. Moreover, for the output of the super channel $\widetilde{W}_{\widetilde{Z}|X_{1:q}}$, these codes together ensure $\mathbb{V}\left(P_{\bf\widetilde{Z}},(Q_{\text{PPM}}^{m})^{\proddist \ell}\right)\leq\sum_{j=1}^q\delta_j$.
\end{lemma}
\begin{proof}
  The result follows from observations similar to those in the proof of Lemma~\ref{lm:channel_reliability} regarding the symmetries of the \ac{PPM} modulation. Specifically, consider level $j$. For all codewords $\mathbf{c}_{j+1},\cdots,\mathbf{c}_q$ used in the upper levels, we have
  \begin{align}
    \widetilde{W}^{\proddist\ell}\left({\bf \widetilde{z}} \vert {\bf x}_1,\dots,{\bf x}_{j-1},{\bf c}_{j},\dots,{\bf c}_q\right) = \prod_{i=1}^{\ell} Q_0^{\proddist 2^q}(\widetilde{z}_i) \frac{Q_1(z_{i,\calA^q(x_{1:j-1,i},c_{j:q,i})})}{Q_0(z_{i,\calA^q(x_{1:j-1,i},c_{j:q,i})})},
  \end{align}
  and
  \begin{align}
    \widetilde{W}^{\proddist\ell}\left({\bf \widetilde{z}} \vert {\bf x}_1,\dots,{\bf x}_{j-1},{\bf x}_{j},\dots,{\bf c}_q\right) = \prod_{i=1}^{\ell} Q_0^{\proddist 2^q}(\widetilde{z}_i) \frac{Q_1(z_{i,\calA^q(x_{1:j,i},c_{j+1:q,i})})}{Q_0(z_{i,\calA^q(x_{1:j,i},c_{j+1:q,i})})}.
  \end{align}
  Hence, we can bound $\mathbb{V}\left( P_{\bf \widetilde{Z}}^{(j-1)}, P_{\bf \widetilde{Z}}^{(j)} \right)$ by substituting the above equations in~\eqref{eq:bound_v_differential} as
  \begin{align}
    \mathbb{V}\left( P_{\bf \widetilde{Z}}^{(j-1)},\right. &\left.P_{\bf \widetilde{Z}}^{j} \right) \leq\frac{1}{M_{j+1}}\sum_{\mathbf{c}_{j+1}\in\mathcal{C}_{j+1}}\cdots \frac{1}{M_q} \sum_{\mathbf{c}_q\in\mathcal{C}_q}  \sum_{\mathbf{\widetilde{z}}} \left| \frac{1}{2^\ell} \sum_{\bf x_1} \cdots \frac{1}{2^\ell} \sum_{\bf x_{j-1}}\right.\nonumber\\ &\left.\left(\frac{1}{M_{j}}\sum_{\mathbf{c}_j\in\mathcal{C}_j}\prod_{i=1}^{\ell}Q_0^{\proddist{2^q}}(\widetilde{z}_i)\frac{Q_1(z_{i,\calA^q(x_{1:j-1,i},c_{j:q,i})})}{Q_0(z_{i,\calA^q(x_{1:j-1,i},c_{j:q,i})})} - \frac{1}{2^\ell} \sum_{\mathbf{x}_{j}} \prod_{i=1}^{\ell} Q_0^{\proddist{2^q}}(\widetilde{z}_i)\frac{Q_1(z_{i,\calA^q(x_{1:j,i},c_{j+1:q,i})})}{Q_0(z_{i,\calA^q(x_{1:j,i},c_{j+1:q,i})})}  \right)\right| \displaybreak[0]\\
    &\overset{(a)}{=}\frac{1}{M_{j+1}}\sum_{\mathbf{c}_{j+1}\in\mathcal{C}_{j+1}}\cdots \frac{1}{M_q} \sum_{\mathbf{c}_q\in\mathcal{C}_q}  ~\sum_{\mathbf{\widetilde{z}}_{\calA^q(\mathbf{c}_{j+1:q})^c}} ~\prod_{i'=1}^{\ell} ~~\prod_{k'\in\calA^q(c_{j+1:q,i'})^c} Q_0(z_{i',k'})\nonumber\\ &\qquad\qquad\qquad\sum_{\mathbf{\widetilde{z}}_{\calA^q(\mathbf{c}_{j+1:q})}} \left| \frac{1}{2^\ell} \sum_{\bf x_1} \cdots \frac{1}{2^\ell} \sum_{\bf x_{j-1}} \left(\frac{1}{M_{j}} \sum_{\mathbf{c}_j\in\mathcal{C}_j} \prod_{i=1}^{\ell} Q_0^{\proddist 2^j}(\widetilde{z}_{i,\calA^q(c_{j+1:q,i})}) \frac{Q_1(z_{i,\calA^q(x_{1:j-1,i},c_{j:q,i})})}{Q_0(z_{i,\calA^q(x_{1:j-1,i},c_{j:q,i})})} \right.\right.\nonumber\\ &\qquad\qquad\qquad\qquad\qquad\qquad\qquad\quad\left.\left.-\frac{1}{2^\ell} \sum_{\mathbf{x}_{j}} \prod_{i=1}^{\ell} Q_0^{\proddist 2^j}(\widetilde{z}_{i,\calA^q(c_{j+1:q,i})}) \frac{Q_1(z_{i,\calA^q(x_{1:j,i},c_{j+1:q,i})})}{Q_0(z_{i,\calA^q(x_{1:j,i},c_{j+1:q,i})})}  \right)\right|\displaybreak[0]\\
    &=\frac{1}{M_{j+1}}\sum_{\mathbf{c}_{j+1}\in\mathcal{C}_{j+1}}\cdots \frac{1}{M_q} \sum_{\mathbf{c}_q\in\mathcal{C}_q}  ~\sum_{\mathbf{\widetilde{z}}_{\calA^q(\mathbf{c}_{j+1:q})^c}}~\prod_{i'=1}^{\ell} ~~\prod_{k'\in\calA^q(c_{j+1:q,i'})^c} Q_0(z_{i',k'})  \nonumber\\ &\qquad\qquad\sum_{\mathbf{\widetilde{z}}_{\calA^q(\mathbf{c}_{j+1:q})}} \left| \frac{1}{M_{j}} \sum_{\mathbf{c}_j\in\mathcal{C}_j}  \prod_{i=1}^{\ell} \frac{1}{2^{j-1}} \sum_{x_{1,i}} \cdots \sum_{x_{j-1,i}}  Q_0^{\proddist 2^j}(\widetilde{z}_{i,\calA^q(c_{j+1:q,i})})\frac{Q_1(z_{i,\calA^q(x_{1:j-1,i},c_{j:q,i})})}{Q_0(z_{i,\calA^q(x_{1:j-1,i},c_{j:q,i})})}\right.  \nonumber\\ &\qquad\qquad\qquad\qquad\left.-\frac{1}{2^\ell}\sum_{\mathbf{x}_{j}} \prod_{i=1}^{\ell} \frac{1}{2^{j-1}} \sum_{x_{1,i}} \cdots \sum_{x_{j-1,i}} Q_0^{\proddist 2^j}(\widetilde{z}_{i,\calA^q(c_{j+1:q,i})})\frac{Q_1(z_{i,\calA^q(x_{1:j,i},c_{j+1:q,i})})}{Q_0(z_{i,\calA^q(x_{1:j,i},c_{j+1:q,i})})}  \right|\displaybreak[0]\\
    &=\frac{1}{M_{j+1}}\sum_{\mathbf{c}_{j+1}\in\mathcal{C}_{j+1}}\cdots \frac{1}{M_q} \sum_{\mathbf{c}_q\in\mathcal{C}_q}  ~\sum_{\mathbf{\widetilde{z}}_{\calA^q(\mathbf{c}_{j+1:q})^c}}~\prod_{i'=1}^{\ell} ~~\prod_{k'\in\calA^q(c_{j+1:q,i'})^c} Q_0(z_{i',k'})  \nonumber\\ &\qquad\qquad\qquad\sum_{\mathbf{\widetilde{z}}_{\calA^q(\mathbf{c}_{j+1:q})}} \left| \frac{1}{M_{j}} \sum_{\mathbf{c}_j\in\mathcal{C}_j}  \prod_{i=1}^{\ell} \frac{1}{2^{j-1}} \sum_{k\in \calA^q(c_{j:q,i})} Q_0^{\proddist 2^j}(\widetilde{z}_{i,\calA^q(c_{j+1:q,i})})\frac{Q_1(z_{i,k})}{Q_0(z_{i,k})}\right.\nonumber\\ &\qquad\qquad\qquad\qquad\qquad\qquad\left.-\frac{1}{2^\ell}\sum_{\mathbf{x}_{j}}\prod_{i=1}^{\ell}\frac{1}{2^{j-1}}\sum_{k\in\calA^q(x_{j,i},c_{j+1:q,i})} Q_0^{\proddist 2^j}(\widetilde{z}_{i,\calA^q(c_{j+1:q,i})})\frac{Q_1(z_{i,k})}{Q_0(z_{i,k})}  \right|\displaybreak[0]\\
    &= \frac{1}{M_{j+1}}\sum_{\mathbf{c}_{j+1}\in\mathcal{C}_{j+1}}\cdots \frac{1}{M_q} \sum_{\mathbf{c}_q\in\mathcal{C}_q}  ~\sum_{\mathbf{\widetilde{z}}_{\calA^q(\mathbf{c}_{j+1:q})^c}} ~\prod_{i'=1}^{\ell} ~~\prod_{k'\in\calA^q(c_{j+1:q,i'})^c} Q_0(z_{i',k'}) \nonumber\\ &\qquad\qquad\qquad\qquad\qquad\qquad\qquad\qquad\sum_{\mathbf{\widetilde{z}}_{\calA^q(\mathbf{c}_{j+1:q})}}\left|P_{\mathbf{\widetilde{Z}}^{(j)}}(\mathbf{\widetilde{z}}^{(j)}(\mathbf{c}_{j+1:q})) - Q_{\widetilde{Z}^{(j)}}^{\proddist\ell}(\mathbf{\widetilde{z}}^{(j)}(\mathbf{c}_{j+1:q})) \right| \displaybreak[0]\\
    &= \frac{1}{M_{j+1}}\sum_{\mathbf{c}_{j+1}\in\mathcal{C}_{j+1}}\cdots \frac{1}{M_q} \sum_{\mathbf{c}_q\in\mathcal{C}_q}  ~\sum_{\mathbf{\widetilde{z}}_{\calA^q(\mathbf{c}_{j+1:q})^c}} ~\prod_{i'=1}^{\ell} ~~\prod_{k'\in\calA^q(c_{j+1:q,i'})^c} Q_0(z_{i',k'})  \V{P_{\mathbf{\widetilde{Z}}^{(j)}},Q_{\widetilde{Z}^{(j)}}^{\proddist\ell}} \\ 
    &=\V{P_{\mathbf{\widetilde{Z}}^{(j)}},Q_{\widetilde{Z}^{(j)}}^{\proddist\ell}} \\
    &\leq \delta_j,
  \end{align}
  where (a) follows from dividing the summation over $\widetilde{\mathbf{z}}$ into summation over components $\mathbf{\widetilde{z}}_{\calA^q(\mathbf{c}_{j+1:q})}$ defined as in~\eqref{eq:def_of_y^{(i)}} and its complementary components denoted by $\mathbf{\widetilde{z}}_{\calA^q(\mathbf{c}_{j+1:q})^c}$.
  
  Hence, from \eqref{eq:variational_dist_mlc}, we have
  \begin{align}
  	\mathbb{V}\left(P_{\bf\widetilde{Z}},(Q_{\text{PPM}}^{m})^{\proddist \ell}\right) &\leq \sum_{j}^{q}\delta_j,
  \end{align}
  which proves the lemma.
\end{proof}

Lemma~\ref{lm:resolvability} may be viewed as the counterpart of Lemma~\ref{lm:channel_reliability} for resolvability instead of reliability. Since the equivalent channel~\eqref{eq:equivalent_channel_resolvability} is again invariant with the number of levels $q\geq i$, we conclude that we can design channel resolvability codes for a fixed channel while still growing the number of levels with the code length.

To conclude regarding the ability to achieve covertness with the multilevel scheme, we note with calculations similar to~\cite[(77)-(79)]{Bloch2016} that
\begin{align}
  \label{eq:covert-bound-v}
  \avgD{P_{\mathbf{\widetilde{Z}}}}{Q_0^\pn} \leq \avgD{P_{\mathbf{\widetilde{Z}}}}{(Q_{\text{PPM}}^{m})^{\proddist \ell}} + 2\V{P_{\mathbf{\widetilde{Z}}},(Q_{\text{PPM}}^{m})^{\proddist \ell}}\max_{\widetilde{z}}\ell \abs{\log\frac{(Q_{\text{PPM}}^{m})(\widetilde{z})}{Q_0^{\proddist m}(\widetilde{z})}} + \D{(Q_{\text{PPM}}^{m})^{\proddist \ell} }{Q_0^\pn},
\end{align}
and using~\cite[(323)]{Sason2016}
\begin{align}
  \label{eq:covert-bound-v2}
  \avgD{P_{\mathbf{\widetilde{Z}}}}{(Q_{\text{PPM}}^{m})^{\proddist \ell}}  \leq \ell\log\left(\frac{1}{\min_{\widetilde{{z}}} (Q_{\text{PPM}}^{m})^{\proddist \ell}(\widetilde{{z}})} \right)\V{P_{\mathbf{\widetilde{Z}}},(Q_{\text{PPM}}^{m})^{\proddist \ell}}.
\end{align}
Consequently, provided $\V{P_{\mathbf{\widetilde{Z}}},(Q_{\text{PPM}}^{m})^{\proddist \ell}}$ decays fast enough at each level and with $\ell$ scaling as in~\cite{Bloch2017}, we may ensure covertness at a throughput close to the covert capacity.

\subsection{Towards a concrete polynomial-complexity instantiation}
\label{sec:concr-polyn-compl}
The key observation to instantiate actual codes is that the problem reduces to constructing codes for the equivalent channels identified in Lemma~\ref{lm:channel_reliability} and Lemma~\ref{lm:resolvability}. If the original channels are degraded, then the equivalent channels are also degraded as shown in~(\ref{eq:degraded_equi_ch1}-\ref{eq:degraded_equi_ch2}). Since the successive cancellation decoder of polar codes satisfies the assumption on the decoding regions that we used for the proof of reliability, polar codes would be a suitable solution to achieve the covert capacity. In fact, we know from~\cite{Chou2016} how to design polar codes simultaneously for reliability and resolvability, with a negligible amount of shared randomness and metrics (probability of error and variational distance) that decay fast with the blocklength. Such constructions would carry over directly. One subtle point is how to address coding for the higher and very noisy levels. In fact, achieving reliability would be next to impossible at such low rates, and one would therefore not communicate over these levels, and just achieve channel resolvability using bits of private randomness. The top most levels have a rate that vanishes with the block length, and one concern is whether polarization happens fast enough at these levels. Lemma~\ref{lm:capacity_level_bound} shows that the rate decays exponentially with $q$ and as the inverse of $m$. Polarization, however,  only happens at a rate $\frac{1}{m^\gamma}$ for some $\gamma<1$, which will therefore force us to overestimate the number of random bits to input. Fortunately, the number of levels is logarithmic in $m$, so that the rate of private randomness remains negligible.

\subsubsection{Coding scheme using polar codes}
In the following theorem, we show that we can achieve the covert capacity of a \ac{BI-DMC} using polar codes.
\begin{theorem}
  Fix positive constants $\zeta$ and $\delta$. For $n$ large enough, there exist low-complexity polar coding schemes for each level of \ac{MLC}-\ac{PPM} scheme described in previous sections with covert rate at least $\frac{\sqrt{2}\D{P_1}{P_0}}{\sqrt{\chi_2(Q_1\|Q_0)}} - \zeta$, with probability of error at most $\zeta$, and $\D{P_{\mathbf{Z}}}{Q_0^\pn} \leq \delta+\zeta$. 
\end{theorem}

\begin{proof}
	Let  $m$ be the least power of two greater than $\left \lfloor \sqrt{\frac{\chi_2(Q_1\|Q_0)n} {2\delta }  }\right \rfloor$ and $\ell \eqdef \left \lfloor \frac{n}{m}\right \rfloor $.  We consider an MLC-PPM scheme with the number of levels $q = \log_2 m$ in which the transmission occurs in blocks of $\ell$ \ac{PPM} symbols of order $m$. Since Willie's channel is degraded \ac{wrt} Bob's channel, for each level $i\in\intseq{1}{q}$, the channel $W^i(z_1, \cdots, z^{2^i}|x_i)$ defined in \eqref{eq:equivalent_channel_resolvability} is also degraded \ac{wrt} the channel $W^i(y_1, \cdots, y_{2^i}|x_i)$ defined in \eqref{eq:equivalent_channel}. Let $Q_{\widetilde{Z}^{(i)}}$ represent the distribution induced at the output of $W^i(z_1, \cdots, z_{2^i}|x_i)$ when the input to the channel is uniformly distributed. Let $$R_i^Y \eqdef I(X_i;\widetilde{Y}|X^{i+1:q}) - \frac{A}{\ell^{\epsilon \kappa}},\quad\text{and}\quad R_i^Z\eqdef I(X_i;\widetilde{Z}|X^{i+1:q}) + \frac{A}{\ell^{\epsilon\kappa}}.$$
	Then, for some $\beta\in]0, \frac{1}{2}[$ and $\epsilon\in]0, 1-2\beta[$, by~\cite[Proposition 3, Lemma 1, and Lemma 2]{Freche2017}, there exist constants $\kappa$, $A$, and $C$ and low-complexity polar codes $\calC_i$ for each level $i\in\intseq{1}{q}$ with rate $R_i\eqdef \max(R_i^Y,R_i^Z)$ such that if $\ell R_i$ bits are coded into a binary codeword of length $\ell > 2^C$ and transmitted over both $W^i(y_1, \cdots, y_{2^i}|x_i)$ and $W^i(z_1, \cdots, z_{2^i}|x_i)$, for the distribution $P_{\widetilde{\mathbf{Z}}^{(i)}}$ induced at the output of the channel $W^i(z_1, \cdots, z_{2^i}|x_i)$ , we have $\V{P_{\widetilde{\mathbf{Z}}^{(i)}}, Q_{\widetilde{Z}^{(i)}}^{\proddist \ell}} \leq \calO(\sqrt{\ell2^{-\ell^{\beta}}})$ and a probability of error upper-bounded by $\calO(\ell2^{-\ell^\beta})$. Note that we need $R_i\geq R_i^Z$ to ensure channel resolvability and when $R_i^Z>R_i^Y$, we need key with rate $R_i^Z-R_i^Y$.
  
  If the decoder uses the successive cancellation decoder of polar codes, by Lemma~\ref{lm:channel_reliability}, the probability of error is upper bounded by $\calO(q\ell2^{-\ell^\beta})$, which is less than $\zeta$ for large enough $n$. For covertness, first notice that by Lemma~\ref{lm:resolvability}, we have  $\V{P_{\mathbf{\widetilde{Z}}},(Q_{\text{PPM}}^{m})^{\proddist \ell}}\leq \calO(q\sqrt{\ell2^{-\ell^{\beta}}})$. Thus, from \eqref{eq:covert-bound-v} and \eqref{eq:covert-bound-v2}, we have
  \begin{align}
    \D{P_{\widetilde{\mathbf{Z}}}}{Q_0^\pn} &\leq  \ell\log\left(\frac{1}{\min_{\widetilde{z}} Q_{\text{PPM}}^{m}(\widetilde{z})}\right) \V{P_{\mathbf{\widetilde{Z}}},(Q_{\text{PPM}}^{m})^{\proddist \ell}} + \D{(Q_{\text{PPM}}^{m})^{\proddist \ell} }{Q_0^\pn} \nonumber\\ &\qquad\qquad\qquad\qquad\qquad\qquad\qquad + 2\V{P_{\mathbf{\widetilde{Z}}},(Q_{\text{PPM}}^{m})^{\proddist \ell}}\max_{\widetilde{z}}\ell \abs{\log\frac{Q_{\text{PPM}}^{m}(\widetilde{z})}{Q_0^{\proddist m}(\widetilde{z})}}.\label{eq:divergence_polar_code}
  \end{align}
  We have 
  \begin{align}
  	Q_{\text{PPM}}^m(\widetilde{z}) &= \frac{1}{m}\sum_{i=1}^{m}Q_1(z_i)\prod_{j=1,j\neq i}^{m}Q_0(z_j)\\
  	&\geq \frac{1}{m}\mu_1\mu_0^{m-1},
  \end{align}
  where $\mu_0=\min_{z\in\text{supp}(Q_0)} Q_0(z)$ and $\mu_1 = \min_{z\in\text{supp}(Q_1)} Q_1(z)$.
  Hence,
  \begin{align}
  	\log\left(\frac{1}{\min_{\widetilde{z}} Q_{\text{PPM}}^{m}(\widetilde{z})}\right) &\leq (m-1)\log\frac{1}{\mu_0}+\log\frac{m}{\mu_1}.
  \end{align}
  We also have
  \begin{align}
  	\frac{Q_{\text{PPM}}^{m}(\widetilde{z})}{Q_0^{\proddist m}(\widetilde{z})} &= \frac{1}{m}\sum_{i=1}^{m}\frac{Q_1(z_i)}{Q_0(z_i)}.
  \end{align}
  Hence,
  \begin{align}
  	\frac{\mu_1}{m} \leq \frac{Q_{\text{PPM}}^{m}(\widetilde{z})}{Q_0^{\proddist m}(\widetilde{z})} \leq \frac{1}{\mu_0},
  \end{align}
  and
  \begin{align}
 	\abs{\log\frac{Q_{\text{PPM}}^{m}(\widetilde{z})}{Q_0^{\proddist m}(\widetilde{z})}} \leq \max\left(\log\frac{m}{\mu_1},\log\frac{1}{\mu_0}\right).
  \end{align}
  For large enough $m$, we have
  \begin{align}
  \abs{\log\frac{Q_{\text{PPM}}^{m}(\widetilde{z})}{Q_0^{\proddist m}(\widetilde{z})}} \leq \log\frac{m}{\mu_1}.
  \end{align}
  From \eqref{eq:divergence_PPM}, we obtain
  \begin{align}
  	\D{(Q_{\text{PPM}}^{m})^{\proddist \ell} }{Q_0^\pn} \leq \delta + \calO(\frac{1}{m}).
  \end{align}
  Therefore, \eqref{eq:divergence_polar_code} becomes
  \begin{align}
  	\D{P_{\widetilde{\mathbf{Z}}}}{Q_0^\pn} &\leq  \delta + \calO(\frac{1}{m}) +\calO(q\ell^2\sqrt{\ell2^{-\ell^{\beta}}}),
  \end{align}
  which is less than $\delta+\zeta$ for large enough $\ell$.
  
  Finally, the number of transmitted bits is
  \begin{align}
    \ell R_i^Y
    & \geq \ell \sum_{i=1}^q \left({I(X_i;\widetilde{Y}|X^{i+1:q}) - \frac{A}{\ell^{\epsilon \kappa}}}\right)\\
    &=\ell \left(I(\widetilde{X};\widetilde{Y}) - \calO\left(\frac{q}{\ell^{\epsilon \kappa}}\right)\right).
  \end{align}
  The covert throughput is given by
  \begin{align}
  \frac{\ell\sum_{i=1}^q R_i^Y}{\sqrt{m\ell\delta}}
  &\geq\sqrt{\frac{\ell}{m\delta}} \left(I(\widetilde{X};\widetilde{Y}) - \calO\left(\frac{q}{\ell^{\epsilon \kappa}}\right)\right)\\
  &=\sqrt{\frac{2}{\chi_2(Q_1||Q_0)}}\left(I(\widetilde{X};\widetilde{Y}) - \calO\left(\frac{q}{\ell^{\epsilon \kappa}}\right)\right)\\
  &=\sqrt{\frac{2}{\chi_2(Q_1||Q_0)}}\left(\D{P_1}{P_0} - \calO\left(\frac{1}{m}\right) - \calO\left(\frac{q}{\ell^{\epsilon \kappa}}\right)\right).
  \end{align}
  From Section~\ref{sec:degraded_case}, we know that for degraded case, $I(X_i;\widetilde{Y}|X^{i+1:q}) \geq I(X_i;\widetilde{Z}|X^{i+1:q})$ for all $i\in\intseq{1}{q}$. Hence, the number of key bits required is 
  \begin{align}
  	\ell(R_i^Z - R_i^Y) \leq \ell \sum_{i=1}^{q}\frac{2A}{\ell^{\epsilon\kappa}} = \ell \calO\left(\frac{q}{\ell^{\epsilon \kappa}}\right),
  \end{align}
  and the key throughput is
  \begin{align}
  	\frac{\ell(R_i^Z - R_i^Y)}{\sqrt{m\ell\delta}} \leq \calO\left(\frac{q}{\ell^{\epsilon \kappa}}\right).
  \end{align}
  Hence, this coding scheme achieves the covert capacity.
\end{proof}

\subsubsection{Coding scheme using invertible extractors}
We can further reduce the complexity of the coding scheme using invertible extractors for channel resolvability. Since the capacity of the higher levels goes to zero exponentially, we can use the first few levels to code for reliability and only use the higher levels to achieve channel resolvability. Let $u$ represent the number of levels used to code for reliability. We can simplify the code construction by constructing a single code for all the higher levels. If we take the levels from $u+1$ to $q$ as a single channel, the corresponding channel is given by
\begin{align}
  W^{u+1:q}\left( \widetilde{z}|x_{u+1:q}\right) &= \frac{1}{2^u} \sum_{k\in\calA^q(x_{u+1:q})} Q_0^{\proddist 2^q}(\widetilde{z}) \frac{Q_1(z_k)}{Q_0(z_k)}.
\end{align}
This channel is symmetric in the sense that for all $x_{u+1:q}$ and $x_{u+1:q}'$, there exists a permutation of the components of $\widetilde{z}$ denoted by $\pi_{x_{u+1:q}+x_{u+1:q}'}$ such that
\begin{align}
	W^{u+1:q}(\pi_{x_{u+1:q}+x_{u+1:q}'}(\widetilde{z})|x_{u+1:q}) = W^{u+1:q}(\widetilde{z}|x_{u+1:q}').
\end{align}
For $\ell$ independent uses of this channel with input vectors $\mathbf{x}_{u+1:q}$ and $\mathbf{x}_{u+1:q}'$, we represent the component-wise permutation of $\widetilde{\mathbf{z}}$ by $\pi^\ell_{\mathbf{x}_{u+1:q}+\mathbf{x}_{u+1:q}'}$.

We follow a construction of codes using invertible extractors similar to the ones used in~\cite{Bellare2012,Chou2014}. Let Ext be a two-universal extractor defined as
$$\text{Ext}: \mathbb{S} \times \mathbb{F}_2^{(q-u)\ell} \to \mathbb{F}_2^{(q-u)\ell-k}: (s,x) \mapsto b,$$
and Inv be the inverter of Ext defined as
$$\text{Inv}: \mathbb{S} \times \mathbb{F}_2^{(q-u)\ell - k} \times \mathbb{F}_2^{k} \rightarrow \mathbb{F}_2^{(q-u)\ell} : ({s,b,r}) \mapsto {\bf x}_{u+1:q}.$$
Let $\mathcal{P}_{s,b} \eqdef \{x\in\mathbb{F}_2^{(q-u)\ell}: \text{Ext}(s,x) = b\}$. We assume that Ext is regular, that is, $\{\mathcal{P}_{s,b}\}_{b\in\mathbb{F}_2^{(q-u)\ell-k}}$ forms a partition of $\mathbb{F}_2^{(q-u)\ell}$ into bins of equal size.
For a given $s\in\mathbb{S}$ and $b\in\mathbb{F}_2^{(q-u)\ell - k}$, the encoder $\phi$ is given by\\
$$ \phi:\mathbb{F}_2^{k} \rightarrow \mathbb{F}_2^{(q-u)\ell}  : r \mapsto \text{Inv}(s,b,r).$$

Let $P_{\widetilde{\mathbf{Z}}}$ be the distribution induced by the coding scheme given by
\begin{align}
P_{\widetilde{\mathbf{Z}}}(\widetilde{\mathbf{z}}) &= \sum_{\mathbf{x}_{u+1:q} \in \mathcal{P}_{s,b}} \frac{1}{|\mathcal{P}_{s,b}|} {W^{(u+1:q)}}^{\proddist\ell}(\widetilde{\mathbf{z}}|\mathbf{x}_{u+1:q})\\
&= \sum_{\mathbf{x}_{u+1:q}\in \mathcal{P}_{s,b}} \frac{1}{2^{k}} {W^{(u+1:q)}}^{\proddist\ell}(\widetilde{\mathbf{z}}|\mathbf{x}_{u+1:q}),
\end{align}
and $Q_{\widetilde{\mathbf{Z}}}$ be the distribution induced by an uniformly distributed input given by
\begin{align}
Q_{\widetilde{\mathbf{Z}}}(\widetilde{\mathbf{z}}) &= \sum_{\mathbf{x}_{u+1:q}} \frac{1}{2^{(q-u)\ell}} {W^{(u+1:q)}}^{\proddist\ell}(\widetilde{\mathbf{z}}|\mathbf{x}_{u+1:q}).
\end{align}
\begin{lemma}
  The encoder $\phi$ defined above with rate $R_{u+1:q} = \frac{k}{\ell}$ with $s$ and $b$ selected randomly according to uniform distributions $q_S$ and $q_B$, respectively, satisfies
  \begin{align}
    \lim_{\ell\rightarrow\infty} \mathbb{E}_{S,B}\left[\avgD{P_{\widetilde{\mathbf{Z}}}}{Q_{\widetilde{\mathbf{Z}}}}\right] = 0,\label{eq:divergence_extractors}
  \end{align}
  if $R_{u+1:q} > I(X_{u+1:q};\widetilde{Z})+2\epsilon H(\widetilde{Z})$.
\end{lemma}

\begin{proof}
  We define a typical set as follows:
  \begin{align}
  \mathcal{T}_{\epsilon}^\ell(X_{u+1:q},\widetilde{Z}) \eqdef \left\{ \left(\mathbf{x}_{u+1},\dots,\mathbf{x}_q,\mathbf{\widetilde{z}} \right)\in \mathcal{X}^\ell \times \cdots \times \mathcal{X}^\ell \times \widetilde{\mathcal{Z}}^\ell : \frac{1}{\ell}\log\frac{{W^{(u+1:q)}}^{\proddist\ell}(\widetilde{\mathbf{z}}|\mathbf{x}_{u+1:q})}{Q_{\widetilde{\mathbf{Z}}}(\widetilde{\mathbf{z}})} < I(X_{u+1:q};\widetilde{Z})+\epsilon\right\}.
  \end{align}
  We have
  \begin{align}
    &\mathbb{E}_{S,B}\left[\avgD{P_{\widetilde{\mathbf{Z}}}}{Q_{\widetilde{\mathbf{Z}}}}\right] = \sum_{s,b} q_S(s) q_B(b) \sum_{\widetilde{\mathbf{z}}} \sum_{\mathbf{x}_{u+1:q}\in \mathcal{P}_{s,b}} \frac{1}{2^{k}} {W^{(u+1:q)}}^{\proddist\ell}(\widetilde{\mathbf{z}}|\mathbf{x}_{u+1:q}) \nonumber \\ &\qquad\qquad\qquad\qquad\qquad\qquad\qquad\qquad\qquad\qquad\qquad\qquad\qquad \log\left[\frac{\sum_{\mathbf{\bar{x}}_{u+1:q}\in \mathcal{P}_{s,b}} {W^{(u+1:q)}}^{\proddist\ell}(\widetilde{\mathbf{z}}|\mathbf{\bar{x}}_{u+1:q})}{2^k Q_{\widetilde{\mathbf{Z}}}(\widetilde{\mathbf{z}}) }\right]\displaybreak[0]\\
    &\qquad\qquad\qquad\overset{(a)}{=}\sum_{s,b,\widetilde{\mathbf{z}},\mathbf{x}_{u+1:q}} q_S(s) \frac{1}{2^{(q-u)\ell}} \mathds{1}\left\{{\text{Ext}(s,\mathbf{x}_{u+1:q}) = b}\right\}  {W^{(u+1:q)}}^{\proddist\ell}(\widetilde{\mathbf{z}}|\mathbf{x}_{u+1:q}) \nonumber\\ &\qquad\qquad\qquad\qquad\qquad\qquad\qquad\qquad\qquad \log\left[\frac{\sum_{\mathbf{\bar{x}}_{u+1:q}} \mathds{1}\left\{{\text{Ext}(s,\mathbf{\bar{x}}_{u+1:q}) = b}\right\}  {W^{(u+1:q)}}^{\proddist\ell}(\widetilde{\mathbf{z}}|\mathbf{\bar{x}}_{u+1:q})}{2^k Q_{\widetilde{\mathbf{Z}}}(\widetilde{\mathbf{z}}) }\right]\displaybreak[0]\\
    &\qquad\qquad\qquad\overset{(b)}{\leq}\sum_{\widetilde{\mathbf{z}},\mathbf{x}_{u+1:q}}  \frac{1}{2^{(q-u)\ell}}   {W^{(u+1:q)}}^{\proddist\ell}(\widetilde{\mathbf{z}}|\mathbf{x}_{u+1:q}) \nonumber\\ &\qquad\qquad\qquad\qquad\qquad \log\left[\frac{\sum_{s,\mathbf{\bar{x}}_{u+1:q}} q_S(s) \mathds{1}\left\{{\text{Ext}(s,\mathbf{\bar{x}}_{u+1:q}) = \text{Ext}(s,\mathbf{x}_{u+1:q}) }\right\}  {W^{(u+1:q)}}^{\proddist\ell}(\widetilde{\mathbf{z}}|\mathbf{\bar{x}}_{u+1:q})}{2^k Q_{\widetilde{\mathbf{Z}}}(\widetilde{\mathbf{z}}) }\right]\displaybreak[0]\\
    &\qquad\qquad\qquad\overset{(c)}{\leq} \sum_{\widetilde{\mathbf{z}},\mathbf{x}_{u+1:q}}  \frac{1}{2^{(q-u)\ell}}   {W^{(u+1:q)}}^{\proddist\ell}(\widetilde{\mathbf{z}}|\mathbf{x}_{u+1:q}) \nonumber\\ &\qquad\qquad\qquad\qquad\qquad \log\left[\sum_{\mathbf{\bar{x}}_{u+1:q}\neq \mathbf{x}_{u+1:q}} {W^{(u+1:q)}}^{\proddist\ell}(\widetilde{\mathbf{z}}|\mathbf{\bar{x}}_{u+1:q}) \frac{2^{-((q-u)\ell-k)}}{2^k Q_{\widetilde{\mathbf{Z}}}(\widetilde{\mathbf{z}}) } + \frac{{W^{(u+1:q)}}^{\proddist\ell}(\widetilde{\mathbf{z}}|\mathbf{{x}}_{u+1:q})}{2^k Q_{\widetilde{\mathbf{Z}}}(\widetilde{\mathbf{z}}) }\right]\displaybreak[0]\\
    &\qquad\qquad\qquad\overset{(d)}{\leq} \sum_{\widetilde{\mathbf{z}},\mathbf{x}_{u+1:q}}  \frac{1}{2^{(q-u)\ell}}   {W^{(u+1:q)}}^{\proddist\ell}(\widetilde{\mathbf{z}}|\mathbf{x}_{u+1:q}) \log\left[1+ \frac{{W^{(u+1:q)}}^{\proddist\ell}(\widetilde{\mathbf{z}}|\mathbf{{x}}_{u+1:q})}{2^k Q_{\widetilde{\mathbf{Z}}}(\widetilde{\mathbf{z}}) }\right]\displaybreak[0]\\
    &\qquad\qquad\qquad\overset{(e)}{\leq} \sum_{(\widetilde{\mathbf{z}},\mathbf{x}_{u+1:q})\in \mathcal{T}_\epsilon^\ell(X_{u+1:q},\widetilde{Z})}  \frac{1}{2^{(q-u)\ell}}   {W^{(u+1:q)}}^{\proddist\ell}(\widetilde{\mathbf{z}}|\mathbf{x}_{u+1:q}) \log\left[1+ 2^{-\ell (R_{u+1:q}-I(X_{u+1:q};\widetilde{Z})-\epsilon)}\right]\nonumber\\ &\qquad\qquad\qquad\qquad\qquad\qquad\qquad\qquad+\sum_{(\widetilde{\mathbf{z}},\mathbf{x}_{u+1:q})\notin \mathcal{T}_\epsilon^\ell(X_{u+1:q},\widetilde{Z})}  \frac{1}{2^{(q-u)\ell}}   {W^{(u+1:q)}}^{\proddist\ell}(\widetilde{\mathbf{z}}|\mathbf{x}_{u+1:q}) \log\left[1+ \mu_z^{-\ell}\right]\\
    &\qquad\qquad\qquad\leq 2^{-\ell (R_{u+1:q}-I(X_{u+1:q};\widetilde{Z})-\epsilon)} + \frac{\ell}{\mu_z} \P{(\mathbf{X}_{u+1:q},\widetilde{\mathbf{Z}}) \notin \mathcal{T}_\epsilon^\ell(X_{u+1:q},\widetilde{Z})} \\
    &\qquad\qquad\qquad\leq 2^{-\ell (R_{u+1:q}-I(X_{u+1:q};\widetilde{Z})-\epsilon)} + \frac{\ell}{\mu_z} e^{-\frac{2\ell\epsilon^2}{q^2}} \overset{\ell \to \infty}{\longrightarrow} 0,
  \end{align}
  where (a) follows from $q_B(b) = \frac{1}{2^{(q-u)\ell-k}}$ and from the definition of $\mathcal{P}_{s,b}$, (b) follows by Jensen's inequality, (c) follows because Ext is two-universal and hence, for all $\mathbf{\bar{x}}_{u+1:q} \neq \mathbf{{x}}_{u+1:q} $, we have $\P[S]{{\text{Ext}(S,\mathbf{\bar{x}}_{u+1:q}) = \text{Ext}(S,\mathbf{x}_{u+1:q}) }} = \sum_{s} q_S(s) \mathds{1}\left\{{\text{Ext}(s,\mathbf{\bar{x}}_{u+1:q}) = \text{Ext}(s,\mathbf{x}_{u+1:q}) }\right\} \leq 2^{-((q-u)\ell-k)}$, (d) follows because from the definition of $Q_{\widetilde{\mathbf{Z}}}(\widetilde{\mathbf{z}})$, we have $\sum_{\mathbf{\bar{x}}_{u+1:q}\neq \mathbf{x}_{u+1:q}}\frac{1}{2^{(q-u)\ell}} {W^{(u+1:q)}}^{\proddist\ell}(\widetilde{\mathbf{z}}|\mathbf{\bar{x}}_{u+1:q})\leq Q_{\widetilde{\mathbf{Z}}}(\widetilde{\mathbf{z}})$ and (e) follows from bounding the term inside $\log$ from the definition of the typical set and using the definition $\mu_z = \min_{z\in\text{supp}(Q_Z)}Q_Z(z)$.
\end{proof}

We now employ a specific extractor. For $s\in\mathbb{S}=\mathbb{F}_2^{(q-u)\ell}\backslash\{\mathbf{0}\}$, define
\begin{align}
\text{Ext}: \mathbb{S} \times \mathbb{F}_2^{(q-u)\ell} \to \mathbb{F}_2^{(q-u)\ell-k}:(s,x)&\mapsto b \eqdef (s^{-1}\odot x)|_{\intseq{1}{(q-u)\ell-k}},
\end{align}
where $\odot$ is the multiplication in the field $\mathbb{F}_2^{(q-u)\ell}$ and $(\cdot)|_{\intseq{1}{(q-u)\ell-k}}$ represents the bits in the positions $\intseq{1}{(q-u)\ell}$. Ext is a two-universal hash function~\cite{Bellare2012}, whose inverter is given by
\begin{align}
	\text{Inv}: \mathbb{S}\times\mathbb{F}_2^{(q-u)\ell-k}\times\mathbb{F}_2^k\to\mathbb{F}_2^{(q-u)\ell}:(s,b,r)\mapsto s\odot(b\Vert r),
\end{align}
where $(\cdot\Vert\cdot)$ denotes the concatenation of two sequences of bits. We now show that for the encoder implemented using this inverter, the divergence is same for any $b$. We can express $Q_{\widetilde{\mathbf{Z}}}$ as follows:
\begin{align}
	Q_{\widetilde{\mathbf{Z}}}(\widetilde{\mathbf{z}}) &= \sum_{\mathbf{x}_{u+1:q}} \frac{1}{2^{(q-u)\ell}} {W^{(u+1:q)}}^{\proddist\ell}(\widetilde{\mathbf{z}}|\mathbf{x}_{u+1:q})\\
	&=\sum_{b\in\mathbb{F}_2^{(q-u)\ell-k}} \sum_{r\in\mathbb{F}_2^{k}} \frac{1}{2^{(q-u)\ell}} {W^{(u+1:q)}}^{\proddist\ell}(\widetilde{\mathbf{z}}|s\odot(b||r))\\
	&=\sum_{b\in\mathbb{F}_2^{(q-u)\ell-k}} \sum_{r\in\mathbb{F}_2^{k}} \frac{1}{2^{(q-u)\ell}} {W^{(u+1:q)}}^{\proddist\ell}(\pi^\ell_{s\odot(b'||0)}(\widetilde{\mathbf{z}})|s\odot(b\oplus{b'}||r))\\
	&=Q_{\widetilde{\mathbf{Z}}}(\pi^\ell_{s\odot(b'||0)}(\widetilde{\mathbf{z}})).
\end{align}
Taking expectation of the divergence between $P_{\widetilde{\mathbf{Z}}}$ and $Q_{\widetilde{\mathbf{Z}}}$ over $S$ for a fixed $b$, we have
\begin{align}
	\mathbb{E}_{S,B=b}\left[\avgD{P_{\widetilde{\mathbf{Z}}}}{Q_{\widetilde{\mathbf{Z}}}}\right] &= \sum_{s} q_S(s) \sum_{\widetilde{\mathbf{z}}} \sum_{\mathbf{x}_{u+1:q}\in \mathcal{P}_{s,b}} \frac{1}{2^{k}} {W^{(u+1:q)}}^{\proddist\ell}(\widetilde{\mathbf{z}}|\mathbf{x}_{u+1:q}) \nonumber \\ &\qquad\qquad\qquad\qquad\qquad\qquad \log\left[\frac{\sum_{\mathbf{\bar{x}}_{u+1:q}\in \mathcal{P}_{s,b}} {W^{(u+1:q)}}^{\proddist\ell}(\widetilde{\mathbf{z}}|\mathbf{\bar{x}}_{u+1:q})}{2^k Q_{\widetilde{\mathbf{Z}}}(\widetilde{\mathbf{z}}) }\right]\displaybreak[0]\\
	&= \sum_{s} q_S(s) \sum_{\widetilde{\mathbf{z}}} \sum_{r\in\mathbb{F}_2^k} \frac{1}{2^{k}} {W^{(u+1:q)}}^{\proddist\ell}(\widetilde{\mathbf{z}}|s\odot(b||r)) \nonumber \\ &\qquad\qquad\qquad\qquad\qquad\qquad \log\left[\frac{\sum_{r'\in\mathbb{F}_2^k} {W^{(u+1:q)}}^{\proddist\ell}(\widetilde{\mathbf{z}}|s\odot(b||r'))}{2^k Q_{\widetilde{\mathbf{Z}}}(\widetilde{\mathbf{z}}) }\right]\displaybreak[0]\\
	&=\sum_{s} q_S(s) \sum_{\widetilde{\mathbf{z}}} \sum_{r\in\mathbb{F}_2^k} \frac{1}{2^{k}} {W^{(u+1:q)}}^{\proddist\ell}(\pi^\ell_{s\odot(b\oplus b'||0)}(\widetilde{\mathbf{z}})|s\odot(b'||r)) \nonumber \\ &\qquad\qquad\qquad\qquad\qquad\qquad \log\left[\frac{\sum_{r'\in\mathbb{F}_2^k} {W^{(u+1:q)}}^{\proddist\ell}(\pi^\ell_{s\odot(b\oplus{b'}||0)}(\widetilde{\mathbf{z}})|s\odot(b'||r'))}{2^k Q_{\widetilde{\mathbf{Z}}}(\pi^\ell_{s\odot(b\oplus{b'}||0)}(\widetilde{\mathbf{z}})) }\right]\displaybreak[0]\\
	&=\sum_{s} q_S(s) \sum_{\widetilde{\mathbf{z}}'} \sum_{r\in\mathbb{F}_2^k} \frac{1}{2^{k}} {W^{(u+1:q)}}^{\proddist\ell}(\widetilde{\mathbf{z}}'|s\odot(b'||r)) \nonumber \\ &\qquad\qquad\qquad\qquad\qquad\qquad \log\left[\frac{\sum_{r'\in\mathbb{F}_2^k} {W^{(u+1:q)}}^{\proddist\ell}(\widetilde{\mathbf{z}}'|s\odot(b'||r'))}{2^k Q_{\widetilde{\mathbf{Z}}}(\widetilde{\mathbf{z}}') }\right]\\
	&=\mathbb{E}_{S,B=b'}\left[\avgD{P_{\widetilde{\mathbf{Z}}}}{Q_{\widetilde{\mathbf{Z}}}}\right].
\end{align}
Hence, from \eqref{eq:divergence_extractors}
\begin{align}
	\lim_{\ell\rightarrow\infty}\mathbb{E}_{S,B=b}\left[\avgD{P_{\widetilde{\mathbf{Z}}}}{Q_{\widetilde{\mathbf{Z}}}}\right] &=\lim_{\ell\rightarrow\infty}\mathbb{E}_{S,B}\left[\avgD{P_{\widetilde{\mathbf{Z}}}}{Q_{\widetilde{\mathbf{Z}}}}\right] = 0.
\end{align}
From the above, we conclude that the choice of $b$ is irrelevant for achieving channel resolvability. Hence, we can choose $b=0$. Finally, note that we do not require $S$ to be private, so we can use publicly available common randomness as the source for $S$, and we can reduce the rate of $S$ by using a chaining strategy similar to the one in~\cite{Chou2014}.

\section{Conclusion}
\label{sec:conclusion}
Coding for covert communication over a \ac{DMC} requires biasing the input distribution so that codewords of length $n$ contain only on the order of $\sqrt{n}$ non-innocent symbols. This requirement is challenging to achieve using known codes; in fact, we have shown that it is impossible to achieve covertness using linear codes. We have overcome this challenge using \ac{PPM} by transforming the original channel into a \ac{PPM} super channel over which a uniform input distribution of \ac{PPM} symbols achieves the covert capacity. Direct coding over this PPM super channel would require coding over non-binary alphabets by mapping each PPM symbol to a non-binary alphabet; moreover, the order of the PPM symbols and the alphabet size would scale with the blocklength. \ac{MLC} solves this problem by transforming coding over the PPM super channel into coding over several binary-input channels. Although the number of levels scales with the blocklength, we showed that most of the capacity concentrates in the first few levels and the equivalent channel corresponding to each level is a stationary memoryless channel.  Further, we have shown that one can \emph{independently} design codes for each level by constructing codes for the equivalent channels, which in turn achieve overall reliability and channel resolvability under mild assumptions regarding some symmetry conditions on the decoder. In particular, one may use polar codes because the successive cancellation decoder for polar codes satisfies the symmetry requirement. Finally, since the first few levels concentrate most of the capacity, one may code a limited number of levels for reliability and only code the remaining levels for resolvability using a negligible number of key bits. We have also shown how to code for channel resolvability on all of those remaining levels at once using invertible extractors, further reducing the complexity of the design.

While the discussion has focused on binary-input \acp{DMC}, the proposed coding schemes achieves reliability and covertness over any \ac{DMC}, but without reaching the covert capacity in general. Extending the coding scheme to achieve the covert capacity is certainly possible, for instance by allowing more than a single non-innocent symbol in the \ac{PPM} scheme.

\bibliographystyle{IEEEtran}
\bibliography{covert-codes}

\end{document}